\documentclass[11pt]{article}

\usepackage[margin=1in]{geometry}

\usepackage{amssymb,amsmath,amsthm,amsfonts}
\usepackage{bm}
\usepackage{textgreek}
\usepackage{mathtools}
\usepackage{enumitem}
\usepackage[numbers,comma,sort&compress]{natbib}
\usepackage{authblk}
\usepackage{graphicx}
\usepackage[font=small]{caption}
\usepackage[labelformat=simple]{subcaption}
\usepackage{float}
\usepackage[ruled,vlined,linesnumbered]{algorithm2e}
\usepackage{algorithmic}

\usepackage{physics}
\usepackage{footnote}
\usepackage{xcolor}
\usepackage{mathrsfs}
\usepackage{bbm}
\usepackage{makecell}

\usepackage[colorlinks]{hyperref}
\newtheorem{theorem}{Theorem}
\newtheorem{definition}{Definition}
\newtheorem{lemma}{Lemma}
\newtheorem{proposition}{Proposition}

\newtheorem{corollary}{Corollary}
\newtheorem{remark}{Remark}
\newtheorem{assumption}{Assumption}
\newtheorem{problem}{Problem}

\newcommand{\eqn}[1]{(\ref{eqn:#1})}
\newcommand{\eq}[1]{(\ref{eq:#1})}

\newcommand{\thm}[1]{\hyperref[thm:#1]{Theorem~\ref*{thm:#1}}}
\newcommand{\cor}[1]{\hyperref[cor:#1]{Corollary~\ref*{cor:#1}}}
\newcommand{\defn}[1]{\hyperref[defn:#1]{Definition~\ref*{defn:#1}}}
\newcommand{\lem}[1]{\hyperref[lem:#1]{Lemma~\ref*{lem:#1}}}
\newcommand{\prop}[1]{\hyperref[prop:#1]{Proposition~\ref*{prop:#1}}}
\newcommand{\assume}[1]{\hyperref[assume:#1]{Assumption~\ref*{assume:#1}}}
\newcommand{\fig}[1]{\hyperref[fig:#1]{Figure~\ref*{fig:#1}}}
\newcommand{\tab}[1]{\hyperref[tab:#1]{Table~\ref*{tab:#1}}}
\newcommand{\algo}[1]{\hyperref[algo:#1]{Algorithm~\ref*{algo:#1}}}
\newcommand{\prob}[1]{\hyperref[prob:#1]{Problem~\ref*{prob:#1}}}
\renewcommand{\sec}[1]{\hyperref[sec:#1]{Section~\ref*{sec:#1}}}
\newcommand{\append}[1]{\hyperref[append:#1]{Appendix~\ref*{append:#1}}}
\newcommand{\fac}[1]{\hyperref[fac:#1]{Fact~\ref*{fac:#1}}}
\newcommand{\lin}[1]{\hyperref[lin:#1]{Line~\ref*{lin:#1}}}

\def\>{\rangle}
\def\<{\langle}

\newcommand{\specialcell}[2][c]{%
  \begin{tabular}[#1]{@{}c@{}}#2\end{tabular}}

\newcommand{\vect}[1]{\ensuremath{\mathbf{#1}}}
\newcommand{\x}{\ensuremath{\mathbf{x}}}

\newcommand{\cN}{\mathcal{N}}

\newcommand{\R}{\mathbb{R}}

\newcommand{\E}{\mathbb{E}}

\renewcommand{\H}{\mathcal{H}}

\newcommand{\ut}{\mathscr{T}}

\newcommand{\uf}{\mathscr{F}}

\newcommand{\tO}{\tilde{O}}
\newcommand{\tnu}{\tilde{\nu}}
\newcommand{\tTheta}{\tilde{\Theta}}
\newcommand{\tOmega}{\tilde{\Omega}}
\newcommand{\tx}{\tilde{\x}}

\newcommand{\tnabla}{\tilde{\nabla}}

\newcommand{\he}{\hat{\vect{e}}}
\newcommand{\hnu}{\hat{\nu}}

\newcommand{\tp}{t_\text{perturb}}

\DeclareMathOperator{\poly}{poly}
\DeclareMathOperator{\spn}{span}

\DeclareMathOperator{\supp}{supp}
\DeclareMathOperator{\quan}{quantum}
\DeclareMathOperator{\prog}{prog}

\renewcommand{\d}{\mathrm{d}}
\renewcommand{\x}{\vect{x}}
\newcommand{\y}{\vect{y}}
\newcommand{\z}{\vect{z}}
\newcommand{\g}{\vect{g}}
\renewcommand{\c}{\vect{c}}
\renewcommand{\u}{\vect{u}}
\renewcommand{\v}{\vect{v}}

\newcommand{\0}{\mathbf{0}}

\def\Tr{\operatorname{Tr}}\def\:{\hbox{\bf:}}

\let\oldnl\nl
\newcommand{\nonl}{\renewcommand{\nl}{\let\nl\oldnl}}

\begin{document}

\title{Robustness of Quantum Algorithms for Nonconvex Optimization}
\author{Weiyuan Gong\thanks{Equal contribution.}         \thanks{\mbox{Institute for Interdisciplinary Information Sciences, Tsinghua University. Email: \href{mailto:wygong19@mails.tsinghua.edu.cn}{gongwy19@mails.tsinghua.edu.cn}}}
\qquad
Chenyi Zhang$^{*}$\thanks{Computer Science Department, Stanford University; Institute for Interdisciplinary Information Sciences, Tsinghua University. Email: \href{mailto:chenyiz@stanford.edu}{chenyiz@stanford.edu}}
\qquad
Tongyang Li\thanks{Corresponding author. Center on Frontiers of Computing Studies and School of Computer Science, Peking University.  Email: \href{mailto:tongyangli@pku.edu.cn}{tongyangli@pku.edu.cn}}}
\date{}

\maketitle

%%%%%%%%%%%%%%%%%%%%%%%%%%%%%%%%%%%%%%%%%%%%%%%%%%%%%%%%%%%%%%%%%%%%%%%%%%%%%%
\begin{abstract}
Recent results suggest that quantum computers possess the potential to speed up nonconvex optimization problems. However, a crucial factor for the implementation of quantum optimization algorithms is their robustness against experimental and statistical noises. In this paper, we systematically study quantum algorithms for finding an $\epsilon$-approximate second-order stationary point ($\epsilon$-SOSP) of a $d$-dimensional nonconvex function, a fundamental problem in nonconvex optimization, with noisy zeroth- or first-order oracles as inputs. We first prove that, up to noise of $O(\epsilon^{10}/d^5)$, accelerated perturbed gradient descent equipped with quantum gradient estimation takes $O(\log d/\epsilon^{1.75})$ quantum queries to find an $\epsilon$-SOSP. We then prove that standard perturbed gradient descent is robust to the noise of $O(\epsilon^6/d^4)$ and $O(\epsilon/d^{0.5+\zeta})$ for $\zeta>0$ on the zeroth- and first-order oracles, respectively, which provides a quantum algorithm with poly-logarithmic query complexity. We then propose a stochastic gradient descent algorithm using quantum mean estimation on the Gaussian smoothing of noisy oracles, which is robust to $O(\epsilon^{1.5}/d)$ and $O(\epsilon/\sqrt{d})$ noise on the zeroth- and first-order oracles, respectively. The quantum algorithm takes $O(d^{2.5}/\epsilon^{3.5})$ and $O(d^2/\epsilon^3)$ queries to the two oracles, giving a polynomial speedup over the classical counterparts. As a complement, we characterize the domains where quantum algorithms can find an $\epsilon$-SOSP with poly-logarithmic, polynomial, or exponential number of queries in $d$, or the problem is information-theoretically unsolvable even with an infinite number of queries. In addition, we prove an $\Omega(\epsilon^{-12/7})$ lower bound on $\epsilon$ for any randomized classical and quantum algorithm to find an $\epsilon$-SOSP using either noisy zeroth- or first-order oracles. As far as we know, this is the first quantum lower bound in $\epsilon$ for finding $\epsilon$-approximate second-order stationary points in nonconvex optimization.
\end{abstract}

\section{Introduction}

Optimization theory is a central topic in computer science and applied mathematics, with wide applications in machine learning, operations research, statistics, and many other areas.
Currently, various quantum algorithms for optimization have been proposed, ranging from linear programs~\cite{casares2020quantum,rains1999monotonicity} and semidefinite programs~\cite{brandao2017quantum,brandao2017SDP,van2019improvements,van2020quantum} to general convex optimization~\cite{chakrabarti2020optimization,van2020convex} and nonconvex optimization~\cite{liu2022quantum,zhang2021quantum}.

A crucial factor of quantum optimization algorithms is their robustness. On the one hand, current quantum applications suffer from noises generated by near-term quantum devices~\cite{Preskill2018NISQ}, which may create adversarial perturbations in the worst-case that result in disastrous failures. To deal with this issue, some (elements of)  quantum algorithms, such as some adiabatic quantum algorithms~\cite{childs2001robustness}, quantum gates~\cite{harrow2003robustness}, and machine learning algorithms~\cite{liu2021rigorous,cross2015quantum,lu2020quantum}, are robustness against experimental noises or noisy quantum queries~\cite{buhrman2007robust}. An alternative solution is to develop error correction~\cite{gottesman1997stabilizer} or error mitigation~\cite{endo2018practical,endo2021hybrid} mechanisms to reduce the influences of experimental noises. In the context of nonconvex optimization, developing robust quantum algorithms is essential for future practical implementations of these algorithms on near-term devices.

On the other hand, robustness is a natural and crucial requirement for solving classical optimization problems. For instance, statistical machine learning, which is a widely explored task, concerns the problem with data generated from an underlying probability distribution $\mathcal{D}$ (i.e., population), and optimizes the objective function (i.e., population risk) $F$ defined by the expectation:
\begin{align}
F(\theta)=\E_{\vect{z}\sim\mathcal{D}}[L(\theta;\vect{z})],
\end{align}
where the expectation is averaged over all possible continuous loss functions $\{L(\cdot\,;\vect{z})\}$ with $\z\sim\mathcal{D}$. The optimization algorithm does not access $F$ directly but can take queries to the empirical risk function $f(\theta)=\sum_{i=1}^nL(\theta;\vect{z}_i)/n$ via querying the loss function values at $L$ sample points. The optimization of $F$ given access to the empirical risk $f$ is known as the \emph{empirical risk minimization}~\cite{belloni2015escaping,jin2018local,vapnik1991principles}.
Under this setting, the noisy evaluation of $F$ can be poorly behaved -- it might have exponentially many shallow local minima even if $F$ has a good landscape and satisfies smoothness or Lipschitz assumptions~\cite{auer1995exponentially,brutzkus2017globally}.

To analyze the problem of nonconvex optimization using noisy queries, previous literature on classical optimization (see e.g.~\cite{bartlett2002rademacher,boucheron2013concentration}) assumed that $f$ and $F$ are pointwise close to each other:
\begin{align}\label{eq:ZeroIntro}
\|F-f\|_{\infty}\leq\nu,
\end{align}
where the error $\nu$ usually decays with the number of samples. Under this assumption, $f$ may still be non-smooth and contain additional shallow local minima independent of $F$. Nevertheless, it is possible to exploit the pointwise closeness between $f$ and $F$ to escape from highly suboptimal local minima that only exist in $f$ and find an approximate local minimum of $F$.

Another nonconvex optimization model using noisy queries is to find local minima of $F$ with empirical first-order information \cite{jin2018local}. Similar to \eq{ZeroIntro}, we query a stochastic gradient $\nabla f$ uniformly close to the actual gradient $\nabla F$. This model is widely considered in stochastic scenarios where we evaluate the gradient information using a sampling procedure to zeroth-order function values. A well-known example is the \emph{stochastic gradient descent}~\cite{jin2021nonconvex,sun2019optimization}, 
where we obtain an approximated gradient value by sampling mini-batch function values. As the mini-batch size $m$ increases, the gradient evaluation converges to the actual gradient with high probability:
\begin{align}\label{eq:FirstIntro}
\norm{\nabla F-\nabla f}_\infty\leq\tnu,
\end{align}
where the error $\tnu$ typically decreases with the mini-batch size $m$. Here, $\norm{\nabla F-\nabla f}_\infty$ take the maximal value of the infinity-norm taken both over the input $\x$ and the $d$ different entries of the gradient at the $\x$.

Various approaches have been developed to investigate the robustness of optimization algorithms from different perspectives~\cite{belloni2015escaping,zhang2017hitting,jin2018local,risteski2016algorithms,singer2015information,karabag2021smooth,roy2020escaping,zhang2022zeroth}. In the context of convex optimization, Ref.~\cite{belloni2015escaping} proposed an algorithm for finding an $\epsilon$-approximate global minimum of an approximate convex function, where $\epsilon$ is the precision guarantee for the optimization output (see \assume{ZeroProb} and \assume{FirstProb} for the formal definition). This algorithm requires $\tilde{O}(d^{7.5}/\epsilon^2)$\footnote{The $\tilde{O}$ notation omits poly-logarithmic terms, i.e., $\tilde{O}(g)=O(g\poly(\log g))$.} queries to the stochastic noisy function evaluation oracle, which has zero-mean and sub-Gaussian distributed noise. Very recently, Ref.~\cite{Li2022Quantum} improved this result by proposing a quantum algorithm with query complexity $\tilde{O}(d^{5}/\epsilon)$ for the same task, giving a polynomial quantum speedup compared to the classical counterpart. In addition, Ref.~\cite{singer2015information} proposed an information-theoretic lower bound for any convex optimization algorithms to find minima within $\epsilon$ multiplicative error using noisy function evaluation oracles. In Ref.~\cite{risteski2016algorithms}, an algorithm with optimal dependence on $d$ was proposed to find an $\epsilon$-approximate minimum taking queries to noisy function evaluation oracles. 

In the context of nonconvex optimization, Refs.~\cite{chen2020stationary,zhang2017hitting} considered querying oracles with bounded noise $\nu\leq O(\epsilon^2/d^8)$, where $\epsilon$ is the precision and $d$ is the dimension of $F$. This work developed an efficient classical algorithm to escape from the noise-induced ``shallow" local minima using simulated annealing and stochastic gradient Langevin dynamics (SGLD). More recently, improved polynomial algorithms to solve the nonconvex optimization problem with bounded noise of $O(\epsilon^{1.5}/d)$ and $O(\epsilon/\sqrt{d})$ using zeroth- and first- order noisy oracles were obtained in Ref.~\cite{jin2018local}. Ref.~\cite{jin2018local} further discussed the noise threshold to guarantee the existence of polynomial-query classical algorithms for finding an $\epsilon$-approximate local minimum of $F$. In contrast, little has been known about the influence of noises on quantum algorithms for general nonconvex optimization problems, upon which this paper systematically investigates. 

%===================================================================================================================

\subsection{Nonconvex Optimization with Noisy Oracle}

We consider the nonconvex optimization problem with a twice-differentiable target function $F\colon\R^d\to\R$ satisfying
\begin{itemize}
    \item $F$ is $B$-bounded: $\sup_{\x\in\R^d}\abs{F(\x)}\leq B$;
    \item $F$ is $\ell$-smooth ($\ell$-gradient Lipschitz): $\norm{\nabla F(\x_1)-\nabla F(\x_2)}\leq\ell\norm{\x_1-\x_2},\quad\forall\x_1,\x_2\in\R^d$;
    \item $F$ is $\rho$-Hessian Lipshitz: $\norm{\nabla^2 F(\x_1)-\nabla^2 F(\x_2)}\leq\rho\norm{\x_1-\x_2},\quad\forall\x_1,\x_2\in\R^d$.
\end{itemize}

The goal is to find an $\epsilon$-approximate second order stationary point ($\epsilon$-SOSP)\footnote{A more general target is to find an $(\epsilon,\gamma)$-SOSP $\x$ such that $\norm{F(\x)}\leq\epsilon$ and $\lambda_{\min}(\nabla^2 F(\x))\geq-\gamma$. The definition of an $\epsilon$-SOSP in \eq{SOSPDef} was proposed first by Ref.~\cite{nesterov2006cubic} and has been taken as a standard assumption in the subsequent papers~\cite{jin2017escape,jin2018accelerated,xu2017neon,xu2018first,carmon2018accelerated,agarwal2017finding,tripuraneni2018stochastic,fang2019sharp,jin2021nonconvex,zhang2021quantum}.} such that 
\begin{align}\label{eq:SOSPDef}
\norm{F(\x)}\leq\epsilon,\qquad\lambda_{\min}(\nabla^2 F(\x))\geq-\sqrt{\rho\epsilon}.
\end{align}
Instead of directly querying $F$, we assume one can access a noisy function $f$ that is pointwise close to $F$. 
\begin{assumption}[Noisy evaluation query]\label{assume:ZeroProb}
The target function $F$ is $B$-bounded, $\ell$-smooth, and $\rho$-Hessian Lipschitz, and we can query a noisy function $f$ that is $\nu$-pointwise close to $F$:
\begin{align}
\norm{F-f}_\infty\leq\nu.
\end{align}
\end{assumption}
We further consider finding an $\epsilon$-SOSP of $F$ given an alternative condition that the gradient $\nabla f$ of function $f$ is pointwise close to $\nabla F$.
\begin{assumption}[Noisy gradient query]\label{assume:FirstProb}
The target function $F$ is $B$-bounded, $\ell$-smooth, and $\rho$-Hessian Lipschitz, and we can query the gradient $\g\coloneqq\nabla f$ of an $L$-smooth function $f$. The gradient $\g$ is pointwise close to gradient of $F$:
\begin{align}
\norm{\nabla F-\nabla f}_\infty\leq\tnu.
\end{align}
\end{assumption}

In the quantum context, the oracles are unitary operators rather than classical procedures. Under \assume{ZeroProb}, one can query a \textit{quantum evaluation oracle} (quantum zeroth-order oracle) $U_f$, which can be represented as
\begin{align}\label{eq:QZeroOracle}
U_f(\ket{\x}\otimes\ket{0})\to\ket{\x}\otimes\ket{f(\x)},\qquad\forall \x\in\R^d.
\end{align}
Furthermore, quantum oracles allow coherent \textit{superpositions} of queries. Given $m$ vectors $\ket{\x_1},\ldots,$ $\ket{\x_m}\in\R^d$ and a coefficient vector $\c\in\mathbb{C}^m$ such that $\sum_{i=1}^m\abs{\c_i}^2=1$, the quantum oracle outputs $U_f(\sum_{i=1}^m\c_i\ket{\x_i}\otimes\ket{0})\to\sum_{i=1}^m\c_i\ket{\x_i}\otimes\ket{f(\x_i)}$. Compared to the classical evaluation oracle, the ability to query different locations simultaneously in superposition is the essence of quantum speedup. In addition, if a classical oracle can be implemented by a classical circuit, the corresponding quantum oracle can be implemented by a quantum circuit of the same size.

Similarly, in the first-order scenario we assume that one can access the \textit{quantum gradient oracle} $U_\g$ under \assume{FirstProb},
which can be represented as
\begin{align}\label{eq:QFirstOracle}
U_\g(\ket{\x}\otimes\ket{\0})\to\ket{\x}\otimes\ket{\nabla f(\x)},\qquad\forall \x\in\R^d.
\end{align}

\subsection{Contributions}
In this paper, we conduct a systematic study of quantum algorithms for nonconvex optimization using noisy oracles. Using zeroth- or first-order oracles as inputs, we rigorously characterize different domains where quantum algorithms can find an $\epsilon$-SOSP using poly-logarithmic, polynomial, or exponential number of queries, respectively. We also identify the domain where it is information-theoretically unsolvable to find an $\epsilon$-SOSP even using an infinite number of queries.

In some of the domains, we further develop lower bounds on the query complexity for any classical algorithms and propose quantum algorithms with polynomial or exponential speedups compared to either the classical lower bounds or the complexities of corresponding state-of-the-art classical algorithms. We summarize our main results under \assume{ZeroProb} and \assume{FirstProb} in \tab{main1} and \tab{main2}, respectively. 

\begin{table}[ht]
\centering
\resizebox{1.0\columnwidth}{!}{
\begin{tabular}{llll}
\hline
Noise Strength & Classical Bounds & Quantum Bounds & Speedup in $d$ \\ \hline
$\nu=\Omega(\epsilon^{1.5})$ & Unsolvable~\cite{jin2018local} & Unsolvable~(\thm{ZeroNoLower}) & N/A \\\hline
$\nu=O(\epsilon^{1.5})$, $\nu=\tOmega(\epsilon^{1.5}/d)$ & $O(\exp(d))$, $\Omega(d^{\log d})$~\cite{jin2018local} & $\Omega(d^{\log d})$~(\thm{ZeroMeanLower}) & N/A \\\hline
$\nu=O(\epsilon^{1.5}/d)$, $\nu=\tOmega(\epsilon^6/d^4)$ & $\tO(d^4/\epsilon^5)$~\cite{jin2018local,zhang2017hitting} & $\tO(d^{2.5}/\epsilon^{3.5})$~(\thm{ZeroMean}) & Polynomial\\ \hline
$\nu=\tO(\epsilon^6/d^4)$, $\nu=\tOmega(\epsilon^{10}/d^5)$ & $\Omega(d/\log d)$ (\thm{ZeroClassLowerJordan}) & $O(\log^4 d/\epsilon^2)$ (\thm{ZeroJordan}) & Exponential\\
\hline
$\nu=\tO(\epsilon^{10}/d^5)$ & $\Omega(d/\log d)^{*}$ (\thm{ZeroClassLowerJordan}) & $O(\log d/\epsilon^{1.75})$ (\thm{TinyAGD}) & Exponential$^*$\\
\hline
\end{tabular}
}
\caption{A summary of our results and comparisons with the state-of-the-art classical upper and lower bounds under \assume{ZeroProb}. The query complexities are highlighted in terms of the dimension $d$ and the precision $\epsilon$. ($*$) In the last row, we can obtain the desired classical lower bound and thus an exponential speedup in the query complexity when $\nu=\tOmega(\poly(1/d,\epsilon))$ as \thm{ZeroClassLowerJordan} works for $\nu=\tOmega(\poly(1/d,\epsilon))$. }
\label{tab:main1}
\end{table}

\begin{table}[ht]\small
\centering
\resizebox{1.0\columnwidth}{!}{
\begin{tabular}{llll}
\hline
Noise Strength & Classical Bounds & Quantum Bounds & Speedup in $d$ \\ \hline
$\tnu=\Omega(\epsilon)$ & Unsolvable~(\thm{FirstLower}) & Unsolvable~(\thm{FirstLower}) & N/A \\ \hline
$\tnu=O(\epsilon)$, $\tnu=\tOmega(\epsilon/d^{0.5})$ & $\Omega(d^{\log d})$~(\thm{FirstLower}) & $\Omega(d^{\log d})$~(\thm{FirstLower}) & N/A \\ \hline
$\tnu=\Theta(\epsilon/d^{0.5})$ & $O(d^3/\epsilon^4)$~\cite{jin2018local} & $O(d^2/\epsilon^3)$~(\thm{FirstMean}) & Polynomial \\ \hline
$\tnu=O(\epsilon/d^{0.5+\zeta})$ & $O(\log^4 d/\epsilon^2)$ (\cor{FirstPGD}) & $O(\log^4 d/\epsilon^2)$ (\cor{FirstPGD}) & No
\\ \hline
\end{tabular}
}
\caption{A summary of our results and comparisons with the state-of-the-art classical upper and lower bounds under \assume{FirstProb}. In the last line, $\zeta>0$ and $\zeta=\Omega(1/\log(d))$ (for instance, this is satisfied for any constant $\zeta>0$).}
\label{tab:main2}
\end{table}

\subsubsection{Upper bounds}
\paragraph{Tiny noise: robustness of perturbed accelerated gradient descent.} We start by adding tiny noise to the oracles in quantum gradient descent algorithms. In particular, we consider the function pair $(F,f)$ satisfying \assume{ZeroProb} and assume that one can access the function values of the noisy evaluation function $f$. We remark that $f$ may even be non-differentiable or non-smooth. In addition, the noise between $f$ and the target function $F$ might introduce additional SOSPs. Nevertheless, recent work~\cite{anonymous2023faster} indicates that the performance of accelerated gradient descent algorithm (PAGD)~\cite{jin2018accelerated,zhang2021escape} persists when the gradients are inexact. We rigorously prove that the perturbed accelerated gradient descent algorithm with accelerated negative curvature~\cite{zhang2021escape} equipped with Jordan's algorithm for quantum gradient estimation~\cite{jordan2005fast} is robust to the tiny noise on zeroth-order oracles. We formulate our first main result as follow:

\begin{theorem}[Informal]\label{thm:TinyAGD}
Given a target function $F$ and a noisy function $f$ satisfying \assume{ZeroProb} with $\nu=\Omega(\epsilon^{10}/d^5)$, there exists a quantum algorithm that finds an $\epsilon$-SOSP of $F$ with high probability using $\tO(\log d/\epsilon^{1.75})$ queries to the noisy zeroth-order oracle $U_f$.
\end{theorem}

We leave the formal version of \thm{TinyAGD}, the corresponding algorithm, and the proof to \sec{RobustAGD}. \thm{TinyAGD} demonstrates that if the noise is small enough, the impact on PAGD algorithm will not lead to an increase on the query complexity. If $\nu=\Omega(\poly(\epsilon,1/d))$, we further demonstrate that this robustness only exists for quantum algorithms by proving a polynomial lower bound in \thm{ZeroClassLowerJordan} for any classical algorithm.

\paragraph{Small noise: robustness of quantum gradient estimation.} When the strength of noise increases, the negative curvature estimation in standard PAGD will fail. In this case, we show the robustness of the gradient descent algorithm with quantum gradient estimation against the noise. We consider the function pair $(F,f)$ satisfying \assume{ZeroProb} when we can access noisy function $f$. Refs.~\cite{zhang2021quantum,chakrabarti2020optimization} conveyed the conceptual message that perturbed gradient descent (PGD)~\cite{jin2021nonconvex} algorithm with Jordan's gradient estimation~\cite{jordan2005fast} possesses a certain degree of robustness to noise. In this work, we formalize this intuition and obtain the following result:
\begin{theorem}[Informal]\label{thm:ZeroJordan}
Given a target function $F$ and a noisy function $f$ satisfying \assume{ZeroProb} with $\nu\leq \tO(\epsilon^6/d^4)$, there exists a quantum algorithm that finds an $\epsilon$-SOSP of $F$ with high probability using $\tO(\log^4 d/\epsilon^2)$ queries to the noisy zeroth-order oracle $U_f$.
\end{theorem}

The formal version of \thm{ZeroJordan}, the corresponding algorithms, and the proof are given in \sec{RobustPGD}. \thm{ZeroJordan} demonstrates if the noise on the zeroth-order oracle is below a certain threshold, a quantum algorithm can find an $\epsilon$-SOSP of $F$ within a number of queries that is poly-logarithmic in terms of the dimension $d$. Similar to \thm{TinyAGD}, this robustness only exists in quantum algorithms and provide an exponential quantum speedup in the query complexity compared to the classical counterpart.

We further extend \thm{ZeroJordan} to function pair $(F,f)$ satisfying \assume{FirstProb}. We prove in \cor{FirstPGD} that the classical PGD iteration is robust against the noise of $\tnu\leq O(\epsilon/d^{0.5+\zeta})$ on the first-order gradient information, where $\zeta=\Omega(1/\log(d))$.

\paragraph{Intermediate noise: speedup from quantum mean estimation.} When the strength of noise keeps increasing, the robustness of Jordan's algorithm will also fail to handle the gap between the noisy function $f$ and the target function $F$. To address this issue, we develop a quantum algorithm based on the Gaussian smoothing of $f$ inspired by Ref.~\cite{jin2018local}. We consider function pairs $(F,f)$ satisfying \assume{ZeroProb}. We sample the value $\z[f(\x+\z)-f(\x)]/\sigma^2$, where $\z\sim\cN(0,\sigma^2I)$ is chosen from Gaussian distribution with parameter $\sigma^2$~\cite{duchi2015optimal}. We then apply quantum mean estimation to approximate the gradient from the samples of stochastic gradients. The performance of the algorithm is given by the following theorem:
\begin{theorem}[informal]\label{thm:ZeroMean}
Given a target function $F$ and a noisy function  $f$ satisfying \assume{ZeroProb} with $\nu\leq O(\epsilon^{1.5}/d)$, there exists a quantum algorithm that finds an $\epsilon$-SOSP of $F$ with high probability taking $\tO(d^{2.5}/\epsilon^{3.5})$ queries to the noisy zeroth-order oracle $U_f$.
\end{theorem}

\thm{ZeroMean} indicates that the quantum algorithm can find an $\epsilon$-SOSP of $F$ using polynomial number of queries to $f$ with bounded strength of noise $\nu\leq O(\epsilon^{1.5}/d)$. Recall that the state-of-art classical algorithm~\cite{jin2018local} solves this problem with the same noise strength $\nu\leq O(\epsilon^{1.5}/d)$ using $O(d^4/\epsilon^5)$ queries, our algorithm provides a polynomial improvement compared to the best known classical result in terms of both the dimension $d$ and the precision $\epsilon$. 

In \sec{FirstGaussMean}, we consider the problem of finding an $\epsilon$-SOSP of functions $F$ taking queries to the quantum gradient oracle in \eq{QFirstOracle}. We use $\nabla f(\x+\z)-\nabla f(\x)$ as a stochastic gradient estimation, where $\z\sim\cN(0,\sigma^2I)$ is chosen from Gaussian distribution with parameter $\sigma^2$. Similar to the zeroth-order scenario, we implement quantum mean estimation and 
derive the following algorithmic upper bound.
\begin{theorem}[Informal]\label{thm:FirstMean}
Given a target function $F$ and the gradient information of a noisy function $f$ satisfying \assume{FirstProb} with $\tnu\leq O(\epsilon/d^{0.5})$, there exists a quantum algorithm that finds an $\epsilon$-SOSP of $F$ with high probability using $\tO(d^2/\epsilon^3)$ queries to the noisy first-order gradient oracle $U_\g$.
\end{theorem}
The tolerance on $\tnu$ and the query complexity is larger compared to \thm{ZeroMean}, where we access a zeroth-order oracle. The best-known classical algorithm finding an $\epsilon$-SOSP under \assume{FirstProb} requires $O(d^3/\epsilon^4)$ queries. Hence, this quantum algorithm also provides a polynomial reduction on the sample complexity compared to the classical result.

\subsubsection{Lower bounds}

\paragraph{Large noise: quantum query complexity lower bound in $d$.} In this work, we also provide lower bounds concerning $d$ on the query complexity required for any classical and quantum algorithms under \assume{ZeroProb} and \assume{FirstProb}. In particular, we construct a hard instance inspired by Ref.~\cite{jin2018local} (as shown in \fig{MainLowerHard} (a)): we define a target function $F$ in a hypercube and use the hypercube to fill the entire space $\mathbb{R}^d$. By adding noise to the zeroth- or first-order oracle $f$, we can erase the information of $F$ such that a limited number of classical or quantum queries cannot find any $\epsilon$-SOSPs with high probability.

For a function pair $(F,f)$ satisfying \assume{ZeroProb}, our first result in this part is the following quasi-polynomial lower bound.
\begin{theorem}[Informal]\label{thm:ZeroMeanLower}
We can find functions $F$ and $f$ satisfying \assume{ZeroProb} with $\nu=\tTheta(\epsilon^{1.5}/d)$ such that any \textit{quantum} algorithm requires at least $\Omega(d^{\log d})$ queries to $U_f$ to find any $\epsilon$-SOSP of $F$ with high probability.
\end{theorem}

The detailed proof for \thm{ZeroMeanLower} is a bit technically involved and is left to \sec{DLowerQuasiPoly}. Here, we provide the intuition for the proof. To prove this lower bound, we define a function $F$ in a hyperball $\mathbb{B}(\0,r)$ and embed the hyperball into a hypercube, with which we can cover the whole space. Next, we introduce noise to create $f$ with a non-informative area around $\0$ (in the sense that any query to this area will obtain no information about any SOSPs of the target function $F$). Then, we transfer this problem into an unstructured search problem. The final lower bound for nonconvex optimization is obtained by applying the quantum lower bound for unstructured search. We mention that the $\epsilon$ and the $d$ dependence for $\nu$ in \thm{ZeroMeanLower} are tight up to logarithmic factors. The classical version of \thm{ZeroMeanLower} is proved in Ref.~\cite{jin2018local}. The parallelism in quantum algorithms possesses the potential to query different points in superposition. However, \thm{ZeroMeanLower} demonstrates that the same query complexity lower bound holds even for quantum algorithms. 

If the noise $\nu$ keeps increasing, we can further prove the following lower bound in \sec{DLowerZeroNo} that prevents any quantum algorithm from finding any $\epsilon$-SOSPs of target function $F$:
\begin{theorem}[Informal]\label{thm:ZeroNoLower}
For any \textit{quantum} algorithm, there exists a pair of functions $(F,f)$ satisfying \assume{ZeroProb} with $\nu=\tTheta(\epsilon^{1.5})$ such that it will fail, with large probability, to find any $\epsilon$-SOSP of $F$ given access to $f$.
\end{theorem}

Despite the quantum lower bound, we also propose a classical lower bound concerning nonconvex optimization using zeroth-order oracle with noise strength $\nu=O(1/\poly(d))$.
\begin{theorem}\label{thm:ZeroClassLowerJordan}
For any $\epsilon\leq \epsilon_0<1$, where $\epsilon_0$ is some constant, there exists a function pair $(F,f)$ satisfying \assume{ZeroProb} with $\nu= \Omega(1/\poly(d))$, such that any \textit{classical} algorithm that outputs an $\epsilon$-SOSP of $F$ with high probability requires at least $\Omega(d/\log d)$ classical queries to the noisy function $f$.
\end{theorem}

We prove \thm{ZeroClassLowerJordan} using an information-theoretic argument inspired by Ref.~\cite{chakrabarti2020optimization}. \thm{TinyAGD}, \thm{ZeroJordan}, and \thm{ZeroClassLowerJordan} establish the exponential separation between classical and quantum query complexities required for nonconvex optimization using oracles with noise $\nu=\tOmega(\poly(\epsilon,1/d))$. This separation originates from the Jordan's gradient estimation algorithm~\cite{jordan2005fast}. Classically, querying the evaluation oracle can only provide information at one point. Quantumly, however, one can take the superposition on different points and query the quantum evaluation oracle in parallel~\cite{gilyen2019optimizing,chakrabarti2020optimization}.

Moreover, we extend the above lower bound to function pairs $(F,f)$ satisfying \assume{FirstProb} in \sec{DLowerFirst}. If the noise increases by even a factor that is logarithmic in $d$ from $\tTheta(\epsilon/d^{0.5})$, we can prove an exponential lower bound for any classical or quantum algorithm through a similar construction of hard instance used in \thm{ZeroMeanLower} (as shown in \fig{MainLowerHard} (b)). Moreover, if the noise increases to $\Omega(\epsilon)$, there exists a similar hard instance with \thm{ZeroNoLower} that prevents any classical or quantum algorithm from finding any $\epsilon$-SOSP of $F$. Formally, we can extend \thm{ZeroMeanLower} and \thm{ZeroNoLower} in the context of \assume{FirstProb}:
\begin{theorem}[Informal]\label{thm:FirstLower}
We can find functions $F$ and $f$ satisfying \assume{FirstProb} with $\tnu=\tTheta(\epsilon/d^{0.5})$ such that any classical or quantum algorithm that finds an $\epsilon$-SOSP of $F$ with high probability requires at least $\Omega(d^{\log d})$ queries to $U_\g$. Moreover, for any classical or quantum algorithm, we can find functions $F$ and $f$ satisfying \assume{FirstProb} with $\tnu=\Theta(\epsilon)$ such that it will fail with high probability.
\end{theorem}

\paragraph{Quantum query complexity lower bound in $\epsilon$.} Finally, we establish query complexity lower bounds for classical and quantum nonconvex optimization algorithms under \assume{ZeroProb} or \assume{FirstProb}, respectively, where our results are summarized in \tab{main3}.
\begin{table}[ht]
\resizebox{1.0\columnwidth}{!}{
\small
\begin{tabular}{llll}
\hline
Input Oracle\ \ & Noise Strength\ \ & \specialcell{Deterministic Classical\\Lower Bounds}\ \ & \specialcell{Randomized Classical and\\Quantum Lower Bounds}\ \  \\ \hline
Zeroth-order & $\nu=0$ & N/A & N/A\\
\hline
Zeroth-order & $\nu=\Omega(\epsilon^{-16/7}/d)$ & $\Omega(d/\log d)^{*}$ (\thm{ZeroClassLowerJordan}) & $\Omega(\epsilon^{-12/7})$ (\thm{QuantumCarmon})\\
\hline
First-order & $\tnu=0$ & $\Omega(\epsilon^{-12/7})$~\cite{carmon2021lower} &
N/A \\
\hline
First-order & $\tnu=\Omega(\epsilon^{-8/7}/\sqrt{d})$ & $\Omega(\epsilon^{-12/7})$~\cite{carmon2021lower} & $\Omega(\epsilon^{-12/7})$ (\thm{QuantumCarmon})\\
\hline
\end{tabular}
}
\caption{A summary of our results on classical and quantum query complexity lower bounds in $\epsilon$ under \assume{ZeroProb} or \assume{FirstProb}, respectively. The query complexities are highlighted in terms of the dimension $d$ and the precision $\epsilon$. 
($*$) As \thm{ZeroClassLowerJordan} works for noise strength of $\tOmega(\poly(\epsilon,1/d))$, this lower bound holds for both  deterministic and randomized classical algorithms.}
\label{tab:main3}
\end{table}
\begin{theorem}[informal]\label{thm:QuantumCarmon}
There exists a function pair $F$ and $f$ satisfying either \assume{ZeroProb} or \assume{FirstProb} with $\nu=\Omega(\epsilon^{-16/7}/d)$ or $\tnu=\Omega(\epsilon^{-8/7}/\sqrt{d})$, respectively, and additionally $F(\0)-\inf_{\x}F(\x)\leq\Delta$ for some constant $\Delta$, such that any classical or quantum algorithm with query complexity $\Omega\big(\epsilon^{-12/7}\big)$ will fail with high probability to find an $\epsilon$-SOSP of target function $F$.
\end{theorem}

We provide the detailed proof for \thm{QuantumCarmon} in \sec{EpsLower} using the hard instance inspired by Refs.~\cite{carmon2020lower,carmon2021lower}. Previously, there have been two lower bounds concerning $\epsilon$ dependence that apply to classical algorithms for nonconvex optimization. In Ref.~\cite{carmon2020lower}, it is proved that at least $\Omega(\epsilon^{-3/2})$ queries are required in finding an $\epsilon$-SOSP of a Hessian Lipshitz function $F$ even provided both zeroth- and first-order oracles for either random or deterministic classical algorithms. 
Using similar techniques, Carmon et al.~\cite{carmon2021lower} further proved that deterministic classical algorithms using first-order \textit{noiseless} oracle require $\Omega(\epsilon^{-12/7})$ queries to find an $\epsilon$-SOSP of a Hessian Lipshitz function $F$. 

On the other hand, despite recent papers~\cite{garg2020no,garg2021near} studying quantum lower bounds on convex optimization, quantum lower bounds on nonconvex optimization are still widely open. In this paper, we fill this conceptual gap by extending the classical deterministic lower bound~\cite{carmon2021lower} to all classical randomized algorithms and even quantum algorithms, given that noise exists in the function evaluation. In particular, noise allows us to construct a hard instance by creating a non-informative area around $\0$. According to the concentration of measure phenomenon, the non-informative area will occupy an overwhelming proportion of the whole space. Although its intuition and structure are different from the hard instance in Refs.~\cite{garg2020no,garg2021near} constructed via performing maximization, the hard instance we construct here exhibits a similar property that, if the number of quantum queries is below a certain threshold, in expectation the output state will barely change if we replace the quantum oracle by an oracle that only encodes ``partial" information of the objective function, where the missing information is crucial for any (classical or quantum) algorithm to find an $\epsilon$-SOSP of $F$. 

Moreover, we note that our lower bound result in \thm{QuantumCarmon} can be extended to the case where the goal is merely to find an $\epsilon$-SOSP if we waive the $B$-bounded requirement on $F$, which may be of independent interest.

\subsection{Open Questions}
Our paper leaves several open questions for future investigations:
\begin{itemize}
    \item Can we give quantum algorithms for the task of nonconvex optimization with better performance using noisy oracles? For instance, can we obtain a quantum algorithm with better dependence on $d$ and $\epsilon$ compared to \thm{ZeroMean}?
    
    \item Can we derive tighter lower bounds on quantum algorithms for nonconvex optimization? In particular, it is natural to investigate sublinear or poly-logarithmic quantum lower bounds in dimension $d$ on general optimization problems using either noiseless or noisy oracles.
    
    \item In this work, we employ a simple model on the noise in \assume{ZeroProb} and \assume{FirstProb}: only the upper bound of noise strength is considered. In general, can we demonstrate the robustness and speedups for nonconvex optimization algorithms analytically under other noise assumptions (say, more practical quantum noise models or stochastic noise models), or experimentally by numerical simulations or on real-world quantum computers?
\end{itemize}

\subsection{Organization}
The rest of the paper is organized as follows:
\begin{itemize}
    \item In \sec{RobustNoise}, we prove the robustness for the standard gradient-based algorithms. In particular, we consider the tiny noise case and prove the robustness of PAGD equipped with quantum gradient estimation in \sec{RobustAGD}. In \sec{RobustPGD}, we consider the small noise case and prove the robustness of standard PGD equipped with quantum gradient estimation. 
    \item In \sec{GaussMean}, we consider the intermediate noise case and propose the stochastic gradient descent algorithm using Gaussian smoothing and quantum mean estimation, which provides a polynomial speedup compared to classical algorithms under \sec{ZeroGaussMean} and \sec{FirstGaussMean}, respectively. 
    \item In \sec{DLower}, we prove lower bounds concerning dimension $d$ for classical and quantum algorithms under different noise strengths. Specifically, in \sec{DLowerQuasiPoly} and \sec{DLowerZeroNo}, we prove the existence of hard instances under \assume{ZeroProb} for any (polynomial) quantum algorithm when $\nu=\tTheta(\epsilon^{1.5}/d)$ ($\nu=\tTheta(\epsilon^{1.5})$). In \sec{DLowerClassical}, we prove the $\Omega(d/\log d)$ classical query complexity lower bound using zeroth-order oracles with $\nu=\Omega(1/\poly(d))$. We prove the lower bound under \assume{FirstProb} in \sec{DLowerFirst}.
    \item In \sec{EpsLower}, we prove lower bounds concerning the precision $\epsilon$ for both (possibly randomized) classical algorithms and quantum algorithms.
    \item In the appendices, we introduce necessary existing tools for our proofs in \append{existing-lemmas}. Technical lemmas for the main text are given in \append{technical-lemmas}. Additional information and extended discussions on PGD equipped with quantum simulation and quantum tunneling walk are provided in \append{PGDQSGC} and \append{QTW}, respectively.
\end{itemize}

%%%%%%%%%%%%%%%%%%%%%%%%%%%%%%%%%%%%%%%%%%%%%%%%%%%%%%%%%%%%%%%%%%%%%%%%%%%%%%%%

\section{Robustness of Quantum and Classical Algorithms with Small Noise}\label{sec:RobustNoise}
In this section, we propose two quantum nonconvex optimization algorithms that are robust for tiny noise and small noise, respectively. These algorithms find to an $\epsilon$-SOSP of $F$ using only polylogarithmic queries to noisy empirical function $f$.

\subsection{Robustness of Classical Perturbed Accelerated Gradient Descent with tiny noise}\label{sec:RobustAGD}
To begin with, we introduce the quantum perturbed accelerated gradient descent (PAGD) with accelerated negative curvature finding algorithm, which is inspired by the noiseless nonconvex optimization algorithm in Ref.~\cite{zhang2021escape}. To find an $\epsilon$-SOSP of $F$ using quantum evaluation oracle specified in \eq{QZeroOracle}, an important step is to approximate the gradient at each iteration. An ingenious quantum approach initiated by Ref.~\cite{jordan2005fast} takes a uniform mesh around the point and queries the quantum evaluation oracle (in uniform superposition) in phase using the standard phase kickback technique~\cite{chakrabarti2020optimization,gilyen2019optimizing}. Then by the Taylor expansion, we have
\begin{align}
\sum_{\x}\exp(if(\x))\x\approx\sum_{\x}\bigotimes_{k=1}^d\exp(i\frac{\partial f}{\partial \x_k}\x_k)\x_k.
\end{align}
The algorithm finally recovers all the partial derivatives by applying a quantum Fourier transformation (QFT). We refer to Ref.~\cite{chakrabarti2020optimization} for a precise version of Jordan's gradient estimation algorithm with the following performance guarantee:
\begin{lemma}[Lemma 2.2, Ref.~\cite{chakrabarti2020optimization}]\label{lem:JordanPerf}
Given a target function $F$ and its noisy evaluation $f$ satisfying \assume{ZeroProb} with noisy rate $\nu$, there exists a quantum algorithm that uses one query to the noisy oracle defined in \eq{QZeroOracle} and outputs a vector $\tnabla F(\x)$ such that
\begin{align}
\Pr\left[\norm{\tnabla F(\x)-\nabla F(\x)}\geq 400\omega d\sqrt{\nu\ell}\right]\leq\min\left\{\frac{d}{\omega-1},1\right\},\qquad\forall\omega>1.
\end{align}
\end{lemma}

This lemma indicates that with probability at least $1-\delta$, one can use one query to the noisy zeroth-oracle and obtain a vector $\tnabla F(\x)$ such that
\begin{align}
\norm{\tnabla F(\x)-\nabla F(\x)}\leq O( d^2\sqrt{\nu\ell}/\delta).
\end{align}

Now, we are ready to introduce our first algorithm as shown in \algo{PAGDANCFQGC}. This algorithm replaces the gradient queries in Perturbed Accelerated Gradient Descent ~\cite{zhang2021escape,anonymous2023faster} with Jordan's gradient estimation in \lem{JordanPerf}. The negative curvature exploitation (NCE) subroutine as shown in \algo{NCE} is applied if the following condition holds.
\begin{align}\label{eq:NCECond}
f(\x_{t})\leq f(\y_{t})+\expval{\tnabla F(\y_{t}),\x_{t}-\y_{t}}-\frac\gamma 2\norm{\x_{t}-\y_{t}}^2.
\end{align}
The intuition for NCE (\algo{NCE}) will be discussed later.

We prove that \algo{PAGDANCFQGC} has the following performance guarantee:
\begin{theorem}[Formal version of \thm{TinyAGD}]\label{thm:TinyAGDF}
Consider a target function $F$ and its noisy evaluation $f$ satisfying \assume{ZeroProb} with $\nu\leq \tO(\delta^2\epsilon^{10}/d^5)$. \algo{PAGDANCFQGC} can find an $\epsilon$-SOSP of $F$ satisfying Eq.~\eq{SOSPDef} with probability at least $1-\delta$, using
\begin{align}
\tilde{O}\left(\frac{\ell B}{\epsilon^{1.75}}\cdot\log d\right)
\end{align}
queries to $U_f$ defined in \eq{QZeroOracle}, under the following parameter choices:
\begin{align}\label{eq:TinyAGDParam}
&\eta=\frac{1}{4\ell},\quad\theta=\frac{1}{4\sqrt{\kappa}},\quad\gamma=\frac{\theta^2}{\eta},\quad s=\frac{\gamma}{4\rho},\quad\delta_0=\frac{\delta\epsilon^{1.75}}{c_\delta\ell B}\cdot\log d,\\
&r=\frac{\delta_0\epsilon}{c_r}\sqrt{\frac{\pi}{\rho d}},\quad\ut=c_r\sqrt{\kappa}\log\left(\frac{\ell\sqrt{d}}{\delta_0\sqrt{\rho\epsilon}}\right),\quad\uf=\sqrt{\frac{\epsilon^3}{\rho}}c^{-7},
\end{align}
where $c$, $c_r$, and $c_\delta$ are some large enough constants, and $\kappa=\ell/\sqrt{\rho\epsilon}$.
\end{theorem}

\begin{algorithm}[H]
\caption{Perturbed Accelerated Gradient Descent with Accelerated Negative Curvature Finding and Quantum Gradient Computation}
\label{algo:PAGDANCFQGC}
\begin{algorithmic}[1]
\REQUIRE $\x_0$, learning rate $\eta$, noise ratio $r$, parameters $\ut$, $\iota$, $\theta$, $\gamma$, and $s$ to be fixed later.
\STATE $\tp\leftarrow-\ut-1$, $\y_0\leftarrow\x_0$, $\tx\leftarrow\x_0$, and $\iota\leftarrow\0$
\FOR{$t=0,1,\ldots,T$}
\STATE Apply \lem{JordanPerf} to compute an estimation $\tnabla F(\x)$ of $\nabla F(\x)$
\IF {$\norm{\tnabla F(\x)}\leq 3\epsilon/4$ and $t-\tp>\ut$}
\STATE $\tx\leftarrow\x_t$
\STATE $\x_t=\tx+\xi_t$
\STATE $\y_t=\x_t$, $\iota=\tnabla F(\tx)$, $\tp\leftarrow t$
\ENDIF
\IF {$\tp\neq-\ut-1$ and $t-\tp=\ut$}
\STATE $\he\leftarrow (\x_t-\tx)/\norm{\x_t-\tx}$
\STATE $\x_t\leftarrow\arg\min_{\x\in\{\tx-\frac14\sqrt{\frac\epsilon\rho}\he,\tx+\frac14\sqrt{\frac\epsilon\rho}\he\}}f(\x)$
\STATE $\y_t=\x_t$, $\iota=0$
\ENDIF
\STATE $\x_{t+1}=\y_t=\eta(\tnabla F(\y_t)-\iota)$
\STATE $\v_{t+1}=\x_{t+1}-\x_t$
\STATE $\y_{t+1}=\x_{t+1}+(1-\theta)\v_{t+1}$
\IF {$\tp\neq-\ut-1$ and $t-\tp\leq\ut$}
\STATE $(\y_{t+1},\x_{t+1})=\tx+r\cdot\left(\frac{\y_{t+1}-\tx}{\norm{\y_{t+1}-\tx}}\cdot\frac{\x_{t+1}-\tx}{\norm{\x_{t+1}-\tx}}\right)$
\ELSIF {$f(\x_{t+1})\leq f(\y_{t+1})+\expval{\tnabla F(\y_{t+1}),\x_{t+1}-\y_{t+1}}-\frac\gamma 2\norm{\x_{t+1}-\y_{t+1}}^2$}
\STATE $(\x_{t+1},\v_{t+1})\leftarrow\text{NCE}(\x_{t+1},\v_{t+1},s)$
\STATE $\y_{t+1}\leftarrow\x_{t+1}+(1-\theta)\v_{t+1}$
\ENDIF
\ENDFOR
\end{algorithmic}
\end{algorithm}
\begin{algorithm}[htbp]
\caption{Negative Curvature Exploitation (NCE) $(\x_t,\v_t,s)$}
\label{algo:NCE}
\begin{algorithmic}[1]
\IF {$\norm{\v_t}\geq s$}
\STATE $\x_{t+1}\leftarrow\x_t$
\ELSE
\STATE $\delta=s\cdot\v_t/\norm{\v_t}$
\STATE $\x_{t+1}\leftarrow\arg\min_{\x\in\{\x_t+\delta,\x_t-\delta\}}f(\x)$
\ENDIF
\end{algorithmic}
\end{algorithm}

For simplicity, we denote the error of Jordan's gradient estimation as $\hnu$. To solve the problem of monotonic decrease for function value in momentum-based nonconvex optimization problems, we consider the Hamiltonian of the function~\cite{jin2018accelerated} in our proof, which is defined as 
\begin{align}\label{eq:Hamiltonian-defn}
E_t=F(\x_t)+\frac{1}{2\eta}\norm{\v_t}^2.
\end{align}
The Hamiltonian composes a potential energy term and a kinetic energy term. It monotonically decreases in the continuous-time scenario. To prove \thm{TinyAGDF}, we consider the dynamics of \algo{PAGDANCFQGC} in the two different cases depending on whether \eq{NCECond} holds. If it does not hold, the following lemma holds by using Lemma 4 of Ref.~\cite{anonymous2023faster} and replacing the zeroth-order queries to $F(\x_t)$ and $F(\y_t)$ with the noisy queries $f(\x_t)$ and $f(\y_t)$.
\begin{lemma}[Adaptive version of Lemma 3, Ref.~\cite{anonymous2023faster}]\label{lem:NCEFreeHam}
We consider $F(\cdot)$ is $\ell$-smooth and $\rho$-Hessian Lipschitz. Assume one can access the zeroth-order oracle with noise $\nu$ and the first-order oracle with noise $\hnu$. Set the learning rate $\eta\leq1/4\ell$, $\theta\in[2\eta\gamma,1/2]$. For each iteration $t$ where \eq{NCECond} does not hold, running \algo{PAGDANCFQGC} will decrease the Hamiltonian defined in \eq{Hamiltonian-defn} by
\begin{align}
E_{t+1}\leq E_t-\frac{\theta}{2\eta}\norm{\v_t}^2-\frac{\eta}{4}\norm{\nabla f(\y_t)}^2+O(\eta\hnu^2)+O(\nu).
\end{align}
\end{lemma}

On the other hand, if \eq{NCECond} holds, the function has an approximate large negative curvature between $\y_t$ and $\x_t$. The accelerated gradient step might not decrease the value for the Hamiltonian. We thus call the negative curvature exploitation subroutine (\algo{NCE}) to further decrease the Hamiltonian. In particular, when choosing large enough constant $c_r$, the following lemme holds by replacing the zeroth-order query to $F(\x_t)$ and $F(\y_t)$ with the noisy query $f(\x_t)$ and $f(\y_t)$ and noise term $O(\nu)$, respectively, in Lemma 4 of Ref.~\cite{anonymous2023faster}.
\begin{lemma}[Adapted version of Lemma 4, Ref.~\cite{anonymous2023faster}]\label{lem:NCEHam}
Assume that $F(\cdot)$ is $\ell$-smooth, $\rho$-Hessian Lipschitz, and we are given the zeroth-order oracle with noise strength $\nu$ and the first-order oracle with noise strength $\hnu$. Set the learning rate $\eta\leq1/4\ell$, $\theta\in[2\eta\gamma,1/2]$. For each iteration $t$ where \eq{NCECond} holds, running \algo{PAGDANCFQGC} wiil decrease the Hamiltonian defined in \eq{Hamiltonian-defn} by
\begin{align}
E_{t+1}\leq E_t-\min\left\{\frac{s^2}{2\eta},\frac12\gamma s^2-\rho s^3-O\left(\frac{\hnu^2}{\gamma}\right)\right\}+O(\nu).
\end{align}
\end{lemma}
We set an additional parameter $\ut'=\Theta(\sqrt{\kappa})$. Based on \lem{NCEFreeHam} and \lem{NCEHam}, and proper choices of $\hnu$ and $\nu$, Lemma 5 of Ref.~\cite{anonymous2023faster} carries over as the below lemma when the norm of the estimated gradient is large enough, i.e. $\big\|\tnabla F(\x_t)\big\|\geq3\epsilon/4$.
\begin{lemma}[Adaptive version of Lemma 5, Ref.~\cite{anonymous2023faster}]\label{lem:LargePAGD}
If $\big\|\tnabla F(\x_t)\big\|\geq3\epsilon/4$ and the noise strengths are bounded by $\nu,\hnu\leq O(\epsilon^{1.25})$ for all $\tau\in[0,\ut']$, \algo{PAGDANCFQGC} can decrease the Hamiltonian by $E_{\ut'}-E_0\leq-\uf$ using
\begin{align}
    \ut'=\sqrt{\kappa}\chi c
\end{align}
iterations in \algo{PAGDANCFQGC}, where $\chi=\max\{1,\log(d\ell B/\rho\epsilon\delta_0)\}$, and $c$ is a large enough constant given in \thm{TinyAGDF}
\end{lemma}

On the other hand, when the estimated gradient is small, we obtain the following adaptive version of Lemma 7 of Ref.~\cite{anonymous2023faster}.
\begin{lemma}[Adaptive version of Lemma 7, Ref.~\cite{anonymous2023faster}]\label{lem:SmallPAGDDir}
Suppose $\big\|\tnabla F(\x_t)\big\|\leq3\epsilon/4$ and the noise strengths are bounded by $\nu,\hnu\leq O(\epsilon^{3.25}/d^{0.5})$, $\lambda_{\min}(\nabla F(\x_t))\leq-\sqrt{\rho\epsilon}$. For any $0\leq\delta_0\leq 1$, we set the parameters as \thm{TinyAGD}. Suppose no perturbation is added in the iterations $[t-\ut,t]$. By running \algo{PAGDANCFQGC} for $\ut$ iterations, we have
\begin{align}
\he^\top\nabla^2F(\x_t)\he\leq-\frac{\sqrt{\rho\epsilon}}{4},
\end{align}
with probability at least $1-\delta_0$.
\end{lemma}

Furthermore, the following lemma from Ref.~\cite{zhang2021escape} indicates that the function value of $F$ will decrease fast along the direction of $\he$.
\begin{lemma}[Lemma 6, Ref.~\cite{zhang2021escape}]\label{lem:SmallPAGD}
Suppose the function $F$ is $\ell$-smooth and $\rho$-Hessian Lipschitz. For any point $\x_t$, if there exists a unit vector $\he$ satisfying $\he^\top F(\x_t)\he\leq-\sqrt{\rho\epsilon}/4$, we have
\begin{align}
F\left(\x_t-\frac{F_{\he}'(\x_t)}{4\abs{F_{\he}'(\x_t)}}\cdot\sqrt{\frac{\epsilon}{\rho}}\right)\leq F(\x_t)-\frac{1}{384}\sqrt{\frac{\epsilon^3}{\rho}},
\end{align}
where $F_{\he}'(\x_t)$ is the entry of the derivative along $\he$. 
\end{lemma}

Now, we are ready to prove \thm{TinyAGDF}.

\begin{proof}
We first set $\nu\leq\frac{C_0}{\ell}\cdot\frac{\epsilon^{10}}{d^5}$ for some small enough constant $C_0$. According to \lem{JordanPerf}, we bound $\hnu\leq O(\epsilon^{3.25}/d^{0.5})$ with probability at least $1-\delta_0$. Assume \algo{PAGDANCFQGC} starts at point $\x_0$ and the local minimum of $F$ has value $F^*$. Since $F$ is $B$-bounded, $F(\x_0)-F^*\leq 2B$. Set the total number of iterations $T$ to be:
\begin{align}
T=3\max\left\{\frac{2B\ut'}{\uf},768B\ut\cdot\sqrt{\frac{\rho}{\epsilon^3}}\right\}.
\end{align}
Suppose for some iterations $\x_t$, we have $\tnabla F(\x_t)\leq3\epsilon/4$ and $\lambda_{\min}(\nabla^2 F(\x_t))\leq-\sqrt{\rho\epsilon}$. The error probability of this assumption is given later. Under this assumption, the function value decreases for $\frac{1}{384}\cdot\sqrt{\frac{\epsilon^3}{\rho}}$ after each $\ut'$ iterations. The number of such iterations when \lem{SmallPAGD} can be called is bounded by $T/3$ times, for otherwise the function value will decrease greater than $2B\geq F(\x_0)-F^*$, which is impossible. The failure probability is composed of two parts: the failure probability of estimating the gradients in \lem{JordanPerf} and the failure probability of \lem{SmallPAGD}. In each iteration, the probability of failure is bounded by $2\delta_0$ according to the union bound. When we choose a large enough constant $\c_\delta$, the overall probability that \algo{PAGDANCFQGC} fails to indicate a negative curvature is upper bounded by
\begin{align}
\frac{T}{3}\cdot 2\delta_0\leq\frac{\delta}{2}.
\end{align}

Excluding the iterations that \lem{SmallPAGD} is applied, there are $2T/3$ iterations left. We consider the iterations $\x_t$ with large gradients, ${\tnabla F(\x_t)}\geq3\epsilon/4$. According to \lem{LargePAGD}, the function value decreases by at least $\uf$ with probability at least $1-\delta_0$ in $\ut'$ iterations. Thus there can be at most $T/3$ steps with large gradients, for otherwise, the function value will decrease greater than $2B\geq F(\x_0)-F^*$, which is impossible. The fail probability is bounded by
\begin{align}
\frac{T}{3}\cdot \delta_0\leq\frac{\delta}{2}.
\end{align}

In summary, we can deduce that with probability at least $1-\delta$, there are at most $T/3$ iterations within which the neighboring $\ut$ iterations have small gradients but large negative curvatures, and at most $T/3$ iterations with large gradients. Therefore, the rest $T/3$ iterations must be $\epsilon$-SOSPs of target function $F$. The number of queries is thus bounded by
\begin{align}
T\leq\tO\left(\frac{B\ell}{\epsilon^{1.75}}\cdot\log d\right).
\end{align}
\end{proof}

\subsection{Robustness of Quantum Perturbed Gradient Descent}\label{sec:RobustPGD}
When the noise rate increases but is still bounded by $\nu\leq\tO(\epsilon^6/d^4)$, some quantum algorithms using perturbed gradient descent (PGD) for noiseless cases are robust against such noise. We introduce the quantum PGD algorithm, which is the one of the standard methods used for noiseless nonconvex optimization~\cite{zhang2021quantum}. 

\algo{PGDQGC} replaces the gradient queries in PGD~\cite{jin2021nonconvex} by Jordan's gradient estimations in \lem{JordanPerf}.
\begin{algorithm}[htbp]
\caption{Perturbed Gradient Descent with Quantum Gradient Computation}
\label{algo:PGDQGC}
\begin{algorithmic}[1]
\REQUIRE $\x_0$, learning rate $\eta$, noise ratio $r$
\FOR{$t=0,1,\ldots,T$}
\STATE Apply \lem{JordanPerf} to compute an estimation $\tnabla F(\x)$ of $\nabla F(\x)$
\STATE $\x_{t+1}\leftarrow \x_t-\eta(\tnabla F(\x)+\xi_t)$, $\xi_t$ uniformly $\sim B_0(r)$
\ENDFOR
\end{algorithmic}
\end{algorithm}

We prove that \algo{PGDQGC} has the following performance guarantee:
\begin{theorem}[Formal version of \thm{ZeroJordan}]\label{thm:ZeroJordanF}
Suppose we have a target function $F$ and its noisy evaluation $f$ satisfying \assume{ZeroProb} with $\nu\leq \tO(\delta^2\epsilon^6/d^4)$. \algo{PGDQGC} can find an $\epsilon$-SOSP of $F$ satisfying Eq.~\eq{SOSPDef} with probability at least $1-\delta$, using
\begin{align}
\tilde{O}\left(\frac{\ell B}{\epsilon^2}\cdot\log^4 d\right)
\end{align}
queries to $U_f$ defined in \eq{QZeroOracle}, under the following parameter choices:
\begin{align}\label{eq:ZeroJordanParam}
\eta=\frac{1}{\ell},\quad\delta_0=\frac{\delta\epsilon^2}{32\ell B}\chi^{-4},\quad r=\epsilon\chi^{-3}c^{-6},\quad\ut=\frac{\chi c}{\eta\sqrt{\rho\epsilon}},\quad\uf=\sqrt{\frac{\epsilon^3}{\rho}}\chi^{-3}c^{-5},
\end{align}
where $c$ is some large enough constant and $\chi=\max\{1,\log(d\ell B/\rho\epsilon\delta_0)\}$.
\end{theorem}

To prove \thm{ZeroJordanF}, we consider two cases where the current iteration $\x_t$ is not an $\epsilon$-SOSP of $F$. In the first case, the gradient $\norm{\nabla F(\x_t)}\geq\epsilon$ is larger than $\epsilon$. In the second case, the gradient $\norm{\nabla F(\x_t)}\leq\epsilon$ but the minimal eigenvalue of the Hessian matrix satisfies $\lambda_{\min}(\nabla^2 F(\x_t))\leq-\sqrt{\rho\epsilon}$. Intuitively, the proof of \thm{ZeroJordanF} is composed of the performance guarantees regarding both cases. For \algo{PGDQGC}, it takes $\ut=O(\log d)$ queries to $U_f$ to decrease the function value by $\uf=O(1/\log^3 d)$~\cite{jin2018local}. 

We first set $\nu\leq\frac{C_0}{\ell}\cdot(\frac{\delta\epsilon^3}{Bd^2\chi^4\ell})^2$ for some small enough constant $C_0$. Formally, we introduce the following lemma characterizing the performance of \algo{PGDQGC} when the gradient is large:
\begin{lemma}\label{lem:JordanLargeGrad}
Under the setting of \thm{ZeroJordanF}, for any iteration $t$ of \algo{PGDQGC} with $\norm{\nabla F(\x_t)}\geq\epsilon$, we have $F(\x_{t+1})-F(\x_t)\leq-\eta\epsilon^2/4$ with probability at least $1-\delta_0$, where $\delta_0$ is defined in Eq.~\eq{ZeroJordanParam}.
\end{lemma}
\begin{proof}
We set $\omega=2d/\delta_0$ and choose $C_0$ small enough such that
\begin{align}
\norm{\tnabla F(\x)-\nabla F(\x)}\leq \frac{\epsilon}{20}
\end{align}
with probability at least $1-\delta_0$ according to \lem{JordanPerf}.

Next, we choose $c$ such that $\norm{\tnabla F(\x)-\nabla F(\x)}\leq\epsilon/20$. Recall that the perturbation $\xi_t$ is chosen from $B_0(r)$, the stochastic part in each iteration $\kappa_t=\tnabla F(\x)-\nabla F(\x)+\xi_t$ is bounded by $\norm{\kappa_t}=\epsilon/10$. According to the update rule $\x_{t+1}=\x_t-\eta(\nabla F(\x_t)+\kappa_t)$ of \algo{PGDQGC}, we have
\begin{align}
F(\x_{t+1})&\leq F(\x_t)+\expval{\nabla F(\x_t),\x_{t+1}-\x_t}+\frac{\ell}{2}\norm{\x_{t+1}-\x_t}^2\nonumber\\
&\leq F(\x_t)-\eta\left[\norm{\nabla F(\x_t)}^2-\norm{\nabla F(\x_t)}\norm{\kappa_t}\right]+\frac{\eta^2\ell}{2}\left[\norm{\nabla F(\x_t)}^2+2\norm{\nabla F(\x_t)}\norm{\kappa_t}+\norm{\kappa_t}^2\right]\nonumber\\
&\leq F(\x_t)-\eta\norm{\nabla F(\x_t)}\left[\frac12\norm{\nabla F(\x_t)}-2\norm{\kappa_t}\right]+\frac{\eta}{2}\norm{\kappa_t^2}\nonumber\\\label{eq:JordanLargeGrad}
&\leq F(\x_t)-\frac{\eta\epsilon^2}{4}.
\end{align}
\end{proof}

In addition, we can generalize the following lemma in Ref.~\cite{jin2018local}.
\begin{lemma}[Lemma 67, Ref.~\cite{jin2018local}]\label{lem:JordanLargeGradJin}
Suppose we are given a oracle that outputs an gradient estimation $\tnabla F(\x_t)$ such that $\norm{\tnabla F(\x_t)-\nabla F(\x_t)}\leq\epsilon/20$. Consider a iteration $t$ of \algo{PGDQGC} with $\norm{\nabla F(\x_t)}\geq\epsilon$. By using the PGD update rule $\x_{t+1}=\x_t-\eta(\tnabla F(\x_t)+\xi_t)$, we have $F(\x_{t+1})-F(\x_t)\leq-\eta\epsilon^2/4$ with probability at least $1-\delta_0$, where $\delta_0$ is defined in Eq.~\eq{ZeroJordanParam}.
\end{lemma}

When the gradient is small but the minimal eigenvalue of the Hessian matrix is large, i.e., the function has a large negative curvature at the current iteration, we have the following lemma from Ref.~\cite{jin2018local}.
\begin{lemma}[Lemma 68, Ref.~\cite{jin2018local}]\label{lem:JordanSmallGrad}
Suppose we are given a oracle that outputs an gradient estimation $\tnabla F(\x_t)$ such that $\norm{\tnabla F(\x_t)-\nabla F(\x_t)}\leq\epsilon/20$ and the norm of the perturbation in PGD is bounded by $\norm{\xi_t}\leq r$ with $r>\epsilon/20$. If $\norm{\nabla F(\x_t)}\leq\epsilon$ and $\lambda_{\min}(\nabla^2 F(\x_t))\leq-\sqrt{\rho\epsilon}$. By using the PGD update rule $\x_{t+1}=\x_t-\eta(\tnabla F(\x_t)+\xi_t)$, we have $F(\x_{t+\ut})-F(\x_t)\leq-\uf$ with probability at least $1-\delta_0$ when running \algo{PGDQGC}.
\end{lemma}

Now, we are ready to prove \thm{ZeroJordanF}.
\begin{proof}[Proof of \thm{ZeroJordanF}]
Assume our \algo{PGDQGC} starts at point $\x_0$ and the local minimum of $F$ has value $F^*$. Since $F$ is $B$-bounded, we have $F(\x_0)-F^*\leq 2B$. Set the total number of iterations $T$ to be:
\begin{align}
T=3\max\left\{\frac{8B}{\eta\epsilon^2},\frac{2B\ut}{\uf}\right\}.
\end{align}
Assume for some iterations $\x_t$, we have $\nabla F(\x_t)\leq\epsilon$ and $\lambda_{\min}(\nabla^2 F(\x_t))\leq-\sqrt{\rho\epsilon}$. 
The error probability of this assumption is given later. Under this assumption, the function value decreases for $\uf$ after each $\ut$ iterations. The number of such iterations when \lem{JordanSmallGrad} can be called is bounded by $T/3$ times, for otherwise the function value will decrease greater than $2B\geq F(\x_0)-F^*$, which is impossible. The failure probability is composed of two parts: the failure probability of estimating the gradients in \lem{JordanPerf} and the failure probability of \lem{JordanSmallGrad}. In each iteration, the probability of failure is bounded by $2\delta_0$ according to the union bound. The overall probability that \algo{PGDQGC} fails to indicate a negative curvature is upper bounded by
\begin{align}
\frac{T}{3}\cdot 2\delta_0\leq\frac{\delta}{2}
\end{align}
for any $\chi$.

Excluding the iterations in which \lem{JordanSmallGrad} is applied, we still have $2T/3$ iterations left. We now consider the iterations $\x_t$ with large gradients ${\nabla F(\x_t)}\geq\epsilon$. According to \lem{JordanLargeGrad}, the function value decreases by at least $\eta\epsilon^2/4$ with probability at least $1-\delta_0$ in each iteration. Thus there can be at most $T/3$ steps with large gradients, for otherwise, the function value will decrease greater than $2B\geq F(\x_0)-F^*$, which is impossible. The fail probability is bounded by
\begin{align}
\frac{T}{3}\cdot \delta_0\leq\frac{\delta}{2}.
\end{align}

In summary, we can deduce that with probability at least $1-\delta$, there are at most $T/3$ iterations within which the neighboring $\ut$ iterations have small gradients but large negative curvatures, and at most $T/3$ iterations with large gradients. Therefore, the rest $T/3$ iterations must be $\epsilon$-SOSPs of target function $F$. The number of queries is thus bounded by
\begin{align}
T\leq\tO\left(\frac{B\ell}{\epsilon^2}\cdot\log^4 d\right).
\end{align}
\end{proof}

The above \thm{ZeroJordanF} indicates that our PGD method with quantum gradient computation still converges and finds an $\epsilon$-SOSP using the same number of iterations (i.e., the same number of queries), even if there exists small noise on the quantum evaluation oracles. We remark that compared to Algorithm $4$ in Ref.~\cite{zhang2021quantum}, \algo{PGDQGC} employs a classical perturbation uniformly chosen from the ball $\mathbb{B}(\0,r)$. Therefore, \algo{PGDQGC} requires no access to the quantum evaluation oracle without noise.

It is natural to ask if we can improve the dependence on $\log d$ in the query complexity. We answer this question with an affirmative answer in \append{PGDQSGC} under some additional assumptions. Consider if we have functions $F$ and $f$ that satisfy \assume{ZeroProb} with $\nu\leq O(\epsilon^6/d^4)$ and we further assume that $f$ is twice differentiable with $\sup_\x\norm{\nabla f-\nabla F}\leq O(\ell/d^{2+\zeta})$ and $\sup_\x\norm{\nabla^2 f-\nabla^2 F}\leq O(\rho/d^{1.5+\zeta})$ for arbitrary $\zeta>0$. We propose a quantum algorithm that can find an $\epsilon$-SOSP for $F$ using $\tO(\ell B/\epsilon^2\cdot\log^2 d)$ queries to the quantum evaluation oracle in Eq.~\eq{QZeroOracle}. 

In addition, we can use the techniques above to prove the algorithmic upper bound for function pair $(F,f)$ satisfying \assume{FirstProb} with $\tnu\leq O(\epsilon/d^{0.5+\zeta})$ for $\zeta>0$ and $\zeta=\Omega(1/\log(d))$. We provide the following corollary corresponding to the last line in \tab{main2}.
\begin{corollary}\label{cor:FirstPGD}
Suppose we have a target function $F$ and a noisy function $f$ satisfying \assume{FirstProb} with $\tnu\leq O(\epsilon/d^{0.5+\zeta})$ for $\zeta>0$ and $\zeta=\Omega(1/\log(d))$. Consider the gradient descent $\x_{t+1}=\eta(\nabla f+\xi_t)$ with $\xi_t$ uniformly chosen from ball $\mathbb{B}(\0,r)$. This rule can output an $\epsilon$-SOSP of $F$ satisfying Eq.~\eq{SOSPDef}, using
\begin{align}
\tilde{O}\left(\frac{\ell B}{\epsilon^2}\cdot\log^4 d\right)
\end{align}
queries to $U_\g$ in \eq{QFirstOracle} with probability $1-\delta$, under the following parameter choices
\begin{align}
\eta=\frac{1}{\ell},\quad\delta_0=\frac{\delta\epsilon^2}{4\ell B}\chi^{-4},\quad r=\epsilon\chi^{-3}c^{-6},\quad\ut=\frac{\chi c}{\eta\sqrt{\rho\epsilon}},\quad\uf=\sqrt{\frac{\epsilon^3}{\rho}}\chi^{-3}c^{-5},
\end{align}
where $c$ is some large enough constant and $\chi=\max\{1,\log(d\ell B/\rho\epsilon\delta_0)\}$.
\end{corollary}
\begin{proof}
Without loss of generality, we set $\tnu\leq O(\epsilon/d^{0.5+\zeta})$ and $\zeta\geq\log 20/\log(d)$ such that $\tnu\geq\epsilon/20$. We then choose $c$ large enough such that $\norm{\tnabla F(\x)-\nabla F(\x)}\leq\epsilon/20$. As the perturbation $\xi_t$ is chosen from $B_0(r)$, the stochastic part in each iteration $\kappa_t=\tnabla F(\x)-\nabla F(\x)+\xi_t$ is bounded by $\norm{\kappa_t}=\epsilon/10$. Similar to the proof of \thm{ZeroJordanF}, we set 
\begin{align}
T=3\max\left\{\frac{8B}{\eta\epsilon^2},\frac{2B\ut}{\uf}\right\}.
\end{align}

Suppose for some iterations, the function have small gradients $\nabla F(\x_t)\leq\epsilon$ and large negative curvatures $\lambda_{\min}(\nabla^2 F(\x_t))\leq-\sqrt{\rho\epsilon}$. Under this assumption, the function value decreases for $\uf$ after each $\ut$ iterations according to \lem{JordanSmallGrad}. The number of such iterations when \lem{JordanSmallGrad} can be called is bounded by $T/3$ times, for otherwise, the function value will decrease greater than $2B\geq F(\x_0)-F^*$, which is impossible. The failure probability is bounded above by
\begin{align}
\frac{T}{3}\cdot \delta_0\leq\delta/2.  
\end{align}

Except for the iterations that \lem{JordanSmallGrad} is applied, we still have $2T/3$ iterations left. We now consider the iterations $\x_t$ with large gradients, ${\nabla F(\x_t)}\geq\epsilon$. According to \lem{JordanLargeGradJin}, the function value decreases by at least $\eta\epsilon^2/4$ with the probability at least $1-\delta_0$ in each iteration. Thus there can be at most $T/3$ steps with large gradients, for otherwise, the function value will decrease greater than $2B\geq F(\x_0)-F^*$, which is impossible. The failure probability is again bounded above by
\begin{align}
\frac{T}{3}\cdot \delta_0\leq\delta/2.  
\end{align} 

Therefore, we can deduce that with probability at least $1-\delta$, there are at most $T/3$ iterations resulting in points having small gradients but large negative curvature, and at most $T/3$ iterations with large gradients. Therefore, the rest $T/3$ iterations must be $\epsilon$-SOSPs. The number of the queries is bounded by
\begin{align}
T\leq\tO\left(\frac{B\ell}{\epsilon^2}\cdot\log^4 d\right).
\end{align}
\end{proof}

\cor{FirstPGD} indicates that when the gradient $\g=\nabla f$ of the noisy function is close enough to the gradient $\nabla F$ of the target function, the PGD algorithm can converge even if the gradient $\g$ is noisy. As we can directly query the noisy gradient, the quantum algorithms such as quantum mean estimation~\cite{hamoudi2021quantum,cornelissen2022near} or quantum gradient estimation~\cite{jordan2005fast,gilyen2019optimizing} cannot provide speedup in this case. Moreover, quantum approaches to add perturbation such as quantum simulation~\cite{zhang2021quantum} require zeroth-order information, which is unavailable under \assume{FirstProb}. Therefore, there is no quantum speedup compared to the classical gradient descent in the setting of \cor{FirstPGD}.

%%%%%%%%%%%%%%%%%%%%%%%%%%%%%%%%%%%%%%%%%%%%%%%%%%%%%%%%%%%%%%%%%%%%%%%%%%%%%%

\section{Quantum Speedup Using Mean Estimation}\label{sec:GaussMean}
When the noise strength further increases, it exceeds the robustness of quantum PGD. To handle this issue, we apply a Gaussian smoothing to the noisy function $f$ inspired by Ref.~\cite{jin2018local}, which can turn a possibly nonsmooth or even non-continuous $f$ into a function $f_\sigma$ with ``good" properties such as smoothness and Hessian-Lipschitzness.

\subsection{Zeroth-order Algorithm and Performance Guarantee}\label{sec:ZeroGaussMean}
In this section, we introduce a quantum algorithm based on Gaussian smoothing for function pairs $(F,f)$ satisfying \assume{ZeroProb} with $\nu\leq O(\epsilon^{1.5}/d)$. We formally define the \textit{Gaussian smoothing} for a function $f$ as follows.
\begin{definition}\label{defn:ZeroGaussSmooth}
Given a function $f\colon\mathbb{R}^d\to\mathbb{R}$, we define its Gaussian smoothing $f_{\sigma}\colon\R^d\to\R$ as
\begin{align}
f_\sigma\coloneqq\mathbb{E}_{\z\sim\cN(0,\sigma^2 I)}[f(\x+\z)],
\end{align}
where the parameter $\sigma$ is the smoothing radius.
\end{definition}

Given a noisy function $f$ and a target function $F$ satisfying \assume{ZeroProb}, Gaussian smoothing transfers the (probably even non-smooth or not differentiable) noisy $f$ into a smooth function $f_\sigma$ that has close gradient and Hessian with $F$. Formally, $f_\sigma$ has the following properties according to Ref.~\cite{jin2018local}:
\begin{lemma}[Lemma 13, Ref.~\cite{jin2018local}]\label{lem:ZeroGaussSmoothProp}
 Assume the function pair $(F,f)$ satisfies \assume{ZeroProb}, the Gaussian smoothing $f_\sigma$ of $f$ satisfies the following properties.
\begin{itemize}
    \item $f_\sigma(\x)$ is $O(\ell+\nu/\sigma^2)$-smooth and $O(\rho+\nu/\sigma^3)$-Hessian Lipshitz.
    \item The distance between the gradient and the Hessian of $f_\sigma$ and $F$ at any $\x$ is bounded by $\norm{\nabla f_\sigma(\x)-\nabla F(\x)}\leq O(\rho d\sigma^2+\nu/\sigma)$ and $\norm{\nabla^2 f_\sigma(\x)-\nabla^2 F(\x)}\leq O(\rho\sqrt{d}\sigma+\nu/\sigma^2)$.
\end{itemize}
\end{lemma}

The first part of \lem{ZeroGaussSmoothProp} demonstrates that the Gaussian smoothing $f_\sigma$ is a smooth and Hessian Lipschitz function. Thus we can perform standard gradient descent on $f_\sigma$ with a polynomial convergence rate. The second part of \lem{ZeroGaussSmoothProp} indicates that the gradients and Hessians of $f_\sigma$ are similar to those of the target function $F$ up to a term related to the noise rate $\nu$ and the smoothing radius $\sigma$. As the noise rate $\nu$ increases and the noisy function $f$ deviates further from the target function $F$, we have to choose a larger parameter $\sigma$ to bound the terms $\nu/\sigma$, $\nu/\sigma^2$, and $\nu/\sigma^3$. However, choosing a larger smoothing radius $\sigma$ will increase the term $\rho d\sigma^2$ and $\rho\sqrt{d}\sigma$, which erases the information about local geometry of $F$. Hence, the choice of $\sigma$ must balance between the two terms in the bounds in \lem{ZeroGaussSmoothProp}. 

Suppose we have an $\tilde{\epsilon}$-SOSP $\x_{\text{SOSP}}$ of the Gaussian smoothing $f_\sigma$. One have to guarantee that an $\tilde{\epsilon}$-SOSP of the Gaussian smoothing $f_\sigma$ is also an $\epsilon$-SOSP of $F$. We now search for the value of $\sigma$ and $\tilde{\epsilon}$ such that $\nu$ is maximized. According to \lem{ZeroGaussSmoothProp}, we can bound the gradient and the minimal eigenvalue of Hessian for $F(\x_{\text{SOSP}})$ by the following inequalities. 
\begin{align}
\norm{\nabla F(\x_{\text{SOSP}})}&\leq\rho d\sigma^2+\frac{\nu}{\sigma}+\tilde{\epsilon},
\end{align}
whereas
\begin{align}
\lambda_{\min}(\nabla^2 F(\x_{\text{SOSP}}))&\geq\lambda_{\min}(\nabla^2 f_\sigma(\x_{\text{SOSP}}))+\lambda_{\min}(\nabla^2 F(\x_{\text{SOSP}})-\nabla^2 f_\sigma(\x_{\text{SOSP}}))\nonumber\\
&\geq-\sqrt{\left(\rho+\frac{\nu}{\sigma^3}\right)\tilde{\epsilon}}-\norm{\nabla^2 f_\sigma(\x_{\text{SOSP}})-\nabla^2 F(\x_{\text{SOSP}})}\nonumber\\
&\geq-\sqrt{\left(\rho+\frac{\nu}{\sigma^3}\right)\tilde{\epsilon}}-\left(\rho\sqrt{d}\sigma+\frac{\nu}{\sigma^2}\right).
\end{align}
Hence, to guarantee that an $\tilde{\epsilon}$-SOSP of $f_\sigma$ is an $\epsilon$-SOSP of $F$, we only need the following set of inequalities to be satisfied(up to constant factors).
\begin{align}\label{eq:GaussIneq1}
\rho\sqrt{d}\sigma+\frac{\nu}{\sigma^2}&\leq O(\sqrt{\rho\epsilon}),\\\label{eq:GaussIneq2}
\rho d\sigma^2+\frac{\nu}{\sigma}&\leq O(\epsilon),\\\label{eq:GaussIneq3}
\left(\rho+\frac{\nu}{\sigma^3}\right)\tilde{\epsilon}&\leq O(\rho\epsilon).
\end{align}
From \eq{GaussIneq1} and \eq{GaussIneq2}, we have 
\begin{align}
\sigma&\leq O\left(\sqrt{\frac{\epsilon}{\rho d}}\right),\\
\nu&\leq\sqrt{\rho\epsilon}\sigma^2=O\left(\sqrt{\frac{\epsilon^3}{\rho}}\cdot\frac{1}{d}\right),\\
\tilde{\epsilon}&\leq\frac{\rho\epsilon}{\left(\rho+\frac{\nu}{\sigma^3}\right)}=O\left(\frac{\epsilon}{\sqrt{d}}\right).
\end{align}
The above results indicate that we can guarantee that an $O(\epsilon/\sqrt{d})$-SOSP for $f_\sigma$ is an $\epsilon$-SOSP of the target function $F$.

The next step is to find an $O(\epsilon/\sqrt{d})$-SOSP of $f_\sigma$ using queries to the noisy oracle in~\eq{QZeroOracle}. Through Gaussian smoothing, we convert the function evaluations of $f$ into stochastic gradients of $f_\sigma$. According to Ref.~\cite{duchi2015optimal} the gradients of $f_\sigma$ can be calculated as  
\begin{align}
\nabla f_\sigma=\frac{1}{\sigma^2}\mathbb{E}_{\z\sim\cN(0,\sigma^2I)}[\z(f(\x+\z)-(\x))].
\end{align}
One can thus compute the gradient for $f_\sigma$ by querying the function value of $f$. However, the gradient is unbiasedly computed through averaging over the continuous Gaussian distribution. To approximate the gradient, we employ the zeroth-order quantum oracle in \eq{QZeroOracle} to sample the stochastic gradient estimation $\z[f(\x+\z)-f(\x)]/\sigma^2$, where $\z\sim\cN(0,\sigma^2I)$. The stochastic gradient has the following properties. 
\begin{lemma}[Lemma 14, Ref.~\cite{jin2018local}]\label{lem:StochGradGauss}
We denote $\g(\x;\z)=\z[f(\x+\z)-f(\x)]/\sigma^2$, where $\z\sim\cN(0,\sigma^2I)$. The following inequalities hold:
\begin{align}
&\mathbb{E}_\z \g(\x;\z)=\nabla f_\sigma(\x),\\
&\Pr\left[\norm{\g(\x;\z)-\nabla f_\sigma(\x)}\geq t\right]\leq\exp(-Bt^2/\sigma),\qquad\forall t>0.
\end{align}
The second inequality demonstrates that $\g(\x;\z)$ is a sub-Gaussian random variable with a tail $B/\sigma$.
\end{lemma}

\lem{StochGradGauss} guarantees that by sampling a large mini-batch and evaluating the mean of the stochastic gradients, the value converges to the gradient $\nabla f_\sigma(\x)$ of the Gaussian smoothing $f_\sigma(\x)$. Classically, the batch size required for the sampling can be obtained by the Chernoff bound (say, e.g. Ref.~\cite{lugosi2019mean}).
\begin{lemma}\label{lem:ZeroClassicalSamp}
Given a fixed point $\x$ and the mini-batch size $m$, for any $\delta>0$, we have:
\begin{align}
\norm{\nabla f(\x)-\frac{1}{m}\sum_{i=1}^m \g(\x;\z^{(i)})}\leq\sqrt{\frac{2\sigma_0^2 d}{m}\log(\frac d\delta)}
\end{align}
with probability at least $1-\delta$, where $\sigma_0=B/\sigma$ is the standard deviation of the stochastic gradient.
\end{lemma}
\lem{ZeroClassicalSamp} indicates that it is sufficient to choose a mini-batch of size
\begin{align}
m\geq\frac{2\sigma_0^2 d}{\epsilon^2}\log(\frac d\delta),
\end{align}
where $\sigma_0=B/\sigma$, to estimate the gradient $\nabla f_\sigma(\x)$ within $\epsilon$ under Euclidean norm with probability at least $1-\delta$. In addition, this bound is optimal in any classical algorithms~\cite{hopkins2020mean}, or equivalently, any classical multivariate mean estimator with batch size less than this quantity will fail on a certain stochastic gradient function $\g(\x;\z)$.

Quantumly, the well-known amplitude estimation algorithm~\cite{brassard2002quantum} provides a smaller error rate when estimating the mean of Bernoulli random variables. For the multivariate mean estimation problem of a random vector, quantum algorithms can also provide a speedup under certain circumstances. In particular, we consider the mean estimation task of estimating $\g(\x;\z)$ in \lem{StochGradGauss} given a \textit{binary oracle} defined as follows.
\begin{definition}
Consider the random variable $\g(\x;\z)\in\R^d$ with $\z\sim\cN(0,\sigma^2I)$. Let $\H_\z$ and $\H_{\g,\x}$ be two Hilbert spaces with basis states $\{\ket{\z}\}_{\z}$ and $\{\g(\x;\z)\}_{\z}$, which contains quantum state encoding vectors $\z$ and $\g(\x;\z)$, respectively. The binary oracle $B_\z\colon\H_\z\otimes\H_{\g,\x}\to\H_\z\otimes\H_{\g,\x}$ is defined as
\begin{align}\label{eq:QBinaryOracle}
B_\z:\ket{\z}\ket{\0}\to\ket{\z}\ket{\g(\x;\z)}, \qquad\forall \z\sim\cN(0,\sigma^2 I),
\end{align}
where we assume $\0\in\{\g(\x;\z)\}_{\z}$.
\end{definition}
In practice, the above binary oracle
can be constructed by employing two quantum evaluation oracles in \eq{QZeroOracle}. Using such binary oracle, Ref.~\cite{cornelissen2022near} provides the following performance guarantee.
\begin{lemma}[Theorem 3.5, Ref.~\cite{cornelissen2022near}]\label{lem:ZeroQSamp}
Suppose $\g$ is a $d$-dimensional random vector with mean $\mu$ and covariance matrix $\Sigma$ such that $\Tr(\Sigma)=\sigma_0^2$. Given two real values $\delta\in(0,1)$ and $m\geq\log(d/\delta)$, there exists a quantum algorithm that outputs a mean estimation $\tilde{\mu}$ such that
\begin{align}
\norm{\tilde{\mu}-\mu}\leq
\begin{cases}
O\left(\sqrt{\frac{\sigma_0^2}{m}}\right),&n\leq d,\\
O\left(\frac{\sqrt{d\sigma_0^2}\log(\frac d\delta)}{m}\right),\qquad&n>d,
\end{cases}
\end{align}
with probability at least $1-\delta$. Such an algorithm requires $\tO(m)$ queries to the binary oracle. 
\end{lemma}
\lem{ZeroQSamp} indicates that it only requires
\begin{align}
m\geq O\left(\frac{\sqrt{d}\sigma_0}{\epsilon}\log(\frac d\delta)\right)
\end{align}
samples to estimate the gradient $\nabla f_\sigma(\x)$ within error $\epsilon$ with high probability. Compared with the classical mini-batch size in \lem{ZeroClassicalSamp}, quantum mean estimation provides a quadratic reduction when the classical mini-batch size is $\Omega(d)$. It is worthwhile to mention that the error scaling in \lem{ZeroQSamp} is near-optimal up to logarithmic factors~\cite{cornelissen2022near}.

We consider using the PGD with stochastic gradient estimation to find an $\epsilon$-SOSP of the target function $F$ (also an $O(\epsilon/\sqrt{d})$-SOSP of $f_\sigma$) using noisy function $f$ in \eq{QZeroOracle}. The detailed algorithm is given in \algo{PSGDQME}.

\begin{algorithm}[H]
\caption{Perturbed Stochastic Gradient Descent with Quantum Mean Estimation}
\label{algo:PSGDQME}
\begin{algorithmic}[1]
\REQUIRE $\x_0$, learning rate $\eta$, noise ratio $r$, mini-batch size $m$
\FOR {$t=0,1,...,T$}
\STATE Estimate the gradient $\tnabla f_\sigma(\x_t)$ of $f_\sigma(\x_t)$ using quantum mean estimation and $m$ binary queries to \eq{QBinaryOracle}
\STATE $\x_{t+1}\leftarrow \x_t-\eta(\tnabla f_\sigma(\x_t)+\xi_t)$, $\xi_t$ uniformly $\sim B_0(r)$
\ENDFOR
\end{algorithmic}
\end{algorithm}

We now prove the performance guarantee for \algo{PSGDQME}, which is the formal version of \thm{ZeroMean}.
\begin{theorem}[Formal version of \thm{ZeroMean}]\label{thm:ZeroMeanF}
Suppose we have a target function $F$ and its noisy evaluation $f$ satisfying \assume{ZeroProb} with $\nu\leq O(\sqrt{\epsilon^3/\rho}\cdot(1/d))$. With probability at least $1-\delta$, \algo{PSGDQME} finds an $\epsilon$-SOSP of $F$ satisfying \eq{SOSPDef}, using
\begin{align}
\tO\left(\frac{d^{2.5}}{\epsilon^{3.5}}\cdot\poly(\Delta_f,\ell,\rho)\right)
\end{align}
queries to $U_f$ in \eq{QZeroOracle}, under the following parameter choices:
\begin{align}\label{eq:ZeroMeanParam}
\eta=\frac{1}{\ell'},\quad\delta_0=\frac{\delta\epsilon'^2}{16\ell' \Delta_f}\chi^{-4},\quad r=\epsilon'\chi^{-3}c^{-6},\quad\ut=\frac{\chi c}{\eta\sqrt{\rho'\epsilon'}},\quad\uf=\sqrt{\frac{\epsilon'^3}{\rho'}}\chi^{-3}c^{-5},
\end{align}
where $c$ is some large enough constant, $\Delta_f=f_\sigma(\x_0)-f_\sigma(\x^*)$ is the value between the initial point $\x_0$ and the global minima point $\x^*$, $\chi=\max\{1,\log(d\ell' \Delta_f/\rho'\epsilon'\delta_0)\}$, $\ell'=O(\ell+\sqrt{\epsilon/\rho})$, and $\rho'=O(\rho)$ are the smoothness and Hessian-Lipschitz parameters for $f_\sigma$, and $\epsilon'=O(\epsilon/\sqrt{d})$ such that an $\epsilon'$-SOSP of $f_\sigma$ is an $\epsilon$-SOSP of $F$.
\end{theorem}
\begin{proof}
Notice that an $\epsilon$-SOSP of target function $F$ is an $O(\epsilon/\sqrt{d})$-SOSP of function $f_\sigma$, we only need to prove that \algo{PSGDQME} converges to a $O(\epsilon/\sqrt{d})$-SOSP of $f_\sigma$ using $\tO(\ell B d^{2.5}/\epsilon^{3.5})$ queries.

In each iteration, \algo{PSGDQME} estimates the gradient $\nabla f_\sigma(\x_t)$ using quantum mean estimation with $O(m)$ queries to the quantum evaluation oracle. According to \lem{StochGradGauss}, the stochastic gradient $\g(\x_t;\z^{(i)})$ for $\z\sim\cN(0,\sigma^2I)$ is a random vector with mean $\nabla f_\sigma(\x_t)$ and variance $\sigma_0^2=B^2/\sigma^2$. As we choose $\sigma=\sqrt{\epsilon/\rho d}$, we require
\begin{align}
m&\geq O\left(\frac{\sqrt{d}(B/\sigma)}{\epsilon/\sqrt{d}}\log(\frac {d}{\delta_0})\right)\\
&=\tO(B\sqrt{\rho}\cdot\sqrt{\frac{d^3}{\epsilon^3}})
\end{align}
queries to bound the error $\norm{\tnabla f_\sigma(\x_t)-\nabla f_\sigma(\x_t)}\leq\epsilon'/20\leq O(\epsilon/\sqrt{d}\cdot(1/20))$ with probability at least $1-\delta_0$ according to \lem{ZeroQSamp}.

Next, \algo{PSGDQME} employs the estimations of the gradient and the PGD to find an $\epsilon'$-SOSP of $f_\sigma$. Recall that we choose $\sigma=O(\sqrt{\epsilon/\rho d})$ and $\nu=O(\sqrt{\epsilon^3/\rho}\cdot(1/d))$, $f_\sigma$ is thus $\ell'$-smooth and $\rho'$-Hessian Lipschitz, where
\begin{align}
\ell'&=O\left(\ell+\frac{\nu}{\sigma^2}\right)=O\left(\ell+\sqrt{\frac{\epsilon}{\rho}}\right),\\
\rho'&=O\left(\rho+\frac{\nu}{\sigma^3}\right)=O(\rho).
\end{align}
We consider the number of queries required to find an $\epsilon'$-SOSP of $f_\sigma$. We set the total iteration number to be:
\begin{align}
T=3\max\left\{\frac{4\Delta_f}{\eta\epsilon'^2},\frac{\Delta_f\ut}{\uf}\right\}.
\end{align}
Similar to the proof in the previous section, we consider the two cases when a $\x_t$ is not local minima. Suppose for some iterations $\x_t$, we have $\nabla f_\sigma(\x_t)\leq\epsilon'$ and $\lambda_{\min}(\nabla^2 f(\x_t))\leq-\sqrt{\rho'\epsilon'}$. The error probability of this assumption is given later. Under this assumption, the function value decreases for $\uf$ after each $\ut$ iterations according to \lem{JordanSmallGrad}. Therefore, the number of such iterations when \lem{JordanSmallGrad} can be called is bounded by $T/3$ times, 
for otherwise the function value will decrease greater than $\Delta_f=f_\sigma(\x_0)-f_\sigma(\x^*)$, which is impossible. The failure probability is composed of two parts: the failure probability for estimating the gradient in \lem{JordanPerf} and the failure probability for \lem{JordanSmallGrad}. In each iteration, the probability of failure is bounded by $2\delta_0$ according to the union bound. The overall probability that \algo{PGDQGC} fails to indicate a negative curvature is upper bounded by
\begin{align}
\frac{T}{3}\cdot 2\delta_0\leq\frac{\delta}{2}
\end{align}
for any $\chi$.

Excluding the iterations that \lem{JordanSmallGrad} is applied, we still have $2T/3$ iterations left. We now consider the iterations $\x_t$ with large gradients ${\nabla f_\sigma(\x_t)}\geq\epsilon'$. According to \lem{JordanLargeGradJin}, the function value decreases by at least $\eta\epsilon'^2/4$ with a probability of at least $1-\delta_0$ in each iteration. Thus there can be at most $T/3$ steps with large gradients, for otherwise, the function value will decrease greater than $\Delta_f$, which is impossible. The failure probability is bounded again by
\begin{align}
\frac{T}{3}\cdot \delta_0\leq\frac{\delta}{2}.
\end{align}

In summary, with probability at least $1-\delta$, there are at most $T/3$ iterations within which the neighboring $\ut$ iterations have small gradients but large negative curvatures, and at most $T/3$ iterations with large gradients. Therefore, the rest $T/3$ iterations must be $\epsilon$-SOSPs. The number of the queries is thus bounded by
\begin{align}
T\leq\tO\left(\frac{\Delta_f\ell'}{\epsilon'^2}\cdot\chi^{4}\cdot m\right)=\tO\left(\frac{d^{2.5}}{\epsilon^{3.5}}\cdot\poly(\Delta_f,\ell,\rho)\right).
\end{align}
\end{proof}

\thm{ZeroMeanF} provides a quantum upper bound $O(d^{2.5}/\epsilon^{3.5})$ in finding $\epsilon$-SOSPs of $F$ using noisy oracle $f$ at $\nu\leq O(\sqrt{\epsilon^3/\rho}\cdot 1/d)$ while the classical upper bound requires $O(d^4/\epsilon^5)$ queries~\cite{jin2018local}. The essence of the speedup lies in the quadratic reduction provided by the quantum mean estimation in the mini-batch size $m$.

%==================================================================

\subsection{First-order Algorithm and Performance Guarantee}\label{sec:FirstGaussMean}

Consider a pair of functions $(F,f)$ satisfying \assume{FirstProb} with a relatively large noise strength such that \cor{FirstPGD} fails to apply. Recall that in the previous subsection we have implemented a Gaussian smoothing for the noisy zeroth-oracle defined in \defn{ZeroGaussSmooth}. Now, we introduce the Gaussian smoothing of the noisy gradient, which is defined as:
\begin{align}
\nabla f_\sigma (\x)\coloneqq\mathbb{E}_\z\big[ \nabla f(\x+\z)\big],\quad\z\sim\cN(0,\sigma^2I).
\end{align}
After permutating the expectation operator and the gradient operator, we obtain 
\begin{align}
    \nabla f_\sigma(\x)=\nabla\cdot\big(\mathbb{E}_\z[ f(\x+\z)]\big),
\end{align}
which indicates that $f_\sigma(\x)$ is a Gaussian smoothing of $f$. Similar to \lem{ZeroGaussSmoothProp}, we deduce the following property of $\nabla f_\sigma(\x)$, which originally appeared in Ref.~\cite{jin2018local}. 
\begin{lemma}[Lemma 48, Ref.~\cite{jin2018local}]\label{lem:FirstGaussSmoothProp}
Assume the function pair $(F,f)$ satisfies \assume{FirstProb}. The Gaussian smoothing $\nabla f_\sigma$ of the noisy gradient $\g=\nabla f$ satisfies:
\begin{itemize}
    \item $f_\sigma(\x)$ is $O(\ell+\tnu/\sigma)$-smooth and $O(\rho+\tnu/\sigma^2)$-Hessian Lipschitz.
    \item The distances between the gradients and the Hessians of $f_\sigma$ and $F$ are bounded. In particular, we have $\norm{\nabla f_\sigma(\x)-\nabla F(\x)}\leq O(\rho d\sigma^2+\tnu)$ and $\norm{\nabla^2 f_\sigma(\x)-\nabla^2 F(\x)}\leq O(\rho\sqrt{d}\sigma+\tnu/\sigma)$, respectively.
\end{itemize}
\end{lemma}

We can bound the deviation of $\nabla f_\sigma$ from $\nabla F$ and maintain the information about the local geometry of $F$, as well as guaranteeing that any $\tilde{\epsilon}$-SOSP $\x_{\text{SOSP}}$ of $f_\sigma$ is also an $\epsilon$-SOSP of $F$, by choosing a suitable Gaussian smoothing parameter $\sigma$. We optimize $F$ through the Gaussian smooth $f_\sigma$. According to \lem{FirstGaussSmoothProp}, the gradients and the eigenvalues of Hessians of $F$ are bounded by:
\begin{align}
\norm{\nabla F(\x_{\text{SOSP}})}\leq\rho d\sigma^2+\tnu+\tilde{\epsilon},
\end{align}
whereas
\begin{align}
\lambda_{\min}(\nabla^2 F(\x_{\text{SOSP}}))&\geq\lambda_{\min}(\nabla^2 f_\sigma(\x_{\text{SOSP}}))+\lambda_{\min}(\nabla^2 F(\x_{\text{SOSP}})-\nabla^2 f_\sigma(\x_{\text{SOSP}}))\nonumber\\
&\geq-\sqrt{\left(\rho+\frac{\tnu}{\sigma^2}\right)\tilde{\epsilon}}-\norm{\nabla^2 f_\sigma(\x_{\text{SOSP}})-\nabla^2 F(\x_{\text{SOSP}})}\nonumber\\
&\geq-\sqrt{\left(\rho+\frac{\tnu}{\sigma^2}\right)\tilde{\epsilon}}-\left(\rho\sqrt{d}\sigma+\frac{\tnu}{\sigma}\right).
\end{align}
To guarantee that any $\tilde{\epsilon}$-SOSP of $f_\sigma$ is an $\epsilon$-SOSP of $F$, we only need the following set of inequalities to be satisfied (up to constant factors).
\begin{align}\label{eq:FirstGaussIneq1}
\rho\sqrt{d}\sigma+\frac{\tnu}{\sigma}&\leq O(\sqrt{\rho\epsilon}),\\\label{eq:FirstGaussIneq2}
\rho d\sigma^2+\tnu&\leq O(\epsilon),\\\label{eq:FirstGaussIneq3}
\left(\rho+\frac{\tnu}{\sigma^2}\right)\tilde{\epsilon}&\leq O(\rho\epsilon).
\end{align}
From \eq{FirstGaussIneq1} and \eq{FirstGaussIneq2}, we can deduce that 
\begin{align}
\sigma&\leq O\left(\sqrt{\frac{\epsilon}{\rho d}}\right),\\
\tnu&\leq\sqrt{\rho\epsilon}\sigma=O\left(\frac{\epsilon}{\sqrt{d}}\right),\\
\tilde{\epsilon}&\leq\frac{\rho\epsilon}{\left(\rho+\frac{\tnu}{\sigma^2}\right)}=O\left(\frac{\epsilon}{\sqrt{d}}\right).
\end{align}
Hence, an $O(\epsilon/\sqrt{d})$-SOSP of the Gaussian smoothing $f_\sigma$ is an $\epsilon$-SOSP of the target function $F$. Similar to \assume{ZeroProb}, we now have to find an $O(\epsilon/\sqrt{d})$-SOSP for the Gaussian smoothing $f_\sigma$ through queries to first-order noisy oracle $\nabla f(\x)$. To approximate the gradient $\nabla f_\sigma$, we sample from the stochastic gradient estimation $\g_\sigma(\x;\z)=\nabla f(\x+\z)$, where $\z\sim\cN(0,\sigma^2I)$. As shown in Ref.~\cite{jin2018local}, the stochastic gradient has the following properties, which is similar to \lem{StochGradGauss}. 
\begin{lemma}[Lemma 53, Ref.~\cite{jin2018local}]\label{lem:FirstStochGradGauss}
We denote $\g_\sigma(\x;\z)=\nabla f(\x+\z)$ for a sample from the noisy oracle, where $\z\sim\cN(0,\sigma^2I)$. The following inequalities hold:
\begin{align}
&\mathbb{E}_\z[\g_\sigma(\x;\z)]=\nabla f_\sigma(\x),\\
&\Pr\left[\norm{\g_\sigma(\x;\z)-\nabla f_\sigma(\x)}\geq t\right]\leq\exp(-Lt^2),\qquad\forall t>0.
\end{align}
The second inequality indicates that $\g_\sigma(\x;\z)$ is a sub-Gaussian random variable with a tail $L$ (Recall that $L$ is the smoothness parameter of $f$ in \assume{FirstProb}).
\end{lemma}

By a similar reduction to \lem{ZeroClassicalSamp}, we can deduce that the optimal sampling strategy requires
\begin{align}
m\geq O\left(\frac{\sigma_0^2 d}{\epsilon^2}\log(\frac d\delta)\right)
\end{align}
queries to approximate $\nabla f_\sigma(\x)$ with accuracy $\epsilon$, where $\sigma_0=L$ according to \lem{FirstStochGradGauss}. Quantumly, we can employ \lem{ZeroQSamp} and use only
\begin{align}\label{eq:MFirst}
m\geq O\left(\frac{\sqrt{d}\sigma_0}{\epsilon}\log(\frac d\delta)\right)
\end{align}
queries if given access to a quantum binary oracle defined in Eq.~\eq{QBinaryOracle}, which can be constructed by one query to the first-order oracle defined in Eq.~\eq{QFirstOracle}. We propose a first-order version of PGD with stochastic gradient queries to the smoothed function $f_{\sigma}$ and quantum mean estimation in \algo{FPSGDQME}.

\begin{algorithm}[H]
\caption{First-order Perturbed Stochastic Gradient Descent with Quantum Mean Estimation.}
\label{algo:FPSGDQME}
\begin{algorithmic}[1]
\REQUIRE $\x_0$, learning rate $\eta$, noise ratio $r$, mini-batch size $m$
\FOR {$t=0,1,...,T$}
\STATE Estimating the gradient $\tnabla f_\sigma(\x_t)$ of $f_\sigma(\x_t)$ using quantum mean estimation and $m$ queries to $\g_\sigma(\x;\z)=\nabla f(\x+\z)$ in the binary oracle
\STATE $\x_{t+1}\leftarrow \x_t-\eta(\tnabla f_\sigma(\x_t)+\xi_t)$, $\xi_t$ uniformly $\sim B_0(r)$
\ENDFOR
\ENSURE $\x_T$
\end{algorithmic}
\end{algorithm}

The goal of \algo{FPSGDQME} is to find an $O(\epsilon/\sqrt{d})$-SOSP of the Gaussian smoothing $f_\sigma$. The number of queries required can be bounded by the following theorem, which is the formal version of \thm{FirstMean}.
\begin{theorem}[Formal version of \thm{FirstMean}]\label{thm:FirstMeanF}
Suppose we have a target function $F$ and its noisy evaluation $f$ satisfying \assume{FirstProb} with $\tnu\leq O(\epsilon/\sqrt{d})$. With probability at least $1-\delta$, \algo{FPSGDQME} finds an $\epsilon$-SOSP of $F$ satisfying \eq{SOSPDef}, using
\begin{align}
\tO\left(\frac{d^{2}}{\epsilon^{3}}\cdot\poly(\Delta_f,\ell,\rho)\right)
\end{align}
queries to $U_\g$ defined in Eq.~\eq{QFirstOracle}, under the following choices of parameters:
\begin{align}\label{eq:FirstMeanParam}
\eta=\frac{1}{\ell'},\quad\delta_0=\frac{\delta\epsilon'^2}{4\ell' \Delta_f}\chi^{-4},\quad r=\epsilon'\chi^{-3}c^{-6},\quad\ut=\frac{\chi c}{\eta\sqrt{\rho'\epsilon'}},\quad\uf=\sqrt{\frac{\epsilon'^3}{\rho'}}\chi^{-3}c^{-5},
\end{align}
where $c$ is some large enough constant, $\Delta_f=f_\sigma(\x_0)-f_\sigma(\x^*)$ is the gap between the initial point $\x_0$ and the global minimum $\x^*$, and $\chi=\max\{1,\log(d\ell' \Delta_f/\rho'\epsilon'\delta_0)\}$. Here, $\ell'$, $\rho'$, and $\epsilon'$ have the same definition as in \thm{ZeroMeanF}.
\end{theorem}

Similar to \thm{ZeroMeanF}, \thm{FirstMeanF} presents a polynomial reduction in the query complexity of oracles using quantum mean estimation. In particular, \algo{FPSGDQME} requires only $\tO(d^2/\epsilon^3)$ queries to the first-order gradient oracles while its classical counterpart~\cite{jin2018local} requires $\tO(d^3/\epsilon^4)$.
\begin{proof}[Proof of \thm{FirstMeanF}]
Since an $\epsilon$-SOSP of the target function $F$ is an $O(\epsilon/\sqrt{d})$-SOSP of function $f_\sigma$, we only need to prove that \algo{FPSGDQME} converges to an $O(\epsilon/\sqrt{d})$-SOSP of $f_\sigma$ using $\tO( d^{2}/\epsilon^{3})$ queries.

In each iteration, \algo{FPSGDQME} estimates the gradient $\nabla f_\sigma(\x_t)$ via quantum mean estimation while using $O(m)$ queries to the quantum evaluation oracle in each mini-batch. According to Eq.~\eq{MFirst}, we require the mini-batch size
\begin{align}
m\geq O\left(\frac{\sqrt{d}L}{\epsilon/\sqrt{d}}\log(\frac{d}{\delta_0})\right)=\tilde{O}\left(\frac{d}{\epsilon}\right)
\end{align}
to bound the error $\norm{\tnabla f_\sigma(\x_t)-\nabla f_\sigma(\x_t)}\leq\epsilon'/20\leq O(\epsilon/\sqrt{d}\cdot(1/20))$ with probability at least $1-\delta_0$.

Next, we consider the number of queries required to find an $\epsilon'$-SOSP of $f_\sigma$. We set the total iteration number to be
\begin{align}
T=3\max\left\{\frac{4\Delta_f}{\eta\epsilon'^2},\frac{\Delta_f\ut}{\uf}\right\}.
\end{align}
We repeat the procedure in the proof of \thm{ZeroMeanF}. Suppose for some iterations, the function has small gradients $\nabla f_\sigma(\x_t)\leq\epsilon'$ and large negative curvatures $\lambda_{\min}(\nabla^2 f_\sigma(\x_t))\leq-\sqrt{\rho'\epsilon'}$. Under this assumption, the function value will decrease for $\uf$ after each $\ut$ iterations according to \lem{JordanSmallGrad}. Therefore, the number of such iterations is bounded by $T/3$, for otherwise the function value will decrease greater than $\Delta_f=f_\sigma(\x_0)-f_\sigma(\x^*)$, which is impossible. The failure probability of the above argument is bounded by
\begin{align}
\frac{T}{3}\cdot \delta_0\leq\delta/2.  
\end{align}
Except for the iterations where \lem{JordanSmallGrad} is applied, we still have $2T/3$ iterations left. We now consider the iterations $\x_t$ with large gradients, i.e., ${\nabla f_\sigma(\x_t)}\geq\epsilon'$. According to \lem{JordanLargeGradJin}, the function value decreases by at least $\eta\epsilon^2/4$ with a probability of at least $1-\delta_0$. Therefore, there can be at most $T/3$ steps with large gradients, for otherwise, the function value will decrease greater than $\Delta_f$, which is impossible. The failure probability of the above argument is again bounded by
\begin{align}
\frac{T}{3}\cdot \delta_0\leq\delta/2.  
\end{align}

Therefore, we can deduce that with probability at least $1-\delta$, there are at most $T/3$ iterations within which the neighboring $\ut$ iterations have small gradients but large negative curvatures, and at most $T/3$ iterations with large gradients. Hence, the rest $T/3$ iterations must be $\epsilon$-SOSPs. The number of queries can be bounded by
\begin{align}
T\leq\tO\left(\frac{\Delta_f\ell}{\epsilon'^2}\cdot\chi^4\cdot m\right)=\tO\left(\frac{d^{2}}{\epsilon^{3}}\cdot\poly(\Delta_f,\ell,\rho)\right).
\end{align}
\end{proof}

%%%%%%%%%%%%%%%%%%%%%%%%%%%%%%%%%%%%%%%%%%%%%%%%%%%%%%%%%%%%%%%%%%%%%%%%%%%%%%

\section{Classical and Quantum Lower Bounds in $d$}\label{sec:DLower}

In this section, we prove the query complexity lower bounds in the dimension $d$ of the input. Intuitively, the lower bound is obtained by constructing a hard instance and calculating its worst-case query complexity.

\subsection{Quasi-polynomial Lower Bound for Quantum Zeroth-order Methods}\label{sec:DLowerQuasiPoly}
\begin{figure}
    \centering
    \includegraphics[width=0.99\textwidth]{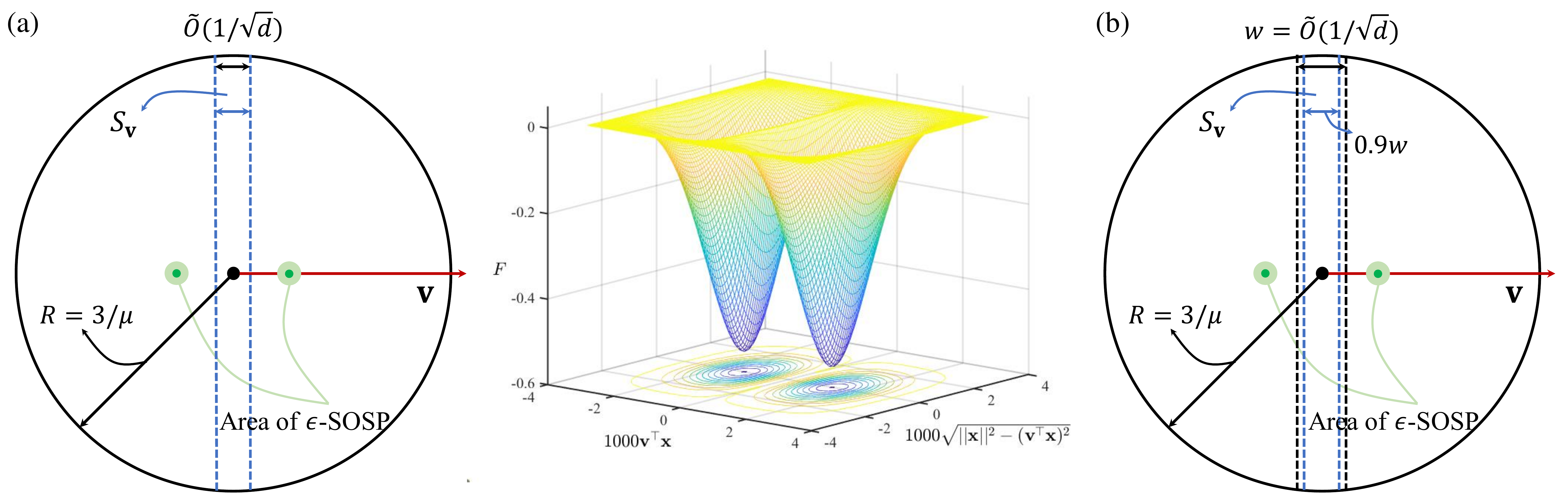}
    \caption{(a) The Sketch of the hard instance function for \thm{ZeroMeanLower}. The left figure illustrates the construction in the domain and the right figure shows a two-dimensional example. (b) The construction of the hard instance function for \thm{FirstLower}.}
    \label{fig:MainLowerHard}
\end{figure}

The constructions of our hard instances (as shown in \fig{MainLowerHard} (a)) are inspired by the idea that originally appeared in Ref.~\cite{jin2018local}. We first consider a ``scale free" version of function pair $(F,f)$, where we assume $\rho=1$ and $\epsilon=1$. Denote $\sin\x:=(\sin(x_1),...,\sin(x_d))$ and $\mathbb{I}(A)$ as the indicator function that has value $1$ when $A$ is true and $0$ otherwise. We set the constant $\mu=300$ and define the target function as
\begin{align}\label{eq:FDef}
F(\x):=h(\sin \x)+\norm{\sin\x}^2,
\end{align}
where $h(\x):=h_1(\v^\top\x)\cdot h_2\left(\sqrt{\norm{\x}^2-(\v^\top\x)^2}\right)$, and
\begin{align}
h_1(x)=g_1(\mu x),\quad g_1(x)&=(-16\abs{x}^5+48x^4-48\abs{x}^3+16x^2)\cdot\mathbb{I}\{\abs{x}<1\},\\
h_2(x)=g_2(\mu x),\quad g_2(x)&=(3x^4+8\abs{x}^3+6x^2-1)\cdot\mathbb{I}\{\abs{x}<1\}.
\end{align}
Here, the vector $\v$ is uniformly chosen from the $d$-dimensional unit sphere. In addition, we can split the domain into different regions upon which analysis and constructions are made separately:
\begin{itemize}
    \item ``ball" $\mathbb{B}(0,3/\mu)=\{\x\in\mathbb{R}^d:\norm{x}\leq 3/\mu\}$ is the $d$-dimensional hyperball with radius $3/\mu$.
    \item ``band" $S_\v=\{\x\in \mathbb{B}(0,3/\mu):\expval{\sin\x,\v}\leq\log d/\sqrt{d}\}$.
\end{itemize}
We provide the landscape of $F$ and the region division in \fig{MainLowerHard} (a). The above construction happens within a hyperball and we cannot fill the entire space $\mathbb{R}^d$ with hyperballs. Therefore, we embed this hyperball into a hypercube and add two regions.
\begin{itemize}
    \item ``hypercube" $H=[-\pi/2,\pi/2]^d$ is the $d$-dimensional hypercube with length $\pi$.
    \item ``padding" $S_2=H-\mathbb{B}(0,3/\mu)$.
\end{itemize}
With the above construction, we can fill the space $\mathbb{R}^d$ using these hypercubes. Meanwhile, the noisy function $f$ is defined as
\begin{align}\label{eq:fDefin}
f(\x)=\begin{cases}
\|\sin\x\|^2,\quad\x\in S_\v,\\
F(\x),\quad\x\notin S_\v.
\end{cases}
\end{align}
The ``band'' region $S_\v$ is known as the non-informative region as any query to $f$ in this area will obtain no information regarding $\v$. Intuitively, the metric of the non-informative area approaches $1$ as $d$ increases according to the measure of concentration. It is hard for any algorithm (both classical and quantum) to find a point out of this region. In particular, the probability of classically querying a point on $S_\v$ is bounded below by
\begin{align}\label{eq:AreaProportion}
\frac{\text{Area}(S_\v)}{\text{Area}(\mathbb{B}(0,3/\mu))}\geq1-O(d^{-\log d})
\end{align}
according to \lem{vector-from-sphere} in \append{existing-lemmas}.The following properties hold for function pair according to Ref.~\cite{jin2018local}.

\begin{lemma}[Lemma 33, Ref.~\cite{jin2018local}]\label{lem:PropHardZero}
The function pair $(F,f)$ defined in \eq{FDef} and \eq{fDefin} above satisfies:
\begin{itemize}
    \item The value of $f$ in the non-informative region $S_\v$ is independent of $\v$.
    \item $\sup_{x\in S_\v}\norm{f-F}_\infty\leq\tO(1/d)$.
    \item $F$ has no $\epsilon$-SOSP in the non-informative region $S_\v$.
    \item $F$ is $O(d)$-bounded, $O(1)$-Hessian Lipschitz, and $O(1)$-gradient Lipschitz.
\end{itemize}
\end{lemma}

The hard instance in the above \lem{PropHardZero} has realized the $\tO(1/d)$ factor for the noise bound in \thm{ZeroMeanLower} and introduced the non-informative area. The next step is to scale the hard instance to reach the lower bound with
correct dependencies on $\rho$ and $\epsilon$. Given $\epsilon>0$ and $\rho>0$, we define the scaling functions
\begin{align}\label{eq:FDefScale}
\tilde{F}(\x)&\coloneqq\epsilon rF\left(\frac{\x}{r}\right),\\\label{eq:fDefinScale}
\tilde{f}(\x)&\coloneqq\epsilon rf\left(\frac{\x}{r}\right),
\end{align}
where $r=\sqrt{\epsilon/\rho}$, and the functions $F,f$ are defined in \eq{FDef} and \eq{fDefin}, respectively. The scaled regions corresponding to $\tilde{F}$ and $\tilde{f}$ are:
\begin{itemize}
    \item ``ball" $\mathbb{B}(0,3r/\mu)=\{\x\in\mathbb{R}^d:\norm{x}\leq 3/\mu\}$ is the $d$-dimensional hyperball with radius $3r/\mu$.
    \item ``band" $\tilde{S}_\v=\{\x\in \mathbb{B}(0,3r/\mu):\expval{\sin(\x/r),\v}\leq\log d/\sqrt{d}\}$.
\end{itemize}
According to \lem{PropHardZero}, the function pair $(\tilde{F},\tilde{f})$ satisfies \assume{ZeroProb} with $\nu=\tTheta(\sqrt{\epsilon^3/\rho}\cdot 1/d)$, upon which we prove our quantum lower bound for finding an $\epsilon$-SOSP of the target function $\tilde{F}$ in \eq{FDefScale} taking queries to the noisy function $\tilde{f}$ in \eq{fDefinScale}. Formally, we provide the following theorem.
\begin{theorem}[Formal version of \thm{ZeroMeanLower}]\label{thm:ZeroMeanLowerF}
For any $B>0,\ell>0,\rho>0$, there exists an $\epsilon_0=\Theta(\min\{\ell^2/\rho,(B^2\rho/d^2)^{1/3}\})$ such that for any $\epsilon\in(0,\epsilon_0]$, the function pair $(\tilde{F},\tilde{f})$ defined in \eq{FDefScale} and \eq{fDefinScale} satisfies \assume{ZeroProb} with $\nu=\tTheta(\sqrt{\epsilon^3/\rho}\cdot 1/d)$, and any \textit{quantum} algorithm that only queries a quasi-polynomial $O(d^{\log(d)})$ times to the zeroth-order quantum oracle $U_{\tilde{f}}$ will fail with high probability to find an $\epsilon$-SOSP of $\tilde{F}$.
\end{theorem}
To prove \thm{ZeroMeanLowerF}, we introduce some lemmas to construct a reduction of the problem. In particular, our goal is to transform the quantum lower bound on the unstructured search problem~\cite{Bennett1997Strengths,Nayak1999Quantum} into a lower bound for the problem of finding an $\epsilon$-SOSP of $\tilde{F}$ considered in \thm{ZeroMeanLowerF}. We discretize the problem via the following results on distributing exponentially many points on $\mathbb{S}^{d-1}$ in a uniform way such that the distances between each pair of points are at least $\delta$.
\begin{lemma}[Lemma D.1, Ref.~\cite{liu2022quantum}]\label{lem:DividingSphere}
For any constant $\delta\in\big(0,\frac{\log d}{2\sqrt{d}}\big)$, there exists a set $\Gamma=\{\y_1,\y_2,\ldots,\y_N\}$ of $N$ unit vectors in $\mathbb{R}^d$ such that
\begin{itemize}
\item $\forall \y_i\neq\y_j\in\Gamma$, $\|\y_i-\y_j\|\geq\delta$;
\item $\forall \z\in\mathbb{S}^{d-1}$, there exists an $\y_i\in\Gamma$ such that $\|\z-\y_i\|\leq\delta$;
\item $N\geq\big(\frac{1}{2\delta}+\frac{1}{2}\big)^d-\big(\frac{1}{2\delta}-\frac{1}{2}\big)^d$.
\end{itemize}
\end{lemma}

Inspired by Ref.~\cite{liu2022quantum}, we consider the following unstructured search problem which can be reduced to finding an $\epsilon$-SOSP of $\tilde{F}$ with polynomial overhead. 
\begin{problem}\label{prob:UnstructuredSearch}
Consider a set $\Gamma$ of $N$ unit vectors in $\mathbb{R}^d$ satisfying \lem{DividingSphere}, for an unknown unit vector $\v\in\Gamma$, we define $q\colon\Gamma\to\mathbb{R}^d$ as follows:
\begin{align}
q(\x)\coloneqq\begin{cases}
\v,\quad\langle \sin\x,\v\rangle > \frac{\log d}{2\sqrt{d}},\\
\vect{0},\quad\text{otherwise}.
\end{cases}
\end{align}
The goal is to find $\v$ only with access to values of $q$.
\end{problem}

We now present the reduction from \prob{UnstructuredSearch} to the problem of finding an $\epsilon$-SOSP of $\tilde{F}$ under the setting of \thm{ZeroMeanLowerF}. To make the reduction more straightforward, we additionally introduce an intermediate function $\hat{q}(\x)\colon\mathbb{B}(\0,3/\mu)\to\mathbb{R}$ between $q$ and $f$. In particular, for any $\x\in\mathbb{B}(\0,3/\mu)$, we use $\hat{\y}(\x)$ to denote the vector $\y_i$ in $\Gamma$ such that the distance $\|\x/\|\x\|-\y_i\|$ is minimized. If more than one of such vectors exists, we choose the one with the smallest lower index. We define $\hat{q}(\x)$ as
\begin{align}\label{eq:g-superpoly}
\hat{q}(\x)\coloneqq\begin{cases}
\|\sin\x\|^2,\quad q(\hat{\y}(\x))=\0,\\
F(\x),\quad \text{otherwise}.
\end{cases}
\end{align}

Similar to $f$, $\hat{q}$ also has a large ``non-informative" region $\hat{S}_\v$ where the function value equals $\norm{\sin\x}^2$ and reveals no information about $\v$. Quantitatively, we can observe that $\hat{S}_\v=\{\y\in\mathbb{B}(\0,3/\mu):g(\hat{\y}(\x))=\0\}$, and $\hat{q}$ has the following properties.
\begin{lemma}\label{lem:hatg-properties}
The function $\hat{q}(\x)$ defined in \eq{g-superpoly} has the following properties:
\begin{itemize}
\item One query to $\hat{q}$ can be implemented using one query to $q$.
\item Its non-informative region $\hat{S}_\v=\{\y\in\mathbb{B}(\0,3/\mu):q(\hat{\y}(\x))=\vect{0}\}$ is a subset of $S_\v$, which is the non-informative region of $f$ defined in \eq{fDefin}. 
\item For any $\x\in\mathbb{B}(\0,3/\mu)-S_\v$, we have $\hat{q}(\x)=f(\x)$.
\end{itemize}
\end{lemma}
\begin{proof}
For the first property, one can observe that for any $\x\in\mathbb{B}(\0,3/\mu)$, $\hat{q}(\x)$ can be expressed as
\begin{align}
\hat{q}(\x)&=\norm{\sin\x}^2+h(\sin\x)\\
&=\norm{\sin\x}^2+h_1(\langle q(\hat{\y}(\x)),\x\rangle)\cdot h_2\big(\sqrt{\norm{\x}^2-\langle q(\hat{\y}(\x)),\x\rangle^2}\big),
\end{align}
which can be implemented using one query to $q(\x)$.

For the second property, $\forall\x\in\hat{S}_\v$, the corresponding $\hat{\y}(\x)$ satisfies $\langle\hat{\y}(\x),\v\rangle\leq\frac{\log d}{2\sqrt{d}}$. Since $\|\x/\|\x\|-\hat{y}(\x)\|\leq\delta=\frac{\log d}{2\sqrt{d}}$ by \lem{DividingSphere}, we deduce that $\langle\x,\v\rangle\leq\langle\x/\|\x\|,\v\rangle\leq\log d/\sqrt{d}$, indicating $\x\in S_\v$.

The third property can be directly obtained from the second property.
\end{proof}

Next, we present the reduction from \prob{UnstructuredSearch} to the problem of finding an $\epsilon$-SOSP of $F$ under the setting of \thm{ZeroMeanLowerF} through the following lemma.
\begin{lemma}\label{lem:reduction-to-hyperball}
Under the setting of \thm{ZeroMeanLowerF}, with polynomial overhead \prob{UnstructuredSearch} can be reduced to the problem of finding an $\epsilon$-SOSP of $F$ defined in \eq{FDef} for any $\epsilon\in(0,\epsilon_0]$, using access to values of $f$ defined in \eq{fDefin}.
\end{lemma}
\begin{proof}
Since one query to $\hat{q}$ can be implemented using one query to $q$ by \lem{hatg-properties}, \prob{UnstructuredSearch} can be reduced to the problem of finding an $\x\in\mathbb{B}(0,3/\mu)$ satisfying $\hat{q}(\x)\neq\norm{\sin\x}^2$, or equivalently, $\x\in\mathbb{B}(0,3/\mu)-\hat{S}_\v$, with only access to values of $\hat{q}$.

By \lem{hatg-properties}, any $\epsilon$-SOSP of $F$, denoted $\x^F_{\text{SOSP}}$, satisfies
\begin{align}
\x^F_{\text{SOSP}}\in \mathbb{B}(0,3/\mu)-S_\v\subseteq\mathbb{B}(0,3/\mu)-\hat{S}_\v.
\end{align}
Therefore, we have reduced \prob{UnstructuredSearch} to the nonconvex optimization task.
\end{proof}

We can scale the ``scale-free" hard instance $(F,f)$ to the hard instance $(\tilde{F},\tilde{f})$ satisfying \assume{ZeroProb} using \eq{FDefScale} and \eq{fDefinScale}. In particular, we introduce the following lemma, which originally appeared in Ref.~\cite{jin2018local}.
\begin{lemma}[Appendix C.2, Ref.~\cite{jin2018local}]\label{lem:extending-to-Rd}
For any $B>0,\ell>0,\rho>0$, there exists an $\epsilon_0=\Theta(\min\{\ell^2/\rho,(B^2\rho/d^2)^{1/3}\})$ such that for any $\epsilon\in(0,\epsilon_0]$, there exists a function pair $(\tilde{F},\tilde{f})$ satisfying the assumptions in \assume{ZeroProb} with $\nu=\tilde{\Theta}(\sqrt{\epsilon^3/\rho}\cdot(1/d))$, so that with constant overhead the problem of finding an $\epsilon$-SOSP of $F$ defined in \eq{FDef} using only access to values of $f$ defined in \eq{fDefin} can be reduced to the problem of finding an $\epsilon$-SOSP of $\tilde{F}$ using only access to values of $\tilde{f}$.
\end{lemma}

Equipped with \lem{extending-to-Rd},, we prove \thm{ZeroMeanLowerF}
\begin{proof}
According to \lem{reduction-to-hyperball} and \lem{extending-to-Rd}, there exists a function pair $(\tilde{F},\tilde{f})$ satisfying \assume{ZeroProb} with $\nu=\tilde{\Theta}(\sqrt{\epsilon^3/\rho}\cdot(1/d))$, such that with polynomial overhead \prob{UnstructuredSearch} can be reduced to the problem of finding an $\epsilon$-SOSP of $\tilde{F}$ using only access to values of $\tilde{f}$.

We divide the $N$ unit vectors in $\Gamma$ into two parts.
\begin{align}
\Gamma_1&=\{\x\in\Gamma\colon q(\x)=\v\}\\
\Gamma_0&=\{\x\in\Gamma\colon q(\x)=\0\}.
\end{align}
We denote the size of the two parts as $N_1=\abs{\Gamma_1}$ and $N_0=\abs{\Gamma_2}$. Our goal is to find any $\x$ in the set $\Gamma_1$. Intuitively, under limitation $\delta\to 0$ and $\delta<\frac{\log d}{2\sqrt{d}}$, we can deduce that
\begin{align}
\frac{N_1}{N}\sim\frac{2\text{Area}(\mathbb{B}(0,3/\mu)-S_\v)}{\text{Area}(\mathbb{B}(0,3/\mu))}=O(d^{-\log d}).
\end{align}
We bound the deviation of $N_1/N$ from $2\text{Area}(\mathbb{B}(0,3/\mu)-S_\v)/\text{Area}(\mathbb{B}(0,3/\mu))$ when $\delta\neq 0$. For $S_\v$, we consider $S_\v'$ the area that is the "band" area along $\v$ within $\log d/2\sqrt{d}-\delta$ from $\0$. The border area of $S_\v$ out of $S_\v'$ contains ignorable $O(\exp(-d))$ directions compared to $S_\v$ when $\delta\ll 1/\sqrt{d}$. 

Even if we consider the boundary above, we can still derive the upper bound for $N_1/N$
\begin{align}
\frac{N_1}{N}\leq\frac{2\text{Area}(\mathbb{B}(0,3/\mu)-S_\v')}{\text{Area}(\mathbb{B}(0,3/\mu))}=O(d^{-\log d}).
\end{align}
The inequality comes from the fact that $\delta\ll\poly(1/d)$ and the boundary area can only bring exponential deviation from the expectation value $2\text{Area}(\mathbb{B}(0,3/\mu)-S_\v)/\text{Area}(\mathbb{B}(0,3/\mu))$. 

With the fact that $N_1/N\leq O(d^{-\log d})$ is quasi-polynomially small, any quantum algorithm that solves \prob{UnstructuredSearch} with high probability requires query complexity at least $\Omega(\sqrt{N/N_1})=\Omega(d^{\log d})$~\cite{Nayak1999Quantum,Nielsen2010Quantum}.
\end{proof}

%=====================================================================

\subsection{Information-theoretic Limitation of Quantum Zeroth-order Methods}\label{sec:DLowerZeroNo}

When the noise between $F$ and $f$ keeps increasing under \assume{ZeroProb}, it can erase the landscape of target function $F$ in the worst case. As a result, when the noise rate is larger than a certain threshold, for any quantum algorithm we can find a hard instance on which it will fail with a large probability. We consider the same target function defined in Eq.~\eq{FDef} with a different noisy function $f$. We apply the scaling in \eq{FDefScale} and \eq{fDefinScale} as
\begin{align}
&F(\x)=h(\sin\x)+\norm{\sin\x}^2,\\
&f(\x)=\norm{\sin\x}^2,\\\label{eq:FfDefinNo}
&\tilde{F}(\x)=\epsilon rF\left(\frac{\x}{r}\right),\qquad
\tilde{f}(\x)=\epsilon rf\left(\frac{\x}{r}\right).
\end{align}
Similar to \lem{PropHardZero}, the following properties hold for the above hard instance $(\tilde{F},\tilde{f})$.
\begin{lemma}[Appendix D.2, Ref.~\cite{jin2018local}]\label{lem:PropHardZeroNo}
The function pair $(F,f)$ defined in \eq{FfDefinNo} above satisfies:
\begin{itemize}
    \item The values of $\tilde{f}$ in $\mathbb{B}(0,3/\mu)$ are independent of $\v$.
    \item $\sup_{x\in\mathbb{B}(0,3/\mu)}\norm{\tilde{f}-\tilde{F}}_\infty\leq\tO(\epsilon^{1.5}/\sqrt{\rho})$.
    \item $F$ is $B$-bounded, $O(\rho)$-Hessian Lipschitz, and $O(\ell)$-gradient Lipschitz.
\end{itemize}
\end{lemma}

We derive the following result concerning the hard instance $(\tilde{F},\tilde{f})$ in Eq.~\eq{FfDefinNo}, which is the formal version of \thm{ZeroNoLower}.
\begin{theorem}[Formal version of \thm{ZeroNoLower}]\label{thm:ZeroNoLowerF}
For any $B>0,\ell>0,\rho>0$, any $\epsilon\in(0,\epsilon_0]$ for some $\epsilon_0=\Theta(\min\{\ell^2/\rho,(B^2\rho/d^2)^{1/3}\})$, and any possible \textit{quantum} algorithm, we can choose a function pair $(\tilde{F},\tilde{f})$ with $\v$ defined in \eq{FfDefinNo} satisfying \assume{ZeroProb} with $\nu=\tTheta(\sqrt{\epsilon^3/\rho})$ such that the quantum algorithm will fail with high probability to find an $\epsilon$-SOSP of $\tilde{F}$ given only access to $U_{\tilde{f}}$.
\end{theorem}
\begin{proof}
As $\tilde{f}$ is independent of $\v$, neither quantum nor classical query can reveal any information on $\v$ and the $\epsilon$-SOSP of $\tilde{F}$. Any solutions output by any algorithm will be independent of $\v$ with probability $1$. Therefore, the probability of success must be independent of the number of iterations, which indicates that any algorithm cannot output an $\epsilon$-SOSP with probability more than a constant. Specifically, no algorithm can do better than random guessing in $\mathbb{B}(0,3/\mu)$ within this construction.
\end{proof}

We remark that the noise bound and its underlying intuition in the quantum case is the same as the classical case~\cite{jin2018local}. However, \thm{ZeroNoLowerF} only indicates the classical and quantum algorithms have the same worst-case lower bound for some level of noise strength, and there is still a possible quantum speedup for solving specific instances $(F,f)$. In \append{QTW}, we show a concrete example in which quantum tunneling walk~\cite{liu2022quantum} can find an $\epsilon$-SOSP of $F$ using polynomial queries and proper initial states containing information of the landscape, while any classical algorithm requires exponential queries even given access to such information.

\subsection{Proof of \thm{ZeroClassLowerJordan}}\label{sec:DLowerClassical}
Here, we use an information-theoretic approach to prove \thm{ZeroClassLowerJordan}, which indicates that if a classical algorithm can find an $\epsilon$-SOSP of $F$ for any function pair $(F,f)$ satisfying \assume{ZeroProb} with $\nu=\Omega(1/\poly(d))$, the query complexity is bounded by $\Omega(d/\log d)$. 

In particular, we consider the target function $\tilde{F}$ defined by \eq{FDef} and \eq{FDefScale}. As $\tilde{F}$ is $\rho$-Hessian Lipschitz and $\ell$-smooth, we can estimate the vector $\v$ within $\poly(\epsilon)$ distance under infinity norm\footnote{Here, we have ignored the dependence on $\ell$ and $\rho$ and regarded them as constants}, which requires $d\log(1/\epsilon)$ bits of information. Furthermore, as the noisy zeroth-order oracle $\tilde{f}$ contains noise of strength $\nu=\Omega(1/\poly(d))$, each classical query reveals at most $O(\log(1/\delta))$ bits of information~\cite{chakrabarti2020optimization}. Therefore, any classical algorithm has to take at least
\begin{align}
\Omega\left(\frac{d\log(1/\epsilon)}{\log(1/\delta)}\right)=\Omega\left(\frac{d}{\log d}\right)
\end{align}
queries to the noisy oracle.

Moreover, Ref.~\cite{chakrabarti2020optimization} shows even estimating a sub-gradient for a Lipschitz convex function within infinity norm $\epsilon=O(1/\poly(d))$ using zeroth-order oracle $f$ with noise rate $\nu=\Omega(1/\poly(d))$ requires $\Omega(d/\log d)$ classical queries. However, Jordan's algorithm enables simultaneous queries to different points using a single oracle query. Thus, only $O(1)$ query to quantum oracle is required to estimate the sub-gradient, which provides the exponential speedup for quantum algorithms.

\subsection{Lower Bound for First-order Methods}\label{sec:DLowerFirst}
We now derive the lower bound for classical and quantum algorithms in finding $\epsilon$-SOSPs of target function $F$ through noisy function satisfying \assume{FirstProb}. We propose the following two theorems as two parts for the formal version of \thm{FirstLower}. In the first part, we consider the case of adding noise such that for any quantum or classical algorithm we can find a hard instance that will make the algorithm fail with high probability, which is an analog of \thm{ZeroNoLowerF} under \assume{FirstProb}. Formally, we have the following theorem 
\begin{theorem}[Formal version of \thm{FirstLower}, Part I]\label{thm:FirstLowerF1}
For any $B>0,\ell>0,\rho>0$, any $\epsilon\in(0,\epsilon_0]$ for some $\epsilon_0=\Theta(\min\{\ell^2/\rho,(B^2\rho/d^2)^{1/3}\})$, and any \textit{quantum} or \textit{classical} algorithm, we choose function pair $(\tilde{F},\tilde{f})$ with $\v$ defined in \eq{FfDefinNo} satisfies \assume{FirstProb} with $\tnu=\tTheta(\epsilon/\sqrt{d})$ such that it will fail with high probability to find any $\epsilon$-SOSP of $\tilde{F}$ given only access to $U_\g$.
\end{theorem}
\begin{proof}
We consider the same target function $\tilde{F}$ and noisy function $\tilde{f}$ in \eq{FfDefinNo}. Except from the properties in \lem{PropHardZeroNo}, the noisy function $\tilde{f}$ is also smooth and $\norm{\nabla \tilde{f}-\nabla \tilde{F}}\leq \tTheta(\epsilon/\sqrt{d})$. Therefore, the hard instance in \eq{FfDefinNo} also satisfies \assume{FirstProb} with $\tnu=\tTheta(\epsilon/\sqrt{d})$. According to \thm{ZeroNoLowerF}, for any \textit{quantum} and \textit{classical} algorithm, we choose function pair $(\tilde{F},\tilde{f})$ with $\v$ such that it will success with probability no more than a constant to find any $\epsilon$-SOSP of $\tilde{F}$ given only access to $U_{\tilde{\g}}$.
\end{proof}

Next, we consider the quasi-polynomial lower bound under \assume{FirstProb}. Unlike \thm{ZeroMeanLowerF}, we cannot directly apply the hard instance $(F,f)$ defined in \eq{FDefScale} and \eq{fDefinScale} because $f$ is not differentiable (or more strictly, not continuous). To address this problem, we construct a different noisy function $f$ (as shown in \fig{MainLowerHard} (b)). We start with the ``scale free" version. We still set the $\mu=300$ and define the target function $F(\x)=h(\sin\x)+\norm{\sin\x}^2$, which is the same with \eq{FDef}. We uniformly choose $\v$ and divide the ``hypercube" into different regions as 
\begin{itemize}
    \item ``hypercube" $H=[-\pi/2,\pi/2]^d$ is the $d$-dimensional hypercube with length $\pi$.
    \item ``ball" $\mathbb{B}(0,3/\mu)=\{\x\in\mathbb{R}^d:\norm{x}\leq 3/\mu\}$ is the $d$-dimensional ball with radius $3/\mu$.
    \item ``band" $S=\{\x\in \mathbb{B}(0,3/\mu):\expval{\sin\x,\v}\leq w\}$ with $w=O(\log d/\sqrt{d})$.
    \item ``non-informative band" $S_\v=\{\x\in \mathbb{B}(0,3/\mu):\expval{\sin\x,\v}\leq 0.9w\}$.
    \item ``padding" $S_2=H-\mathbb{B}(0,3/\mu)$.
\end{itemize}
Meanwhile, the noisy function $f$ is defined as
\begin{align}\label{eq:fDefinFirst}
f(\x)=\begin{cases}
\|\sin\x\|^2,\quad\x\in S_\v,\\
\|\sin\x\|^2+h_3(\x)\cdot h_2(\sqrt{\norm{\sin\x}^2-(\v^\top\sin\x)^2}),\quad\x\in S-S_\v\\
F(\x),\quad\x\notin S,
\end{cases}
\end{align}
where
\begin{align}
h_3(\x)=h_1(\v^\top\sin(10\x-9w\v/2)).
\end{align}
By the chain rule of gradients we deduce the following lemma:
\begin{lemma}\label{lem:PropHardFirst}
The function pair $(F,f)$ defined in \eq{FDef} and \eq{fDefinFirst} satisfies:
\begin{itemize}
    \item The value of $f$ in the non-informative region $S_\v$ is independent of $\v$.  $f$ is differentiable and satisfies the Lipshitz condition.
    \item $\sup_{x\in S}\norm{\nabla f-\nabla F}_\infty\leq\tO(1/\sqrt{d})$.
    \item $F$ has no $\epsilon$-SOSP in the non-informative region $S_\v$.
    \item $F$ is $O(d)$-bounded, $O(1)$-Hessian Lipschitz, and $O(1)$-gradient Lipschitz.
\end{itemize}
\end{lemma}
We apply the scaled version of $f$ as
\begin{align}\label{eq:fDefinFirstScale}
\tilde{f}=\epsilon r f\left(\frac{\x}{r}\right).
\end{align}
Based on the hard instance $(\tilde{F},\tilde{f})$ defined in \eq{FDefScale}, \eq{fDefinFirst} and \eq{fDefinFirstScale}. We propose the following theorem.
\begin{theorem}[Formal version of \thm{FirstLower}, Part II]\label{thm:FirstLowerF2}
For any $B>0,\ell>0,\rho>0$, there exists $\epsilon_0=\Theta(\min\{\ell^2/\rho,(B^2\rho/d^2)^{1/3}\})$ such that for any $\epsilon\in(0,\epsilon_0]$, the function pair $(\tilde{F},\tilde{f})$ defined in \eq{FDefScale}, \eq{fDefinFirst} and \eq{fDefinFirstScale} satisfies \assume{FirstProb} with $\tnu=\tTheta(\epsilon/\sqrt{d})$, and any \textit{quantum} or \textit{classical} algorithm that only requires a quasi-polynomial number $\Omega(d^{\log(d)})$ queries to function values of $\tilde{f}$ will fail with high probability, to find an $\epsilon$-SOSP of $\tilde{F}$.
\end{theorem}
\begin{proof}
According to \lem{PropHardFirst} and \lem{extending-to-Rd}, the width of non-informative band is $0.9w=\tTheta(\epsilon/\sqrt{d})$. Then we can directly apply the similar procedure when we prove \thm{ZeroMeanLower} as the non-informative band has the same width scaling.
\end{proof}

%===============================================================================================================================

\section{Classical and Quantum Lower Bounds in $\epsilon$}\label{sec:EpsLower}
In this section, we prove classical randomized lower bounds and quantum lower bounds in $\epsilon$ for finding an $\epsilon$-SOSP of an objective function $F$ given access to noisy classical or quantum zeroth- or first-order oracles, where $F\colon\R^d\to\R$ is $\ell$-smooth and $\rho$-Hessian Lipschitz, and satisfies
\begin{align}
F(\0)-\inf_\x F(\x)\leq\Delta,
\end{align}
for some constant $\Delta$.

\subsection{Hard Instance for Deterministic Classical Algorithms}\label{sec:classical-det}

We first discuss the construction and intuition of hard instances upon which we can obtain lower bound results for deterministic classical algorithms. Consider the toy example proposed by Nesterov~\cite[Chapter 2.1.2]{nesterov2003introductory},
\begin{align}\label{eqn:nesterov}
F(\x)\coloneqq\frac{1}{2}(x_1-1)^2+\frac{1}{2}\sum_{i=1}^{T-1}(x_i-x_{i+1})^2,
\end{align}
whose gradient satisfies that
\begin{align}
\forall 1<i<T,\quad\nabla_i F(\x)=\0\Leftrightarrow x_{i-1}=x_i=x_{i+1}.
\end{align}
Then, if we query the gradient of $F$ at a point with only its first $t$ entries being nonzero, the derivative can only reveal information about the $(t+1)$th direction, if one does not have knowledge about the directions of the coordinate axes. Formally, such properties can be summarized to consist of the concept of \textit{zero-chain}, which is defined as follows.
\begin{definition}[Definition 3, Ref.~\cite{carmon2020lower}]\label{defn:zero-chain}
A function $F\colon\mathbb{R}^d\to\mathbb{R}$ is called a zero-chain if for every $\x\in\R^d$,
\begin{align}
\supp\{\x\}\subseteq\{1,\ldots,i-1\}\Rightarrow\supp\{\nabla f(\x)\}\subseteq\{1,\ldots,i\},
\end{align}
where the support of a vector $\y\in\R^d$ is defined as
\begin{align}
\supp\{\y\}\coloneqq\{i\in[d]\,|\,y_i\neq 0\}.
\end{align}
\end{definition}

From an algorithmic perspective, if we encode $F(\x)$ or any other $T$-dimensional zero chain into a $d$-dimensional space with $d\gg T$ and apply a random rotation $U$, any deterministic algorithm making fewer than $T$ queries will fail on certain instances to find the directions of all the $T$ axes. Hence, from an algorithmic perspective, if we can construct a $T$-dimensional zero chain with all SOSPs or even FOSPs overlapped with all the $T$ axes, we can establish an $\Omega(T)$ lower bound for all deterministic classical algorithms.

Following this intuition, Ref.~\cite{carmon2021lower} provided a concrete hard instance construction to obtain an $\Omega(1/\epsilon^2)$ lower bound for deterministic classical algorithms. In particular, Ref.~\cite{carmon2021lower} first defined the following zero-chain $\bar{F}_{T;\mu}(\x)\colon\R^{T+1}\to\R$:
\begin{align}\label{eqn:FT-defn}
\bar{F}_{T;\mu}(\x)=\frac{\sqrt{\mu}}{2}(x_1-1)^2+\frac{1}{2}\sum_{i=1}^{T}(x_{i+1}-x_i)^2+\mu\sum_{i=1}^T\Gamma(x_i),
\end{align}
where the non-convex function $\Gamma\colon\R\to\R$ is defined as
\begin{align}
\Gamma(x)=120\int_1^x\frac{t^2(t-1)}{1+t^2}\d t.
\end{align}
According to \lem{FT-large-gradient} in \append{existing-lemmas}, finding an FOSP requires knowledge about the directions of all the $T$ axes. We further apply a unitary rotation $U\in\R^{(T+1)\times d}$ and certain appropriate scaling to obtain the formal hard instance
\begin{align}\label{eqn:tildeF-defn}
\tilde{F}_{T;U}(\x)\coloneqq\lambda\sigma^2\bar{F}_T\big(\<\x/\sigma,\u^{(1)}\>,\<\x/\sigma,\u^{(2)}\>,\ldots,\<\x/\sigma,\u^{(T+1)}\>\big), 
\end{align}
where $\u^{(i)}$ stands for the $i$-th column of the rotation matrix $U$, and all its columns $\{\u^{(1)},\ldots,\u^{(T+1)}\}$ forms a set of orthonormal vectors. We use $\tilde{\mathcal{F}}$ to denote the set of functions that can be presented in the form of \eqn{tildeF-defn} for some suitable parameters $T,U,\lambda,$ and $\sigma$ whose function value at point $\0$ is not far from its minimum value, i.e.,
\begin{align}
F(\0)-\inf_\x F(\x)\leq\Delta,\quad\forall F\in\tilde{\mathcal{F}}.
\end{align}
Based on $\tilde{\mathcal{F}}$, we have the following classical lower bound result.
\begin{lemma}[Theorem 2, Ref.~\cite{carmon2021lower}]\label{lem:classical-deterministic-lower}
There exist numerical constants $c,C\in\R_{+}$ and $\ell_q\leq e^{\frac{3q}{2}\log q+Cq}$ for every $q\in\mathbb{N}$ such that, for any deterministic classical algorithm making at most
\begin{align}
c\cdot\Delta\Big(\frac{L_1}{\ell_1}\Big)^{\frac{3}{7}}\Big(\frac{L_2}{\ell_2}\Big)^{\frac{2}{7}}\epsilon^{-12/7}
\end{align}
gradient queries, there exists a function $\tilde{F}\in\tilde{\mathcal{F}}$ such that the output of this algorithm on $\tilde{F}$ is not an $\epsilon$-FOSP of $\tilde{F}$.
\end{lemma}

This lower bound regarding deterministic classical algorithms is, however, hard to be extended to randomized classical algorithms straightforwardly. Intuitively, the concept of zero-chain in \defn{zero-chain} can be extended to higher-order derivatives, and the hard instance $\tilde{F}_{T;U}$ in \eqn{tildeF-defn} is no longer a zero chain for derivatives of second- or higher-orders. Hence, the algorithm may benefit from adding random perturbations and may not need to discover all the $T+1$ components $\{\u^{(1)},\ldots,\u^{(T+1)}\}$ one by one. To the best of our knowledge, it remains unclear whether the same lower bound result holds for randomized classical algorithms.

In the remaining part of this section, we will demonstrate that the presence of noise can drastically increase the hardness of finding an $\epsilon$-FOSP in the worst case. Specifically, we first derive the lower bound result for general noise models parameterized by the concept of noise radius $r_0$. Then, we discuss the values of $r_0$ in different settings with noisy zeroth-order oracle (\assume{ZeroProb}) or noisy first-order oracle (\assume{FirstProb}), respectively.

\subsection{Noisy Quantum Lower Bound with Bounded Input Domain}\label{sec:noisy-lower}

In this subsection, we first introduce the quantum lower bound on functions with bounded input domains. The intuition is that the noise can create a non-informative region around $\0$, which is a hyperball $\mathbb{B}(0,r_0)$ whose certain radius $r_0$ depends on the noise rate. Then, if the dimension of $\tilde{F}_{T;U}$ defined in Eq.~\eqn{tildeF-defn} is large enough, any random perturbation with bounded norm will fall in $\mathbb{B}(0,r_0)$ with an overwhelming probability, which leads to the fact that the lower bound in \lem{classical-deterministic-lower} additionally holds for not only randomized classical algorithms but also quantum algorithms.

We adopt the quantum query model introduced in \cite{garg2020no}. For a $d$-dimensional objective function $F$, assume we have access to its noisy evaluation $f\colon\R^d\to\R$ via the following quantum oracle $O_f$,
\begin{align}\label{eqn:FOSP-oracle}
O_f\ket{\x}\ket{\y}\to\ket{\x}\ket{\y\oplus\big(f(\x),\nabla f(\x)\big)}.
\end{align}
We remark that the oracle $O_f$ here is even stronger than the zeroth- or the first-order oracles in \eq{QZeroOracle} and \eq{QFirstOracle}. Then, any quantum algorithm $A_{\quan}$ making $q$ queries to $O_f$ can be described by the following sequence of unitaries
\begin{align}
V_qO_fV_{q-1}O_f\cdots V_1O_fV_0
\end{align}
applied to some initial state, say $\ket{0}$ without loss of generality. In the special case where the objective function $\tilde{F}_{T;U}\in\tilde{F}$, for the convenience of notation we denote $O_{f;T;U}$ to be the quantum oracle encoding its noisy evaluation $f_{r_0;T;U}$ in the form of  \eqn{FOSP-oracle}. To obtain our quantum lower bound, we set the noisy function to be in the form
\begin{align}\label{eqn:f-r0-defn}
f_{r_0;T;U}(\x)\coloneqq\lambda\sigma^2\bar{F}_T\big(\<\x/\sigma,\u^{(1)}\>,\ldots,\<\x/\sigma,\u^{(\prog_{\sigma r_0}(\x)+1}\>,0,\ldots,0\big),
\end{align}
where $\prog_{\sigma r_0}(\x)$ is defined as the largest index $j$ between $1$ and $T$ satisfying $|\<\x,\u^{(j)}\>|\geq \sigma r_0$. Moreover, we define the following indicator function
\begin{align}\label{eqn:delta-r0-defn}
\delta_{r_0}(y)\coloneqq\mathbbm{1}\{|y|\geq r_0\}\cdot y.
\end{align}
Intuitively, in the noisy function $f_{r_0;T;U}$ we eliminate the influence of the $i$th component on the function when the overlap between $\x$ and $\u^{(i)}$ is smaller than certain threshold $r_0$. The detailed values of $r_0$ under different noise assumptions will be specified later. Hence, when the dimension $d$ is large enough, any random perturbation with a bounded norm will make no observable difference with an overwhelming probability. Moreover, we can note that $\epsilon$-FOSPs of $f_{r_0;T;U}(\x)$ and $\tilde{F}_{T;U}(\x)$ are the same. Hence, one needs to identify all the $T$ components $\{\u^{(1)},\ldots,\u^{(T)}\}$ to find an $\epsilon$-FOSP, which we demonstrate later that can only be done sequentially even by a quantum algorithm.

For any possible quantum algorithm $A_{\quan}$ making $k<T$ queries in total, adopting a similar technique introduced in~\cite{garg2020no,garg2021near}, we define a sequence of unitaries starting with $A_0=A_{\quan}$ as follows:
\begin{align}\label{eqn:unitary-sequences}
A_0:&=V_{k}O_{f;T;U}V_{k-1}O_{f;T;U}\cdots O_{f;T;U}V_1O_{f;T;U}V_0\\ \nonumber
A_1:&=V_{k}O_{f;T;U}V_{k-1}O_{f;T;U}\cdots O_{f;T;U}V_1O_{f;1;U_1}V_0\\ \nonumber
A_2:&=V_{k}O_{f;T;U}V_{k-1}O_{f;T;U}\cdots O_{f;2;U_2}V_1O_{f;1;U_1}V_0\\ \nonumber
&\vdots\\ \nonumber
A_{k}:&=V_{k}O_{f;k;U_k}V_{k-1}O_{f;k-1;U_{k-1}}\cdots O_{f;2;U_2}V_1O_{f;1;U_1}V_0,
\end{align}
where $U_t\in\R^{d\times t}$ is defined as the orthogonal matrix with columns $\u^{(1)},\ldots,\u^{(t)}$, and the function $f_{r_0;t;U_t}(\x)$ encoded in $O_{f;t;U_t}$ is defined as
\begin{align}
    f_{r_0;t;U_t}\coloneqq\lambda\sigma^2\bar{F}_T\big(\<\x/\sigma,\u^{(1)}\>,\ldots,\<\x/\sigma,\u^{(\prog_{\sigma r_0}^t(\x)+1}\>,0,\ldots,0\big),
\end{align}
where $\prog_{\sigma r_0}^t(\x)$ is defined as the largest index $j$ between $1$ and $t-1$ satisfying $|\<\x,\u^{(j)}\>|\geq \sigma r_0$. Our goal is to demonstrate that $A_0$ will fail to find an $\epsilon$-FOSP with high probability. To do so, we employ a hybrid argument showing that the outputs of $A_{i}$ and $A_{i+1}$ defined in the sequence \eqn{unitary-sequences} are close for every $i<k$, so does the outputs of $A_0$ and $A_k$, which cannot solve the problem with high probability since it contains no information of the $T$-th component, which is necessary for finding an $\epsilon$-FOSP with high success probability. 

\begin{lemma}[$A_t$ and $A_{t-1}$ have similar outputs]\label{lem:similar-outputs}
Consider the hard instance $\tilde{F}_{T;U}(\x)\colon\R^d\to\R$ defined in \eqn{FT-defn} with domain $\mathbb{B}(0,2\sigma\sqrt{T})$ and its noisy evaluation $f_{r_0;T;U}$ defined in \eqn{f-r0-defn} with $d\geq 4T$, let $A_t$ for $t\in[k-1]$ be the unitaries defined in Eq.~\eqn{unitary-sequences}. Then
\begin{align}
\E_U\big(\|A_t\ket{\0}-A_{t-1}\ket{\0}\|^2\big)\leq  8Te^{-dr_0^2/(4T)}.
\end{align}
\end{lemma}
\begin{proof}
From the definition of the unitaries in Eq.~\eqn{unitary-sequences} and the unitary invariance of the spectral norm, we have
\begin{align}
\|A_t\ket{\0}-A_{t-1}\ket{\0}\|=\big\|(O_{f;t;U_t}-O_{f;T;U})V_{t-1}O_{f;t-1;U_{t-1}}\cdots O_{f;1;U_1}V_0\ket{\0}\big\|.
\end{align}
We will prove the claim for any fixed choice of vectors $\{\u^{(1)},\ldots,\u^{(t-1)}\}$, which will imply the claim for any distribution over those vectors. Let us prove the claim for any fixed choice of vectors $\{\u^{(1)},\ldots,\u^{(t-1)}\}$, which will imply the claim for any distribution over those vectors. Once we have fixed these vectors, the state $V_{t-1}O_{f;t-1;U_{t-1}}\cdots O_{f;1;U_1}V_0\ket{\0}$ is a fixed state, which can be referred to as $\ket{\psi}$. Thus our problem reduces to showing for all quantum states $\ket{\psi}$,
\begin{align}\label{eqn:arbitrary-t}
\E_{\{\u^{(t)},\ldots,\u^{(T)}\}}\big(\|(O_{f;t;U_t}-O_{f;T;U})\ket{\psi}\|^2\big)\leq 8Te^{-dr_0^2/(4T)} .
\end{align}
We write the state $\ket{\psi}$ as $\ket{\psi}=\sum_\x\alpha_\x\ket{\x}\ket{\phi_\x}$, where $\x$ is the query made to the oracle, and $\sum_\x|\alpha_\x|^2=1$. Hence, the left-hand side of Eq.~\eqn{arbitrary-t} equals 
\begin{align}
&\E_{\{\u^{(t)},\ldots,\u^{(T)}\}}\Big(\Big\|\sum_\x\alpha_\x (O_{f;t;U_t}-O_{f;T;U})\ket{\x}\ket{\phi_\x}\Big\|^2\Big)\\
&\qquad\leq\sum_\x|\alpha_\x|^2\cdot\E_{\{\u^{(t)},\ldots,\u^{(T)}\}}\big(\big\|(O_{f;t;U_t}-O_{f;T;U})\ket{\x}\big\|^2\big).
\end{align}
Since $|\alpha_\x|^2$ defines a probability distribution over $\x$, we can again upper bound the right-hand side for any $\x$ instead. Note that $O_{f;t;U_t}$ and $O_{f;T;U}$ behave identically for some inputs x, the only nonzero
terms are those where the oracles respond differently, which can only happen if
\begin{align}
    \big(f_{r_0;t;U_t}(\x),\nabla f_{r_0;t;U_t}(\x)\big)\neq\big(f_{r_0;T;U}(\x),\nabla f_{r_0;T;U}(\x)\big).
\end{align}
When the response is different, we can upper bound $\big\|(O_{f;t;U_t}-O_{f;T;U})\ket{\x}\|^2$ by $4$ using the triangle inequality. Thus for any $\x\in\mathbb{B}(\0,2\sigma\sqrt{T})$, we have
\begin{align}
&\E_{\{\u^{(t)},\ldots,\u^{(T)}\}}\big[\big\|(O_{f;t;U_t}-O_{f;T;U})\ket{\x}\big\|^2\big]\\
&\qquad\leq 4\Pr_{\{\u^{(t)},\ldots,\u^{(T)}\}}\big[\big(f_{r_0;t;U_t}(\x),\nabla f_{r_0;t;U_t}(\x)\big)\neq\big(f_{r_0;T;U}(\x),\nabla f_{r_0;T;U}(\x)\big)\big].
\end{align}
We use $\x_{\perp}$ to denote the projection of $\x$ to the span $\{\u^{(t)},\ldots, \u^{(T)}\}$. Intuitively, as long as each component of $\x_{\perp}$ has absolute value smaller than $\sigma r_0$, the components $\{\u^{(t)},\ldots, \u^{(T)}\}$ will have no observable impact. Quantitatively,
\begin{align}
&\Pr\big[\big(f_{r_0;t;U_t}(\x),\nabla f_{r_0;t;U_t}(\x)\big)\neq\big(f_{r_0;T;U}(\x),\nabla f_{r_0;T;U}(\x)\big)\big]\\
&\qquad\leq 1-\Pr\big[|\<\u^{(t)},\x\>|,\ldots,|\<\u^{(T)},\x\>|\leq \delta r_0\big].
\end{align}
Since $\{\u^{(t)},\ldots, \u^{(T)}\}$ are chosen uniformly at random in the $(d-t+1)$-dimensional orthogonal complement of span $\{\u^{(1)},\ldots,\u^{(t-1)}\}$, for any $t\leq i\leq T$, by \lem{vector-from-sphere} we can further derive that
\begin{align}
    \Pr\big[|\<\u^{(i)},\x\>|>r_0\big]\leq 2e^{-dr_0^2/(4T)},
\end{align}
which leads to
\begin{align}
    \Pr\big[|\<\u^{(t)},\x\>|,\ldots,|\<\u^{(T)},\x\>|\leq r_0\big]&\geq\big(1-2e^{-dr_0^2/(4T)}\big)^T\geq 1-2Te^{-dr_0^2/(4T)},
\end{align}
indicating
\begin{align}
    \E_{\{\u^{(t)},\ldots,\u^{(T)}\}}\big[\big\|(O_{f;t;U_t}-O_{f;T;U})\ket{\x}\big\|^2\big]\leq 8Te^{-dr_0^2/(4T)},
\end{align}
and
\begin{align}
    \E_U\big(\|A_t\ket{\0}-A_{t-1}\ket{\0}\|^2\big)\leq 8Te^{-dr_0^2/(4T)}.
\end{align}
\end{proof}

\begin{proposition}\label{prop:A_0-cannot}
Consider the $d$-dimensional function $\tilde{f}_{T;U}(\x)\colon\mathbb{B}(\0,2\sigma\sqrt{T})\to\mathbb{R}$ defined in \eqn{tildeF-defn} with the rotation matrix $U$ being chosen arbitrarily. Consider any quantum algorithm $A_{\quan}$ containing $t<T$ queries to the noisy oracle $O_f$ defined in Eq.~\eqn{FOSP-oracle}, let $p_U$ be the probability distribution over $\x\in\mathbb{B}(\0,2\sigma \sqrt{T})$ obtained by measuring the state $A_{\quan}\ket{0}$, which is related to the rotation matrix $U$. Then,
\begin{align}
    \Pr_{U,\x_{\text{out}}\sim p_U}\big[\|\nabla \tilde{F}_{T;U}(\x_{\text{out}})\|\leq\lambda\sigma\mu^{3/4}/8\big]\leq 16Te^{-dr_0^2/(8T)}.
\end{align}
\end{proposition}

\begin{proof}
Consider the sequence of unitaries $\{A_0,\ldots,A_t\}$ associated with $A_{\quan}$ defined in \eqn{unitary-sequences}, we first demonstrate that $A_t$ cannot find a point with small gradient with high probability. In particular, let $p_U^{(t)}$ be the probability distribution over $\x\in\mathbb{B}(\0,2\sigma\sqrt{T})$ obtained by measuring the output state $A_t\ket{0}$. Then we have
\begin{align}
    &\Pr_{U_t,\x_{\text{out}}\sim p_{U_t}^{(t)}}\big[\|\nabla\tilde{F}_{T;U}(\x_{\text{out}})\|\leq\lambda\sigma\mu^{3/4}/8\big]\\
    &\quad\qquad\leq\max_{\x\in\mathbb{B}(\0,2\sigma\sqrt{T})}\Pr_{\{\u^{(t+1)},\ldots,\u^{(T+1)}\}}\big[\|\nabla \tilde{F}_{T;U}(\x)\|\leq\lambda\sigma\mu^{3/4}/8\big],
\end{align}
whereby \lem{cannot-guess} we have
\begin{align}
    \Pr_{\{\u^{(t+1)},\ldots,\u^{(T+1)}\}}\big[\|\nabla \tilde{F}_{T;U}(\x)\|\leq\lambda\sigma\mu^{3/4}/8\big]\leq 8Te^{-dr_0^2/(8T)}
\end{align}
for any $\x\in\mathbb{B}(\0,2\sigma\sqrt{T})$, which leads to
\begin{align}
    \Pr_{U_t,\x_{\text{out}}\sim p_{U_t}^{(t)}}\big[\|\nabla\tilde{F}_{T;U}(\x_{\text{out}})\|\leq\lambda\sigma\mu^{3/4}/8\big]\leq 8T^2e^{-dr_0^2/(4T)}.
\end{align}
Moreover, by \lem{similar-outputs} and Cauchy-Schwartz inequality, we have
\begin{align}
    \mathbb{E}_U\big[\|A_t\ket{0}-A_0\ket{0}\|^2\big]\leq t\cdot\mathbb{E}_U\big[\sum_{k=1}^{t-1}\|A_{k+1}\ket{0}-A_k\ket{0}\|^2\big]\leq 8T
    ^2e^{-dr_0^2/(4T)}.
\end{align}
Then by Markov's inequality,
\begin{align}
    \Pr_{U}\big[\|A_{t-1}\ket{0}-A_0\ket{0}\|^2\geq 4Te^{-dr_0^2/(8T)}\big]\leq 4Te^{-dr_0^2/(8T)},
\end{align}
since both norms are at most 1. Hence, we can deduce that the total variance distance between $p_U$n and $p_U^{(t)}$ can be bounded by
\begin{align}
    4Te^{-dr_0^2/(8T)}+4Te^{-dr_0^2/(8T)}\leq 8Te^{-dr_0^2/(8T)},
\end{align}
which further leads to
\begin{align}
    &\Pr_{U,\x_{\text{out}}\sim p_U}\big[\|\nabla \tilde{F}_{T;U}(\x_{\text{out}})\|\leq\lambda\sigma\mu^{3/4}/8\big]\\
    &\quad\qquad\leq\Pr_{U_t,\x_{\text{out}}\sim p_{U_t}^{(t)}}\big[\|\nabla\tilde{F}_{T;U}(\x_{\text{out}})\|\leq\lambda\sigma\mu^{3/4}/8\big]+8Te^{-dr_0^2/(8T)}\leq 16Te^{-dr_0^2/(8T)}.
\end{align}
\end{proof}

\begin{proposition}\label{prop:bounded-lower-bound}
    Let $r_0,\Delta,L_1,L_2,\epsilon$ be positive and $\epsilon\leq L_1^2/L_2$. Then there exist positive numerical constants $c,C\in\mathbb{R}$ and $\ell_q\leq e^{\frac{3q}{2}\log q+Cq}$ for every $q\in\mathbb{N}$, and a set $\Omega$ consisting of function pairs $(\tilde{F}_{T;U},f_{r_0;T;U})$ with $\tilde{F}_{T;U}$ and $f_{r_0;T;U}$ defined in \eqn{tildeF-defn} and \eqn{f-r0-defn} respectively upon the input domain $\mathbb{B}(\0,\mathcal{R})$ with
    \begin{align}
        T=\frac{\Delta}{20}\Big(\frac{L_2}{\ell_2}\Big)^{\frac{5}{4}}\Big(\frac{L_1}{2\ell_1}\Big)^{\frac{3}{7}}(8\epsilon)^{-12/7},\qquad\mathcal{R}=2\sqrt{T}\cdot\Big(\frac{L_2}{\ell_2}\Big)^{-3/4}\Big(\frac{L_1}{\ell_1}\Big)^{-1/7}(8\epsilon)^{4/7},
    \end{align} 
    such that, for any quantum algorithm $A_{\quan}$ making fewer than $T$ queries to the oracle $O_f$ in the form of \eqn{FOSP-oracle} encoding the function values and gradients of $f_{r_0;T;U}$ in \eqn{f-r0-defn}, there exists an orthogonal matrix $U$ such that, $A_{\quan}$ cannot find an $\epsilon$-FOSP of the corresponding $\tilde{F}_{T;U}$ with probability larger than, where $\tilde{F}_{T;U}$ is $L_1$-smooth and $L_2$-Hessian Lipschitz and satisfies
    \begin{align}
        \tilde{F}_{T;U}(\0)-\inf_\x\tilde{F}_{T;U}(\x)\leq \Delta.
    \end{align}
\end{proposition}

\begin{proof}
We set the scaling parameters $\lambda$, $\sigma$ in $\tilde{F}_{T;U}$ and $f_{r_0;T;U}$ to be
\begin{align}
    \lambda=\frac{L_1}{2\ell_1},\qquad\mu=\frac{L_2\sigma}{\lambda\ell_2},\qquad\sigma=\Big(\frac{L_2}{\ell_2}\Big)^{-3/4}\lambda^{-1/7}(8\epsilon)^{4/7},
\end{align}
which satisfy $\mu\leq 1$ since $\epsilon\leq L_1^2/L_2$. By \prop{A_0-cannot}, for any possible quantum algorithm $A_{\quan}$ making $t<T$ queries to the oracle $O_f$ defined in \eqn{FOSP-oracle} encoding $(f_{r_0;U;T},\nabla f_{r_0;U;T})$, we have
\begin{align}
    \Pr_{U,\x_{\text{out}}\sim p_U}\big[\|\nabla \tilde{f}_{T;U}(\x_{\text{out}})\|\leq\lambda\sigma\mu^{3/4}/8\big]
    =\Pr_{U,\x_{\text{out}}\sim p_U}\big[\|\nabla \tilde{f}_{T;U}(\x_{\text{out}})\|\leq\epsilon\big]
    \leq 16Te^{-dr_0^2/(8T)},
\end{align}
where $p_U$ is the probability distribution over $\x\in\mathbb{B}(\0,2\sigma\sqrt{T})$ obtained by measuring the state $A_{\quan}\ket{0}$, indicating that the success probability of $A_{\quan}$ finding an $\epsilon$-FOSP of $\tilde{F}_{T;U}$ is at most $16Te^{-dr_0^2/(8T)}$. 
Moreover, by \lem{FT-conditions} we can derive that, for any $T\times d$ orthogonal matrix $U$, the function $\tilde{F}_{T;U}$ is $(1+\mu)\ell_1= L_1$-smooth and $\lambda\mu\ell_2/\sigma=L_2$-Hessian Lipschitz, with
\begin{align}
    \tilde{F}_{T;U}(\0)-\inf_\x\tilde{F}_{T;U}(\x)\leq 
    \lambda\sigma^2\Big(\frac{\sqrt{\mu}}{2}+10\mu T\Big)\leq\Delta.
\end{align}
\end{proof}

\prop{bounded-lower-bound} shows that, if we restrict the input domain of the function pair $(F,f)$ to a hyperball with radius $\mathcal{R}$, in the worst case every quantum algorithm has to make at least $\Omega\big(\epsilon^{-12/7}\big)$ queries to the noisy evaluation $f$ to find an $\epsilon$-FOSP of $F$ with high probability. Moreover, the dimension $d$ of the hard instance achieving this lower bound is of order $\Omega\big(\epsilon^{-12/7}\log (1/\epsilon)/r_0^2\big)$, where the noise radius $r_0$ is determined by the noise rate with different relations under different noise assumptions, on which a detailed discussion is given in \sec{eps-specific-noise-models} after we extend this lower bound to unbounded input domain in \sec{noisy-lower-unbounded}.

%=====================================================================================

\subsection{Noisy Quantum Lower Bound with Unbounded Input Domain}\label{sec:noisy-lower-unbounded}

In this subsection, we extend the quantum lower bound proved in \prop{bounded-lower-bound} to functions with an unbounded input domains. In particular, Ref.~\cite{carmon2020lower} introduced a method for extending lower bound to unbounded input domain by adding a scaling term on the input vector and additionally introducing a quadratic term. The intuition is that, if the input vector has a large norm, the corresponding function value is almost solely determined by the quadratic term and it cannot be an approximate stationary point. Hence, it is not beneficial for any classical algorithm to explore any point outside a certain bounded region, indicating that the lower bound with an unbounded input domain is the same as the one with a bounded input domain. The same argument also holds for quantum algorithms, as shown in Ref.~\cite{zhang2022quantum}.

Quantitatively, we consider the following $T+1$ dimensional kernel function defined on $\R^{T+1}$,
\begin{align}\label{eqn:FT-square-defn}
    \bar{\mathfrak{F}}_{T;\mu}(\x)
\coloneqq\bar{F}_{T;\mu}(\gamma_\alpha(\x))+\beta\|\sin\x\|^2,
\end{align}
where $\gamma_{\alpha}(\x)$ is defined as
\begin{align}
\gamma_{\alpha}(\x)\coloneqq 
\begin{cases}
\Big(1-\frac{\|\x\|}{\alpha\sqrt{T}}\Big)^3\cdot\x,\quad\|\x\|\leq\alpha\sqrt{T},\\
\0,\quad\|\x\|>\alpha\sqrt{T}.
\end{cases}
\end{align}
By \lem{FT-square-conditions}, finding an $\epsilon$-SOSP or even an $\epsilon$-FOSP of $\bar{\mathfrak{F}}_{T;\mu}$ requires knowledge of all the $T+1$ coordinate directions, if it is projected to a $d$-dimensional space via an arbitrary orthogonal matrix $U\in\R^{d\times(T+1)}$. Moreover, to guarantee that the hard instance satisfies the $B$-boundedness condition required in the empirical risk setting considered in this paper, we additionally add a sine function to the quadratic term and obtain the following hard instance defined on the hypercube $[-\pi\mathcal{L}/2,\pi\mathcal{L}/2]^d$ with $\mathcal{L}=\zeta\alpha\sqrt{T}$ for some constant $\zeta\geq 2$,
\begin{align}\label{eqn:hatFT-square-defn}
    \hat{\mathfrak{F}}_{T;U}(\x)\coloneqq\bar{F}_{T;\mu}\big(\gamma_\alpha(U^{T}\x)\big)+\beta\mathcal{L}^2\|\sin(\x/\mathcal{L})\|^2,
\end{align}
where the constants $\alpha,\beta$ are chosen according to \lem{FT-square-conditions}, and for any $\y\in\R^{d}$, $\sin\y$ is defined as 
\begin{align}
    \sin\y\coloneqq(\sin y_1,\ldots,\sin y_d)^{T}.
\end{align}

\begin{lemma}\label{lem:hatFT-square-conditions}
Consider the function $\hat{\mathfrak{F}}_{T;U}\colon[-\pi\mathcal{L}/2,\pi\mathcal{L}/2]^d\to\R$ defined in Eq.~\eqn{hatFT-square-defn}, suppose that the parameter $\mu$ satisfies $\mu\leq 1$. Then, there exist positive constants $\alpha,\beta,\zeta$ such that
\begin{enumerate}
\item For any $\x\in[-2\zeta\alpha\sqrt{T},2\zeta\alpha\sqrt{T}]^{d}$ such that
\begin{align}
|\<\x,\u^{T}\>|,|\<\x,\u^{T+1}\>|\leq\frac{0.05}{T+1/\sqrt{\mu}},
\end{align}
its gradient satisfies
\begin{align}
\|\nabla\bar{\mathfrak{F}}_{T;U}(\x)\|\geq\mu^{3/4}/32;
\end{align}
\item $\bar{\mathfrak{F}}_{T;U}(\0)-\inf_\x\bar{\mathfrak{F}}_{T;U}(\x)\leq\frac{\sqrt{\mu}}{2}+10\mu T$;
\item For $p=1,2$, the $p$-th order derivatives of $\bar{\mathfrak{F}}_{T;U}$ are $(2\mathbb{I}\{p=1\}+\mu)\ell_p$-Lipschitz continuous in the hyperball $\mathbb{B}(\0,\alpha\sqrt{T})$, where $\ell_p\leq\exp\big(\frac{3p}{2}\log p+cp\big)$ for a numerical constant $c<\infty$.
\end{enumerate}
\end{lemma}
\begin{proof}
We set the constants $\alpha,\beta$ according to \lem{FT-square-conditions}. Note that for any vector $\y$ with $\|\y\|\leq 1/(2\zeta)$, the values as well as first- and second-order derivatives of $\|\y\|^2$ and $\|\sin\y\|^2$ are close to each other given that $\zeta$ reaches a large enough value that is independent from $d$. Quantitatively, we have
\begin{align}
    \|\y\|^2-\|\sin\y\|^2\leq\sum_{i=1}^d y_i^2-\big(y_i-y_i^3/6\big)^2\leq\sum_{i=1}^d\frac{y_i^4}{3}\leq\frac{1}{48\zeta^4},
\end{align}
and
\begin{align}
    \big\|\nabla\cdot(\|\y\|^2-\|\sin\y\|^2)\big\|=\|2\y-\sin(2\y)\|
    \leq\frac{1}{6}\cdot\frac{1}{\zeta^3}=\frac{1}{6\zeta^3}.
\end{align}
Moreover, we notice that
\begin{align}
    \nabla^2(\|\y\|^2-\|\sin\y\|^2)=I-
    \begin{bmatrix}
    \cos 2y_1,&\cdots,&0\\
    \vdots, &\ddots,&\vdots\\
    0,&\cdots,&\cos 2y_d
    \end{bmatrix},
\end{align}
which leads to
\begin{align}
    \big\|\nabla^2(\|\y\|^2-\|\sin\y\|^2)\big\|\leq\frac{1}{2}\cdot(2/\zeta)^2=\frac{2}{\zeta^2}.
\end{align}
Hence, there exists a large enough $\zeta=O(1/\mu)$ independent from $d$ such that, $\hat{\mathfrak{F}}_{T;U}$ is close enough to the pure rotation of $\bar{\mathfrak{F}}_{T;\mu}$ in the hyperball $\mathbb{B}(\0,\alpha\sqrt{T})$ up to the second-order derivatives, and the above three conditions can be satisfied. 
\end{proof}

Note that if we replicate the hypercube $[-\pi\mathcal{L}/2,\pi\mathcal{L}/2]^d$ in $\R^d$ consecutively and have the function value in each hypercube being $\hat{\mathfrak{F}}_{T;U}$ respectively, the new function defined on $\mathbb{R}^d$ is still infinitely differentiable. Moreover, we can notice that finding an $\epsilon$-SOSP in $\R^d$ is equivalent to finding an $\epsilon$-SOSP in one specific hypercube $[-\pi\mathcal{L}/2,\pi\mathcal{L}/2]^d$, since for any $\x$ on the boundary of $[-\pi\mathcal{L}/2,\pi\mathcal{L}/2]^d$, the Hessian matrix
\begin{align}
\nabla^2\hat{\mathfrak{F}}_{T;U}(\x)=2\beta\cdot
\begin{bmatrix}
\cos(2x_1/\mathcal{L}) & \ldots & 0\\
\vdots & \ddots & \vdots\\
0 & \cdots & \cos(2x_d/\mathcal{L})
\end{bmatrix},
\end{align}
is positive definite with the matrix norm being $2\beta$, indicating that $\x$ cannot be an $\epsilon$-SOSP. Similar to \sec{noisy-lower}, we add scaling parameters $\lambda$ and $\sigma$ to $\hat{\mathfrak{F}}_{T;U}$ and obtain the formal hard function
\begin{align}\label{eqn:tildeFT-square-defn}
\tilde{\mathfrak{F}}_{T;U}(\x)\coloneqq\lambda\sigma^2\hat{\mathfrak{F}}_{T;U}(\x/\sigma).
\end{align}
Moreover, we assume access to the following noisy evaluation $\mathfrak{f}_{r_0;T;U}$ of $\hat{\mathfrak{F}}_{T;U}$,
\begin{align}\label{eqn:f-r0-square-defn}
\mathfrak{f}_{r_0;T;U}&\coloneqq 
\lambda\sigma^2\big[\bar{F}_{T;\mu}\big(\gamma_\alpha(\<\x/\sigma,\u^{(1)}\>),\ldots,\gamma_\alpha(\<\x/\sigma,\u^{(\prog_{r_0}(\gamma_\alpha(\x/\mu))+1)}\>),0,\ldots,0\big)\\
&\qquad\quad+\beta\mathcal{L}\|\sin(\x/(\sigma\mathcal{L}^2))\|^2\big],  
\end{align}
which is encoded in the quantum oracle $O_{\mathfrak{f}}$ with form \eqn{FOSP-oracle}. Then, we can prove the following quantum lower bound via the function pair $(\tilde{\mathfrak{F}}_{T;U},\mathfrak{f}_{r_0;T;U})$.

\begin{theorem}\label{thm:unbounded-lower-bound}
    Let $\Delta,\epsilon,r_0$ be positive and $\epsilon\leq 1/2$, where the noise radius $r_0$ is a parameter related to the noise rate. Then there exist positive numerical constants $c,C,\alpha,\beta\in\mathbb{R}$ and $\ell_q\leq e^{\frac{3q}{2}\log q+Cq}$ for every $q\in\mathbb{N}$, and a set $\Omega$ consisting of function pairs $(\tilde{\mathfrak{F}}_{T;U},\mathfrak{f}_{r_0;T;U})$ with $\tilde{\mathfrak{F}}_{T;U}$ and $\mathfrak{f}_{r_0;T;U}$ defined in \eqn{tildeFT-square-defn} and \eqn{f-r0-square-defn} respectively with
    \begin{align}
        T=\frac{\Delta}{20}\Big(\frac{1}{\ell_2}\Big)^{\frac{5}{4}}\Big(\frac{1}{2\ell_1}\Big)^{\frac{3}{7}}(16\epsilon)^{-12/7},
    \end{align} 
    such that, for any quantum algorithm $A_{\quan}$ making fewer than $T$ queries to the oracle $O_{\mathfrak{f}}$ in the form of \eqn{FOSP-oracle} encoding the function value and gradient of $\mathfrak{f}_{r_0;T;U}$ in \eqn{f-r0-square-defn}, there exists an orthogonal matrix $U$ such that, $A_{\quan}$ cannot find an $\epsilon$-SOSP of the corresponding $\tilde{\mathfrak{F}}_{T;U}$ with probability larger than $e^{-dr_0^2/(2\alpha^2T)}$, where $\tilde{\mathfrak{F}}_{T;U}$ is $B$ bounded, $L_1$-smooth, and $L_2$-Hessian Lipschitz with 
    \begin{align}
        B=O\big(\Delta+\epsilon^{-12/7}\big),\quad L_1=O(1),\quad L_2=O(1),
    \end{align}
    and satisfies
    \begin{align}
        \tilde{\mathfrak{F}}_{T;U}(\0)-\inf_\x\tilde{\mathfrak{F}}_{T;U}(\x)\leq \Delta.
    \end{align}
\end{theorem}

\begin{proof}
Since the functions and noisy evaluations in each hypercube are the same, without loss of generality we assume all queries happen in the hypercube $[-\pi\sigma\mathcal{L}/2,\pi\sigma\mathcal{L}/2]^d$ centered at $\0$. Similar to the setting of \prop{bounded-lower-bound}, we set the scaling parameters $\lambda,\mu$ and $\mu$ in $\tilde{\mathfrak{F}}_{T;U}$ and $\mathfrak{f}_{r_0;T;U}$ to be
\begin{align}
\lambda=\frac{1}{2\ell_1},\quad\mu=\frac{\sigma}{\lambda\ell_2},\quad\sigma=\ell_2^{3/4}\lambda^{-2/7}(16\epsilon)^{4/7},
\end{align}
which satisfies $\mu\leq 1$ since $\epsilon\leq 1$. By \lem{hatFT-square-conditions}, finding an $\lambda\sigma\mu^{3/4}/32=\epsilon$-SOSP of $\tilde{\mathfrak{F}}$ with high probability requires complete knowledge of all the $T+1$ columns of the matrix $U$. Equivalently, we can find an $2\epsilon$-SOSP of the function $\tilde{F}_{T;U}$ by finding an $\frac{\epsilon}{2}$-SOSP of $\tilde{\mathfrak{F}}_{T;U}$ with the same $U$ and same settings of parameters, which by \prop{bounded-lower-bound} requires at least $T$ queries to the quantum oracle $O_f$ encoding the noisy evaluation $f_{r_0;T;U}$ of $\tilde{F}_{T;U}$ to guarantee a success probability at least $e^{-dr_0^2/(2\alpha^2T)}$.

In addition, we notice that one query to the quantum oracle $O_f$ can be implemented via one query to the quantum oracle $O_{\mathfrak{f}}$ encoding the noisy evaluation $\mathfrak{f}_{r_0;U;T}$. Hence, by \prop{bounded-lower-bound} we can claim that to find an $\epsilon$-SOSP of $\tilde{\mathfrak{F}}_{T;U}$ with success probability at least $e^{-dr_0^2/(2\alpha^2T)}$, it takes at least
\begin{align}
    T=\frac{\Delta}{20}\Big(\frac{1}{\ell_2}\Big)^{\frac{5}{4}}\Big(\frac{1}{2\ell_1}\Big)^{\frac{3}{7}}(16\epsilon)^{-12/7}
\end{align}
queries to the oracle $O_{\mathfrak{f}}$.

Moreover, by the second entry of \lem{hatFT-square-conditions}, we know that 
\begin{align}
    \tilde{\mathfrak{F}}_{T;U}(\0)-\inf_\x\tilde{\mathfrak{F}}_{T;U}(\x)\leq\lambda\sigma^2\Big(\frac{\mu}{2}+10\mu T\Big)\leq \Delta.
\end{align}
Further, we can observe that
\begin{align}
    \sup_\x\tilde{\mathfrak{F}}_{T;U}(\x)-\tilde{\mathfrak{F}}_{T;U}(\0)&\leq \lambda\sigma^2\sup_{\|\x\|\leq\alpha\sqrt{T}}\bar{F}_{T;\mu}(\x)+\lambda\sigma^2\beta\mathcal{L}^2\sup_\x\|\sin(\x/(\sigma\mathcal{L}))\|^2\\
    &\leq \lambda\sigma^2\big(2\alpha^2T+60\alpha^2T+\beta\zeta^2\alpha^2 T\big)\\
    &=O(\lambda\sigma^2\zeta^2T)=O(\lambda\sigma^2\mu^{-2}T)\\
    &=O\big(\epsilon^{-12/7}\big),
\end{align}
indicating that $\tilde{\mathfrak{F}}_{T;U}(\x)$ is $B$-bounded for $B=O(\Delta+\epsilon^{-12/7})$.

By the third entry of \lem{hatFT-square-conditions}, $\tilde{F}_{T;U}$ is $\lambda(2+\mu)\ell_1=O(1)$-smooth and $\mu\ell_2/\sigma=O(1)$-Hessian Lipschitz in the region $\mathbb{B}(\0,\alpha\sigma\sqrt{T})$. For any point $\x\in[-\pi\mathcal{L}\sigma/2,\pi\mathcal{L}\sigma/2]^d-\mathbb{B}(\0,\alpha\sigma\sqrt{T})$, we have
\begin{align}
    \|\nabla^2\tilde{\mathfrak{F}}_{T;U}(\x)\|=\beta\mathcal{L}^2\sigma^2\Big\|\nabla^2\sin^2\Big(\frac{\x}{\sigma\mathcal{L}}\Big)\Big\|\leq 4\beta=O(1),
\end{align}
and 
\begin{align}
\|\nabla^3\tilde{F}_{T;U}(\x)\|=\beta\mathcal{L}^2\sigma^2\Big\|\nabla^3\sin^2\Big(\frac{\x}{\sigma\mathcal{L}}\Big)\Big\|\leq\frac{8\beta}{\sigma\mathcal{L}}=O(1).
\end{align}
Hence, we can conclude that $\tilde{\mathfrak{F}}_{T;U}$ is $O(1)$-smooth and $O(1)$-Hessian Lipschitz in the entire space $\R^d$.
\end{proof}

%=====================================================================================

\subsection{Lower Bound for Quantum Algorithms with Noisy Zeroth- and First-order Oracles}\label{sec:eps-specific-noise-models}

In this subsections, we specify the value of noise radius $r_0$ appearing in \thm{unbounded-lower-bound} when we are given noisy zeroth-order oracle or noisy first-order oracle satisfying \assume{ZeroProb} or \assume{FirstProb}, respectively, and further discuss the requirement on dimension $d$ to obtain our lower bound in $\epsilon$. 

We first discuss the setting with zeroth-order oracle access.

\begin{corollary}[Formal version of \thm{QuantumCarmon}, Part 1]\label{cor:zeroth-lower-bound}
    Let $\Delta,\epsilon>0$ and $\epsilon\leq 1/2$. Then there exist positive numerical constants $c,C,\alpha,\beta\in\mathbb{R}$ and $\ell_q\leq e^{\frac{3q}{2}\log q+Cq}$ for every $q\in\mathbb{N}$, and a set $\Omega$ consisting of function pairs $(F,f)$ satisfying \assume{ZeroProb} with some $\nu$ satisfying
    \begin{align}
    \nu=\Omega\big(\epsilon^{-16/7}/d\big),
    \end{align}
    such that, for any quantum algorithm $A_{\quan}$ making fewer than $\Theta(\epsilon^{-12/7})$ queries to the oracle $O_f$ defined in \eqn{FOSP-oracle} encoding the function value and gradient of $f$, there exists a function pair $(F,f)$ such that $A_{\quan}$ cannot find an $\epsilon$-SOSP of $F$ with probability larger than $1/3$, where $F$ is $B$ bounded, $L_1$-smooth, and $L_2$-Hessian Lipschitz with 
    \begin{align}
        B=O\big(\Delta+\epsilon^{-12/7}\big),\quad L_1=O(1),\quad L_2=O(1),
    \end{align}
    and satisfies
    \begin{align}
        F(\0)-\inf_\x F_{T;U}(\x)\leq \Delta.
    \end{align}
\end{corollary}

\begin{proof}
We adopt the settings of functions and parameters in \thm{unbounded-lower-bound} and set $\Omega$ to be
\begin{align}
    \Omega=\{(\tilde{\mathfrak{F}}_{T;U},\mathfrak{f}_{r_0;T;U})\,|\,U\in\R^{d\times(T+1)}\text{ s.t. }U^\top U=I\},
\end{align}
where 
\begin{align}
    T=\frac{\Delta}{20}\Big(\frac{1}{2\ell_1}\Big)^{\frac{3}{7}}\Big(\frac{1}{\ell_2}\Big)^{\frac{5}{4}}(16\epsilon)^{-12/7}.
\end{align}
By \lem{r0-max-differences}, the parameter $r_0$ satisfies
\begin{align}
\lambda\sigma^2(50r_0^2T+2\alpha r_0\sqrt{T})\leq\nu.
\end{align}
Moreover, by \thm{unbounded-lower-bound}, if the dimension $d$ satisfies
\begin{align}\label{eqn:zeroth-order-d}
    d\geq \frac{4\alpha^2T}{r_0^2},
\end{align}
then for any quantum algorithm $A_{\quan}$ making $T$ queries to $O_f$, there exists a function pair\\ $(F,f)=(\tilde{\mathfrak{F}}_{T;U},\mathfrak{f}_{r_0,T;U})\in\Omega$ such that the success probability of $A_{\quan}$ finding an $\epsilon$-SOSP of $\tilde{\mathfrak{F}}_{T;U}$ is at most
\begin{align}
    \exp\big(-dr_0^2/(2\alpha^2T)\big)\leq\frac{1}{3}.
\end{align}
In order to guarantee inequality \eqn{zeroth-order-d}, we can require $\nu$ to satisfy
\begin{align}
\nu\geq\frac{4\alpha^2(50+2\alpha)\lambda\sigma^2T^2}{d}\geq\Omega(\epsilon^{-16/7}/d).
\end{align}
Moreover, by \thm{unbounded-lower-bound} we can conclude that $F=\tilde{\mathfrak{F}}_{T;U}$ is $O(\Delta+\epsilon^{-12/7})$-bounded, $O(1)$-smooth and $O(1)$ Hessian Lipschitz with
    \begin{align}
        F(\0)-\inf_\x F_{T;U}(\x)\leq\Delta.
    \end{align}
\end{proof}

A similar conclusion can be obtained concerning the setting with first-order oracle access.

\begin{corollary}[Formal version of \thm{QuantumCarmon}, Part 2 
]\label{cor:first-lower-bound}
    Let $\Delta,\epsilon>0$ and $\epsilon\leq 1/2$. Then there exist positive numerical constants $c,C,\alpha,\beta\in\mathbb{R}$ and $\ell_q\leq e^{\frac{3q}{2}\log q+Cq}$ for every $q\in\mathbb{N}$, and a set $\Omega$ consisting of function pairs $(F,f)$ satisfying \assume{FirstProb} except the smoothness condition of $f$ with some $\tilde{\nu}$ satisfying
    \begin{align}
    \tilde{\nu}=\Omega(\epsilon^{-8/7}/\sqrt{d})
    \end{align}
    such that, for any quantum algorithm $A_{\quan}$ making fewer than $\Theta(\epsilon^{-12/7})$ queries to the oracle $O_f$ defined in \eqn{FOSP-oracle} encoding the function value and gradient of $f$, there exists a function pair $(F,f)$ such that $A_{\quan}$ cannot find an $\epsilon$-SOSP of $F$ with probability larger than $1/3$, where $F$ is $B$ bounded, $L_1$-smooth, and $L_2$-Hessian Lipschitz with 
    \begin{align}
        B=O\big(\Delta+\epsilon^{-12/7}\big),\quad L_1=O(1),\quad L_2=O(1),
    \end{align}
    and satisfies
    \begin{align}
        F(\0)-\inf_\x F_{T;U}(\x)\leq \Delta.
    \end{align}
\end{corollary}
\begin{remark}
One may notice that in the statement of \cor{first-lower-bound}, the hard instance $(F,f)$ we consider only satisfies part of \assume{FirstProb} except the smoothness condition of $f$. Nevertheless, adopting a similar smoothing technique presented in \sec{DLowerFirst}, we can modify the hard instance to further satisfy the smoothness condition of $f$ without affecting the asymptotic lower bound.
\end{remark}

\begin{proof}
We adopt the settings of functions and parameters in \thm{unbounded-lower-bound} and set $\Omega$ to be
\begin{align}
    \Omega=\big\{(\tilde{\mathfrak{F}}_{T;U},\mathfrak{f}_{r_0;T;U})\,|\,U\in\R^{d\times(T+1)}\text{ s.t. }U^\top U=I\big\},
\end{align}
where 
\begin{align}
    T=\frac{\Delta}{20}\Big(\frac{1}{2\ell_1}\Big)^{\frac{3}{7}}\Big(\frac{1}{\ell_2}\Big)^{\frac{5}{4}}(16\epsilon)^{-12/7}.
\end{align}
By , the parameter $r_0$ satisfies
\begin{align}
    6\lambda\sigma r_0\sqrt{T}\leq\tilde{\nu}.
\end{align} 
Moreover, by \thm{unbounded-lower-bound}, if the dimension $d$ satisfies
\begin{align}
    d\geq \frac{4\alpha^2T}{r_0^2},
\end{align}
then for any quantum algorithm $A_{\quan}$ making $T$ queries to $O_f$, there exists a function pair\\ $(F,f)=(\tilde{\mathfrak{F}}_{T;U},\mathfrak{f}_{r_0,T;U})\in\Omega$ such that the success probability of $A_{\quan}$ finding an $\epsilon$-SOSP of $\tilde{\mathfrak{F}}_{T;U}$ is at most
\begin{align}\label{eqn:eps-first-order-error-probability}
    \exp\big(-dr_0^2/(2\alpha^2T)\big)\leq\frac{1}{3}.
\end{align}
In order to guarantee inequality \eqn{eps-first-order-error-probability}, we can require $\tilde{nu}$ to satisfy
\begin{align}
\tilde{\nu}\geq6\lambda\sigma \cdot\sqrt{4\alpha^2T^2/d}=\Omega(\epsilon^{-8/7}/\sqrt{d}).
\end{align}
Moreover, by \thm{unbounded-lower-bound} we can conclude that $F=\tilde{\mathfrak{F}}_{T;U}$ is $O(\Delta+\epsilon^{-12/7})$-bounded, $O(1)$-smooth and $O(1)$ Hessian Lipschitz with
    \begin{align}
        F(\0)-\inf_\x F_{T;U}(\x)\leq\Delta.
    \end{align}
\end{proof}

%%%%%%%%%%%%%%%%%%%%%%%%%%%%%%%%%%%%%%%%%%%%%%%%%%%%%%%%%%%%%%%%%%%%%%%%%%%%%%

\section*{Acknowledgement}
We thank Yizhou Liu for helpful discussions about the quantum tunneling walk in~\cite{liu2022quantum}. CZ was supported by the AFOSR under grant FA9550-21-1-039. TL was supported by a startup fund from Peking University, and the Advanced Institute of Information Technology, Peking University.

%%%%%%%%%%%%%%%%%%%%%%%%%%%%%%%%%%%%%%%%%%%%%%%%%%%%%%%%%%%%%%%%%%%%%%%%%%%%%%

\bibliographystyle{myhamsplain}
\bibliography{QuantumNonconvexEmpirical}

\providecommand{\bysame}{\leavevmode\hbox to3em{\hrulefill}\thinspace}
\begin{thebibliography}{10}

\bibitem{agarwal2017finding}
Naman Agarwal, Zeyuan Allen-Zhu, Brian Bullins, Elad Hazan, and Tengyu Ma,
  \emph{Finding approximate local minima faster than gradient descent},
  Proceedings of the 49th Annual ACM SIGACT Symposium on Theory of Computing,
  pp.~1195--1199, 2017,
  \mbox{\href{http://arxiv.org/abs/arXiv:1611.01146}{arXiv:1611.01146}}.

\bibitem{anonymous2023faster}
Anonymous, \emph{Faster gradient-free methods for escaping saddle points},
  Submitted to The Eleventh International Conference on Learning
  Representations, 2023, under review.

\bibitem{auer1995exponentially}
Peter Auer, Mark Herbster, and Manfred~K. Warmuth, \emph{Exponentially many
  local minima for single neurons}, Advances in Neural Information Processing
  Systems, vol.~8, 1995.

\bibitem{bartlett2002rademacher}
Peter~L. Bartlett and Shahar Mendelson, \emph{Rademacher and gaussian
  complexities: Risk bounds and structural results}, Journal of Machine
  Learning Research \textbf{3} (2002), no.~Nov, 463--482.

\bibitem{belloni2015escaping}
Alexandre Belloni, Tengyuan Liang, Hariharan Narayanan, and Alexander Rakhlin,
  \emph{Escaping the local minima via simulated annealing: Optimization of
  approximately convex functions}, Conference on Learning Theory, pp.~240--265,
  PMLR, 2015,
  \mbox{\href{http://arxiv.org/abs/arXiv:1501.07242}{arXiv:1501.07242}}.

\bibitem{Bennett1997Strengths}
Charles~H Bennett, Ethan Bernstein, Gilles Brassard, and Umesh Vazirani,
  \emph{Strengths and weaknesses of quantum computing}, SIAM journal on
  Computing \textbf{26} (1997), no.~5, 1510--1523.

\bibitem{berry2007efficient}
Dominic~W. Berry, Graeme Ahokas, Richard Cleve, and Barry~C. Sanders,
  \emph{Efficient quantum algorithms for simulating sparse {H}amiltonians},
  Communications in Mathematical Physics \textbf{270} (2007), no.~2, 359--371,
  \mbox{\href{http://arxiv.org/abs/arXiv:quant-ph/0508139}{arXiv:quant-ph/0508139}}.

\bibitem{berry2015hamiltonian}
Dominic~W. Berry, Andrew~M. Childs, and Robin Kothari, \emph{Hamiltonian
  simulation with nearly optimal dependence on all parameters}, Proceedings of
  the 56th Annual Symposium on Foundations of Computer Science, pp.~792--809,
  IEEE, 2015,
  \mbox{\href{http://arxiv.org/abs/arXiv:1501.01715}{arXiv:1501.01715}}.

\bibitem{boucheron2013concentration}
St{\'e}phane Boucheron, G{\'a}bor Lugosi, and Pascal Massart,
  \emph{Concentration inequalities: A nonasymptotic theory of independence},
  Oxford university press, 2013.

\bibitem{brandao2017quantum}
Fernando G. S.~L. Brand{\~a}o and Krysta~M. Svore, \emph{Quantum speed-ups for
  solving semidefinite programs}, 2017 IEEE 58th Annual Symposium on
  Foundations of Computer Science (FOCS), pp.~415--426, IEEE, 2017,
  \mbox{\href{http://arxiv.org/abs/arXiv:1609.05537}{arXiv:1609.05537}}.

\bibitem{brandao2017SDP}
Fernando~G.S.L. Brand{\~a}o, Amir Kalev, Tongyang Li, Cedric Yen-Yu Lin,
  Krysta~M. Svore, and Xiaodi Wu, \emph{Quantum {SDP} solvers: {L}arge
  speed-ups, optimality, and applications to quantum learning}, Proceedings of
  the 46th International Colloquium on Automata, Languages, and Programming,
  Leibniz International Proceedings in Informatics (LIPIcs), vol. 132,
  pp.~27:1--27:14, Schloss Dagstuhl--Leibniz-Zentrum fuer Informatik, 2019,
  \mbox{\href{http://arxiv.org/abs/arXiv:1710.02581}{arXiv:1710.02581}}.

\bibitem{brassard2002quantum}
Gilles Brassard, Peter Hoyer, Michele Mosca, and Alain Tapp, \emph{Quantum
  amplitude amplification and estimation}, Contemporary Mathematics
  \textbf{305} (2002), 53--74,
  \mbox{\href{http://arxiv.org/abs/arXiv:quant-ph/0005055}{arXiv:quant-ph/0005055}}.

\bibitem{brutzkus2017globally}
Alon Brutzkus and Amir Globerson, \emph{Globally optimal gradient descent for a
  convnet with gaussian inputs}, International Conference on Machine Learning,
  pp.~605--614, PMLR, 2017,
  \mbox{\href{http://arxiv.org/abs/arXiv:1702.07966}{arXiv:1702.07966}}.

\bibitem{buhrman2007robust}
Harry Buhrman, Ilan Newman, Hein Rohrig, and Ronald de~Wolf, \emph{Robust
  polynomials and quantum algorithms}, Theory of Computing Systems \textbf{40}
  (2007), no.~4, 379--395,
  \mbox{\href{http://arxiv.org/abs/arXiv:quant-ph/0309220}{arXiv:quant-ph/0309220}}.

\bibitem{carmon2018accelerated}
Yair Carmon, John~C. Duchi, Oliver Hinder, and Aaron Sidford, \emph{Accelerated
  methods for nonconvex optimization}, SIAM Journal on Optimization \textbf{28}
  (2018), no.~2, 1751--1772,
  \mbox{\href{http://arxiv.org/abs/arXiv:1611.00756}{arXiv:1611.00756}}.

\bibitem{carmon2020lower}
\bysame, \emph{Lower bounds for finding stationary points {I}}, Mathematical
  Programming \textbf{184} (2020), no.~1, 71--120,
  \mbox{\href{http://arxiv.org/abs/arXiv:1710.11606}{arXiv:1710.11606}}.

\bibitem{carmon2021lower}
Yair Carmon, John~C Duchi, Oliver Hinder, and Aaron Sidford, \emph{Lower bounds
  for finding stationary points {II}: first-order methods}, Mathematical
  Programming \textbf{185} (2021), no.~1, 315--355.

\bibitem{casares2020quantum}
Pablo A.~M. Casares and Miguel~Angel Martin-Delgado, \emph{A quantum
  interior-point predictor--corrector algorithm for linear programming},
  Journal of physics A: Mathematical and Theoretical \textbf{53} (2020),
  no.~44, 445305,
  \mbox{\href{http://arxiv.org/abs/arXiv:1902.06749}{arXiv:1902.06749}}.

\bibitem{chakrabarti2020optimization}
Shouvanik Chakrabarti, Andrew~M. Childs, Tongyang Li, and Xiaodi Wu,
  \emph{Quantum algorithms and lower bounds for convex optimization}, Quantum
  \textbf{4} (2020), 221,
  \mbox{\href{http://arxiv.org/abs/arXiv:1809.01731}{arXiv:1809.01731}}.

\bibitem{chen2020stationary}
Xi~Chen, Simon~S. Du, and Xin~T. Tong, \emph{On stationary-point hitting time
  and ergodicity of stochastic gradient {L}angevin dynamics}, Journal of
  Machine Learning Research (2020), 1--40,
  \mbox{\href{http://arxiv.org/abs/arXiv:1904.13016}{arXiv:1904.13016}}.

\bibitem{childs2017note}
Andrew~M. Childs, \emph{Lecture notes on quantum algorithms},
  \url{https://www.cs.umd.edu/\%7Eamchilds/qa/qa.pdf}, 2017.

\bibitem{childs2001robustness}
Andrew~M. Childs, Edward Farhi, and John Preskill, \emph{Robustness of
  adiabatic quantum computation}, Physical Review A \textbf{65} (2001), no.~1,
  012322,
  \mbox{\href{http://arxiv.org/abs/arXiv:quant-ph/0108048}{arXiv:quant-ph/0108048}}.

\bibitem{cornelissen2022near}
Arjan Cornelissen, Yassine Hamoudi, and Sofiene Jerbi, \emph{Near-optimal
  quantum algorithms for multivariate mean estimation}, Proceedings of the 54th
  Annual ACM SIGACT Symposium on Theory of Computing, pp.~33--43, 2022,
  \mbox{\href{http://arxiv.org/abs/arXiv:2111.09787}{arXiv:2111.09787}}.

\bibitem{cross2015quantum}
Andrew~W. Cross, Graeme Smith, and John~A. Smolin, \emph{Quantum learning
  robust against noise}, Physical Review A \textbf{92} (2015), no.~1, 012327,
  \mbox{\href{http://arxiv.org/abs/arXiv:1407.5088}{arXiv:1407.5088}}.

\bibitem{duchi2015optimal}
John~C. Duchi, Michael~I. Jordan, Martin~J. Wainwright, and Andre Wibisono,
  \emph{Optimal rates for zero-order convex optimization: The power of two
  function evaluations}, IEEE Transactions on Information Theory \textbf{61}
  (2015), no.~5, 2788--2806,
  \mbox{\href{http://arxiv.org/abs/arXiv:1312.2139}{arXiv:1312.2139}}.

\bibitem{endo2018practical}
Suguru Endo, Simon~C. Benjamin, and Ying Li, \emph{Practical quantum error
  mitigation for near-future applications}, Physical Review X \textbf{8}
  (2018), no.~3, 031027,
  \mbox{\href{http://arxiv.org/abs/arXiv:1712.09271}{arXiv:1712.09271}}.

\bibitem{endo2021hybrid}
Suguru Endo, Zhenyu Cai, Simon~C. Benjamin, and Xiao Yuan, \emph{Hybrid
  quantum-classical algorithms and quantum error mitigation}, Journal of the
  Physical Society of Japan \textbf{90} (2021), no.~3, 032001,
  \mbox{\href{http://arxiv.org/abs/arXiv:2011.01382}{arXiv:2011.01382}}.

\bibitem{fang2019sharp}
Cong Fang, Zhouchen Lin, and Tong Zhang, \emph{Sharp analysis for nonconvex
  {SGD} escaping from saddle points}, Conference on Learning Theory,
  pp.~1192--1234, 2019,
  \mbox{\href{http://arxiv.org/abs/arXiv:1902.00247}{arXiv:1902.00247}}.

\bibitem{garg2021near}
Ankit Garg, Robin Kothari, Praneeth Netrapalli, and Suhail Sherif,
  \emph{Near-optimal lower bounds for convex optimization for all orders of
  smoothness}, Advances in Neural Information Processing Systems \textbf{34}
  (2021), 29874--29884,
  \mbox{\href{http://arxiv.org/abs/arXiv:2112.01118}{arXiv:2112.01118}}.

\bibitem{garg2020no}
Ankit Garg, Robin Kothari, Praneeth Netrapalli, and Suhail Sherif, \emph{No
  quantum speedup over gradient descent for non-smooth convex optimization},
  12th Innovations in Theoretical Computer Science Conference, Leibniz
  International Proceedings in Informatics (LIPIcs), vol. 185, pp.~53:1--53:20,
  Schloss Dagstuhl--Leibniz-Zentrum f{\"u}r Informatik, 2021,
  \mbox{\href{http://arxiv.org/abs/arXiv:2010.01801}{arXiv:2010.01801}}.

\bibitem{gilyen2019optimizing}
Andr{\'a}s Gily{\'e}n, Srinivasan Arunachalam, and Nathan Wiebe,
  \emph{Optimizing quantum optimization algorithms via faster quantum gradient
  computation}, Proceedings of the 30th Annual ACM-SIAM Symposium on Discrete
  Algorithms, pp.~1425--1444, Society for Industrial and Applied Mathematics,
  2019, \mbox{\href{http://arxiv.org/abs/arXiv:1711.00465}{arXiv:1711.00465}}.

\bibitem{gottesman1997stabilizer}
Daniel Gottesman, \emph{Stabilizer codes and quantum error correction},
  California Institute of Technology, 1997.

\bibitem{hamoudi2021quantum}
Yassine Hamoudi, \emph{Quantum sub-{G}aussian mean estimator}, 29th Annual
  European Symposium on Algorithms, Leibniz International Proceedings in
  Informatics (LIPIcs), vol. 204, pp.~50:1--50:17, Schloss Dagstuhl --
  Leibniz-Zentrum f{\"u}r Informatik, 2021,
  \mbox{\href{http://arxiv.org/abs/arXiv:2108.12172}{arXiv:2108.12172}}.

\bibitem{harrow2003robustness}
Aram~W. Harrow and Michael~A. Nielsen, \emph{Robustness of quantum gates in the
  presence of noise}, Physical Review A \textbf{68} (2003), no.~1, 012308,
  \mbox{\href{http://arxiv.org/abs/arXiv:quant-ph/0301108}{arXiv:quant-ph/0301108}}.

\bibitem{hopkins2020mean}
Samuel~B Hopkins, \emph{Mean estimation with sub-gaussian rates in polynomial
  time}, The Annals of Statistics \textbf{48} (2020), no.~2, 1193--1213,
  \mbox{\href{http://arxiv.org/abs/arXiv:1809.07425}{arXiv:1809.07425}}.

\bibitem{jin2017escape}
Chi Jin, Rong Ge, Praneeth Netrapalli, Sham~M. Kakade, and Michael~I. Jordan,
  \emph{How to escape saddle points efficiently}, Proceedings of the 34th
  International Conference on Machine Learning, vol.~70, pp.~1724--1732, 2017,
  \mbox{\href{http://arxiv.org/abs/arXiv:1703.00887}{arXiv:1703.00887}}.

\bibitem{jin2018local}
Chi Jin, Lydia~T. Liu, Rong Ge, and Michael~I. Jordan, \emph{On the local
  minima of the empirical risk}, Advances in Neural Information Processing
  Systems, vol.~31, 2018,
  \mbox{\href{http://arxiv.org/abs/arXiv:1803.09357}{arXiv:1803.09357}}.

\bibitem{jin2021nonconvex}
Chi Jin, Praneeth Netrapalli, Rong Ge, Sham~M. Kakade, and Michael~I Jordan,
  \emph{On nonconvex optimization for machine learning: Gradients,
  stochasticity, and saddle points}, Journal of the ACM (JACM) \textbf{68}
  (2021), no.~2, 1--29,
  \mbox{\href{http://arxiv.org/abs/arXiv:1902.04811}{arXiv:1902.04811}}.

\bibitem{jin2018accelerated}
Chi Jin, Praneeth Netrapalli, and Michael~I. Jordan, \emph{Accelerated gradient
  descent escapes saddle points faster than gradient descent}, Conference on
  Learning Theory, pp.~1042--1085, 2018,
  \mbox{\href{http://arxiv.org/abs/arXiv:1711.10456}{arXiv:1711.10456}}.

\bibitem{jordan2005fast}
Stephen~P. Jordan, \emph{Fast quantum algorithm for numerical gradient
  estimation}, Physical Review Letters \textbf{95} (2005), no.~5, 050501,
  \mbox{\href{http://arxiv.org/abs/arXiv:quant-ph/0405146}{arXiv:quant-ph/0405146}}.

\bibitem{karabag2021smooth}
Mustafa~O. Karabag, Cyrus Neary, and Ufuk Topcu, \emph{Smooth convex
  optimization using sub-zeroth-order oracles}, Proceedings of the AAAI
  Conference on Artificial Intelligence \textbf{35} (2021), no.~5, 3815--3822,
  \mbox{\href{http://arxiv.org/abs/arxiv:2103.00667}{arxiv:2103.00667}}.

\bibitem{Li2022Quantum}
Tongyang Li and Ruizhe Zhang, \emph{Quantum speedups of optimizing
  approximately convex functions with applications to logarithmic regret
  stochastic convex bandits}, to appear in the Advances in Neural Information
  Processing Systems, 2022,
  \mbox{\href{http://arxiv.org/abs/arXiv:2209.12897}{arXiv:2209.12897}}.

\bibitem{liu2022quantum}
Yizhou Liu, Weijie~J. Su, and Tongyang Li, \emph{On quantum speedups for
  nonconvex optimization via quantum tunneling walks}, 2022,
  \mbox{\href{http://arxiv.org/abs/arxiv:2209.14501}{arxiv:2209.14501}}.

\bibitem{liu2021rigorous}
Yunchao Liu, Srinivasan Arunachalam, and Kristan Temme, \emph{A rigorous and
  robust quantum speed-up in supervised machine learning}, Nature Physics
  \textbf{17} (2021), no.~9, 1013--1017,
  \mbox{\href{http://arxiv.org/abs/arXiv:2010.02174}{arXiv:2010.02174}}.

\bibitem{lloyd1996universal}
Seth Lloyd, \emph{Universal quantum simulators}, Science \textbf{273} (1996),
  no.~5278, 1073.

\bibitem{low2017optimal}
Guang~Hao Low and Isaac~L. Chuang, \emph{Optimal {H}amiltonian simulation by
  quantum signal processing}, Physical Review Letters \textbf{118} (2017),
  no.~1, 010501,
  \mbox{\href{http://arxiv.org/abs/arXiv:1606.02685}{arXiv:1606.02685}}.

\bibitem{low2019hamiltonian}
\bysame, \emph{Hamiltonian simulation by qubitization}, Quantum \textbf{3}
  (2019), 163,
  \mbox{\href{http://arxiv.org/abs/arXiv:1610.06546}{arXiv:1610.06546}}.

\bibitem{lu2020quantum}
Sirui Lu, Lu-Ming Duan, and Dong-Ling Deng, \emph{Quantum adversarial machine
  learning}, Physical Review Research \textbf{2} (2020), no.~3, 033212,
  \mbox{\href{http://arxiv.org/abs/arXiv:2001.00030}{arXiv:2001.00030}}.

\bibitem{lugosi2019mean}
G{\'a}bor Lugosi and Shahar Mendelson, \emph{Mean estimation and regression
  under heavy-tailed distributions: A survey}, Foundations of Computational
  Mathematics \textbf{19} (2019), no.~5, 1145--1190,
  \mbox{\href{http://arxiv.org/abs/arXiv:1906.04280}{arXiv:1906.04280}}.

\bibitem{Nayak1999Quantum}
Ashwin Nayak and Felix Wu, \emph{The quantum query complexity of approximating
  the median and related statistics}, Proceedings of the thirty-first annual
  ACM symposium on Theory of computing, pp.~384--393, 1999,
  \mbox{\href{http://arxiv.org/abs/arXiv:quant-ph/9804066}{arXiv:quant-ph/9804066}}.

\bibitem{nesterov2003introductory}
Yurii Nesterov, \emph{Introductory lectures on convex optimization: A basic
  course}, vol.~87, Springer Science \& Business Media, 2003.

\bibitem{nesterov2006cubic}
Yurii Nesterov and Boris~T. Polyak, \emph{Cubic regularization of {N}ewton
  method and its global performance}, Mathematical Programming \textbf{108}
  (2006), no.~1, 177--205.

\bibitem{Nielsen2010Quantum}
Michael~A. Nielsen and Isaac~L. Chuang, \emph{Quantum {{Computation}} and
  {{Quantum Information}}}, {Cambridge University Press}, {Cambridge}, 2010.

\bibitem{Preskill2018NISQ}
John Preskill, \emph{Quantum computing in the {NISQ} era and beyond}, {Quantum}
  \textbf{2} (2018), 79,
  \mbox{\href{http://arxiv.org/abs/arXiv:1801.00862}{arXiv:1801.00862}}.

\bibitem{rains1999monotonicity}
Eric~M. Rains, \emph{Monotonicity of the quantum linear programming bound},
  IEEE Transactions on Information Theory \textbf{45} (1999), no.~7,
  2489--2492,
  \mbox{\href{http://arxiv.org/abs/arXiv:quant-ph/9802070}{arXiv:quant-ph/9802070}}.

\bibitem{risteski2016algorithms}
Andrej Risteski and Yuanzhi Li, \emph{Algorithms and matching lower bounds for
  approximately-convex optimization}, Advances in Neural Information Processing
  Systems, vol.~29, 2016.

\bibitem{roy2020escaping}
Abhishek Roy, Krishnakumar Balasubramanian, Saeed Ghadimi, and Prasant
  Mohapatra, \emph{Escaping saddle-point faster under interpolation-like
  conditions}, Advances in Neural Information Processing Systems \textbf{33}
  (2020), 12414--12425.

\bibitem{singer2015information}
Yaron Singer and Jan Vondr{\'a}k, \emph{Information-theoretic lower bounds for
  convex optimization with erroneous oracles}, Advances in Neural Information
  Processing Systems, vol.~28, 2015.

\bibitem{sun2019optimization}
Ruoyu Sun, \emph{Optimization for deep learning: theory and algorithms}, 2019,
  \mbox{\href{http://arxiv.org/abs/arXiv:1912.08957}{arXiv:1912.08957}}.

\bibitem{tripuraneni2018stochastic}
Nilesh Tripuraneni, Mitchell Stern, Chi Jin, Jeffrey Regier, and Michael~I.
  Jordan, \emph{Stochastic cubic regularization for fast nonconvex
  optimization}, Advances in Neural Information Processing Systems,
  pp.~2899--2908, 2018,
  \mbox{\href{http://arxiv.org/abs/arXiv:1711.02838}{arXiv:1711.02838}}.

\bibitem{van2019improvements}
Joran van Apeldoorn and Andr{\'a}s Gily{\'e}n, \emph{Improvements in quantum
  sdp-solving with applications}, 46th International Colloquium on Automata,
  Languages, and Programming (ICALP 2019), Schloss Dagstuhl-Leibniz-Zentrum
  fuer Informatik, 2019,
  \mbox{\href{http://arxiv.org/abs/arXiv:1804.05058}{arXiv:1804.05058}}.

\bibitem{van2020convex}
Joran van Apeldoorn, Andr{\'a}s Gily{\'e}n, Sander Gribling, and Ronald
  de~Wolf, \emph{Convex optimization using quantum oracles}, Quantum \textbf{4}
  (2020), 220,
  \mbox{\href{http://arxiv.org/abs/arXiv:1809.00643}{arXiv:1809.00643}}.

\bibitem{van2020quantum}
\bysame, \emph{Quantum {SDP}-solvers: Better upper and lower bounds}, Quantum
  \textbf{4} (2020), 230,
  \mbox{\href{http://arxiv.org/abs/arXiv:1705.01843}{arXiv:1705.01843}}.

\bibitem{vapnik1991principles}
Vladimir Vapnik, \emph{Principles of risk minimization for learning theory},
  Advances in Neural Information Processing Systems, vol.~4, 1991.

\bibitem{xu2017neon}
Yi~Xu, Rong Jin, and Tianbao Yang, \emph{{NEON}+: Accelerated gradient methods
  for extracting negative curvature for non-convex optimization}, 2017,
  \mbox{\href{http://arxiv.org/abs/arXiv:1712.01033}{arXiv:1712.01033}}.

\bibitem{xu2018first}
\bysame, \emph{First-order stochastic algorithms for escaping from saddle
  points in almost linear time}, Advances in Neural Information Processing
  Systems, pp.~5530--5540, 2018,
  \mbox{\href{http://arxiv.org/abs/arXiv:1711.01944}{arXiv:1711.01944}}.

\bibitem{zhang2021quantum}
Chenyi Zhang, Jiaqi Leng, and Tongyang Li, \emph{Quantum algorithms for
  escaping from saddle points}, Quantum \textbf{5} (2021), 529,
  \mbox{\href{http://arxiv.org/abs/arXiv:2007.10253}{arXiv:2007.10253}}.

\bibitem{zhang2021escape}
Chenyi Zhang and Tongyang Li, \emph{Escape saddle points by a simple
  gradient-descent based algorithm}, Advances in Neural Information Processing
  Systems \textbf{34} (2021), 8545--8556,
  \mbox{\href{http://arxiv.org/abs/arXiv:2111.14069}{arXiv:2111.14069}}.

\bibitem{zhang2022quantum}
\bysame, \emph{Quantum lower bounds for finding stationary points of nonconvex
  functions}, 2022.

\bibitem{zhang2022zeroth}
Hualin Zhang, Huan Xiong, and Bin Gu, \emph{Zeroth-order negative curvature
  finding: Escaping saddle points without gradients}, arXiv preprint
  arXiv:2210.01496 (2022),
  \mbox{\href{http://arxiv.org/abs/arxiv:2210.01496}{arxiv:2210.01496}}.

\bibitem{zhang2017hitting}
Yuchen Zhang, Percy Liang, and Moses Charikar, \emph{A hitting time analysis of
  stochastic gradient {L}angevin dynamics}, Conference on Learning Theory,
  pp.~1980--2022, 2017,
  \mbox{\href{http://arxiv.org/abs/arXiv:1702.05575}{arXiv:1702.05575}}.

\end{thebibliography}

%%%%%%%%%%%%%%%%%%%%%%%%%%%%%%%%%%%%%%%%%%%%%%%%%%%%%%%%%%%%%%%%%%%%%%%%%%%%%%
\newpage
\appendix

\section{Auxiliary Lemmas}\label{append:existing-lemmas}
\begin{lemma}[Lemma 3, Ref.~\cite{carmon2021lower}]\label{lem:FT-large-gradient}
Consider the function $\bar{F}_{T;\mu}\colon\R^{T+1}\to\R$ defined in Eq.~\eqn{FT-defn}, suppose the parameter $\mu$ satisfies $\mu\leq 1$. Then for any $\x\in\R^{T+1}$ such that \footnote{The condition below is a bit different from the original condition in Lemma 3 of \cite{carmon2021lower}, which is $x_T=x_{T+1}=0$. Nevertheless, the following stricter conditions can be achieved with only minor modifications to the original proof.}
\begin{align}
    |x_T|,|x_{T+1}|\leq\frac{0.1}{T+1/\sqrt{\mu}},
\end{align}
we have
\begin{align}
\big\|\nabla\bar{F}_{T;\mu}(\x)\big\|\geq \mu^{3/4}/8.
\end{align}
\end{lemma}

\begin{lemma}[Lemma 4, Ref.~\cite{carmon2021lower}]\label{lem:FT-conditions}
The function $\bar{F}_{T;\mu}(\x)$ defined in Eq.~\eqn{FT-defn} satisfies the following. 
\begin{enumerate}
\item $\bar{F}_{T;\mu}(\0)-\inf_\x\bar{F}_{T;\mu}(\x)\leq\frac{\sqrt{\mu}}{2}+10\mu T$;
\item For $\mu\leq 1$ and every $p\geq 1$, the $p$-th order derivatives of $\bar{F}_{T;\mu}$ are $(\mathbb{I}\{p=1\}+\mu)\ell_p$-Lipschitz continuous, where $\ell_p\leq\exp\big(\frac{3p}{2}\log p+cp\big)$ for a numerical constant $c<\infty$.
\end{enumerate}
\end{lemma}

\begin{lemma}[Lemma 3 and Lemma 4, Ref.~\cite{carmon2021lower}]\label{lem:FT-square-conditions}
Consider the function $\bar{\mathfrak{F}}_{T;\mu}\colon\R^{T+1}\to\R$ defined in Eq.~\eqn{FT-square-defn}, suppose the parameter $\mu$ satisfies $\mu\leq 1$. Then, there exist positive constants $\alpha,\beta$ such that\footnote{The formula of the function $\bar{\mathfrak{F}}_{T;\mu}$ is a bit different from the original function considered in Lemma 3 and Lemma 4 of Ref.~\cite{carmon2021lower}. Nevertheless, this lemma can be proved via only minor modifications to the original proof.}
\begin{enumerate}
\item For any $\x\in\R^{T+1}$ satisfying
\begin{align}
|x_T|,|x_{T+1}|\leq\frac{0.05}{T+1/\sqrt{\mu}},
\end{align}
its gradient satisfies
\begin{align}
\|\nabla\bar{\mathfrak{F}}_{T;\mu}(\x)\|\geq\mu^{3/4}/16;
\end{align}
\item $\bar{\mathfrak{F}}_{T;\mu}(\0)-\inf_\x\bar{\mathfrak{F}}_{T;\mu}(\x)\leq\frac{\sqrt{\mu}}{2}+10\mu T$;
\item For every $p>1$, the $p$-th order derivatives of $\bar{\mathfrak{F}}_{T;\mu}$ are $(2\mathbb{I}\{p=1\}+\mu)\ell_p$-Lipschitz continuous, where $\ell_p\leq\exp\big(\frac{3p}{2}\log p+cp\big)$ for a numerical constant $c<\infty$.
\end{enumerate}
\end{lemma}

\begin{lemma}[Proposition 14, Ref.~\cite{garg2020no}]\label{lem:vector-from-sphere}
Let $\x\in\mathbb{B}(\0,1)$. Then for a $d$-dimensional random unit vector $\u$ and all $c>0$,
\begin{align}
\Pr_\u(|\<\x,\u\>|\geq c)\leq 2e^{-dc^2/2}.
\end{align}
\end{lemma}

\begin{lemma}[Lemma 9, Ref.~\cite{zhang2021quantum}]\label{lem:zhanglem9}
Let $H_1$ and $H_2$ be two Hermitian operators and $H$ be the sum of two operators. For any $t>0$ and state vector $\ket{\varphi}$, we have
\begin{align}
\norm{e^{-iH_1t}e^{-iH_2t}\ket{\varphi}-e^{-iHt}\ket{\varphi}}\leq\frac{t^2}{2}\sup_{\tau_1,\tau_2\in[0,t]}\norm{[H_1,H_2]e^{-iH_2\tau_2}e^{-iH_1\tau_1}\ket{\varphi}}.
\end{align}
\end{lemma}

\begin{lemma}[Corollary 1, Ref.~\cite{zhang2021quantum}]\label{lem:zhangcor1}
Consider a quadratic function of form $F_q=(\x-\x_s)H(\x-\x_s)+f_0$ for Hermitian function $H$ and constant $f_0$, and the Shro\"odinger equation
\begin{align}
i\frac{\partial}{\partial t}\Phi=\left[-\frac{r_0^2}{2}\Delta+\frac{1}{r_0^2}F_q\right]\Phi,    
\end{align}
with periodic boundary conditions and initial state in \eq{ground_state_Phi0}. We have
\begin{align}
\norm{\nabla\Phi(t)}\leq C\sqrt{\frac{d}{r_0}}(\log t)^\alpha
\end{align}
for some constant $\alpha$ and $C$.
\end{lemma}

%%%%%%%%%%%%%%%%%%%%%%%%%%%%%%%%%%%%%%%%%%%%%%%%%%%%%%%%%%%%%%%%%%%%%%%%%%%%%%%%%%%%%%%%%%%%%%%%%%%%%%%%%%%%%%%%%%%%%%%%%%%%%%%%%%%%%%%%%

\section{Technical Lemmas}\label{append:technical-lemmas}
\begin{lemma}[Cannot guess stationary point]\label{lem:cannot-guess}
Let $k<T$ be a positive in $\big\{\u^{(1)},\ldots,\u^{(k)}\big\}$ be a set of orthonormal vectors. Let $\{\u^{k+1},\ldots,\u^{T}\}$ be chosen uniformly at random from $\spn\big(\u^{(1)},\ldots,\u^{(k)}\big)^{\perp}$ such that all columns of the matrix $U=\big[\u^{(1)},\ldots,\u^{(T+1)}\big]$ forms a set of orthonormal vectors. Then,
\begin{align}
    \forall\x\in\mathbb{B}\big(\0,2\sigma \sqrt{T}\big),\quad\Pr_{\{\u^{k+1},\ldots,\u^{T+1}\}}\big[\|\nabla\tilde{F}_{T;U}(\x)\|\leq\lambda\sigma\mu^{3/4}/8\big]\leq 8Te^{-dr_0^2/(8T)},
\end{align}
for the function $\tilde{F}_{T;U}\colon\R^d\to\R$ defined in Eq.~\eqn{tildeF-defn}, given that the parameters $r>1$ and
\begin{align}
    r_0\leq\frac{0.2\sqrt{T}}{T+1/\sqrt{\mu}}.
\end{align}
\end{lemma}

\begin{proof}
By \lem{vector-from-sphere}, with probability at least $1-8Te^{-dr_0^2/(8T)}$, we have
\begin{align}\label{eqn:small-uT}
    |\<\x,\u^{(T)}/\sigma\>|,|\<\x,\u^{(T+1)}/\sigma\>|\leq\frac{r_0}{2\sqrt{T}}\leq\frac{0.1}{T+1/\sqrt{\mu}}
\end{align}
at the same time. Then by \lem{FT-large-gradient}, we have
\begin{align}
    \|\nabla \tilde{F}_{T;U}(\x)\|\leq \lambda\sigma\|\nabla\bar{F}_{T;\mu;r}(U^\top\x/\sigma)\|\leq\lambda\sigma\mu^{3/4}/8.
\end{align}
\end{proof}

\begin{lemma}\label{lem:r0-max-differences}
Consider the functions $\tilde{\mathfrak{F}}_{T;U}\colon\R^d\to\R$
and $\mathfrak{f}_{r_0;T;U}\colon\R^d\to\R$
defined in Eq.~\eqn{tildeFT-square-defn} and Eq.~\eqn{f-r0-square-defn}, respectively. We have 
\begin{align}
    \max_{\x}|\tilde{\mathfrak{F}}_{T;U}(\x)-\mathfrak{f}_{r_0;T;U}(\x)|
    \leq\lambda\sigma^2(50r_0^2T+2\alpha r_0\sqrt{T}),
\end{align}
and
\begin{align}
    \max_\x\big\|\nabla\tilde{\mathfrak{F}}_{T;U}-\nabla\mathfrak{f}_{r_0;T;U}(\x)\big\|\leq 6\lambda\sigma r_0\sqrt{T}.
\end{align}
where we regard $\nabla\mathfrak{f}_{r_0;T;U}$ to be $\0$ on the sphere $\mathbb{S}(\0,\alpha\sigma\sqrt{T})$ of each hypercube in $\R^d$, given that $\mu,\epsilon\leq 1$ and $r_0<\sigma$.
\end{lemma}
\begin{proof}
Without loss of generality, we prove this lemma with the input domain being the hypercube $[-\pi\sigma\mathcal{L}/2,\pi\sigma\mathcal{L}/2]^d$ centered at $\0$. Denote the vector $\y\in\R^{T+1}$ to be
\begin{align}
\y\coloneqq \big(\gamma_{\alpha}(\<\u^{(1)},\x/\sigma\>),\ldots,\gamma_{\alpha}(\<\u^{(T+1)},\x/\sigma\>)\big)^\top,
\end{align}
which satisfies $\|\y\|\leq\alpha\sqrt{T}$. For the convenience of notations, we denote $y_0\coloneqq 0$. Then,
\begin{align}
    \max_\x\big|\tilde{\mathfrak{F}}_{T;U}(\x)-\mathfrak{f}_{r_0;T;U}(\x)\big|=\lambda\sigma^2\max_{\y\in\mathbb{B}(\0,\alpha\sqrt{T})}\big|\bar{F}_{T;\mu}(\y^{\prog})-\bar{F}_{T;\mu}(\y)\big|,
\end{align}
where
\begin{align}
    \y^{\prog}\coloneqq\big(\gamma_\alpha(\<\x/\sigma,\u^{(1)}\>),\ldots,\gamma_\alpha(\<\x/\sigma,\u^{(\prog_{ r_0}(\gamma_\alpha(\x/\mu))+1)}\>),0,\ldots,0\big)^\top.
\end{align}
Note that
\begin{align}
    &\max_{\y\in\mathbb{B}(\0,\alpha\sqrt{T})}\big|\bar{F}_{T;\mu}(\y^{\prog})-\bar{F}_{T;\mu}(\y)\big|\\
    &\quad\qquad\leq\frac{1}{2}\max_{\y\in\mathbb{B}(\0,\alpha\sqrt{T})}\Big|\sum_{i=0}^{T}(y_i-y_{i+1})^2-(y^{\prog}_i-y^{\prog}_{i+1})^2\Big|\label{eqn:difference-second-term}\\
    &\quad\qquad+\mu\max_{\y\in\mathbb{B}(\0,2\alpha\sqrt{T})}\Big|\sum_{i=1}^{T+1}\big(\Gamma(y_i)-\Gamma(y^{\prog}_i)\big)\Big|,
\end{align}
where
\begin{align}
    \max_{\y\in\mathbb{B}(\0,\alpha\sqrt{T})}\Big|\sum_{i=1}^{T+1}\big(\Gamma(y_i)-\Gamma(y_i^{\prog})\big)\Big|\leq (T+1)(\Gamma(0)-\Gamma(r_0))\leq 40Tr_0^3,
\end{align}
and
\begin{align}
    (y_1-1)^2-(y_1^{\prog}-1)^2=0,\qquad\forall\y\in\mathbb{B}(\0,\alpha\sqrt{T}),
\end{align}
since $\prog_{r_0}(\gamma_{\alpha}(\x/\mu))+1\geq 1$ for all possible $\x$ and corresponding $\y$. As for the term in \eqn{difference-second-term}, we note that 
\begin{align}
    &\max_{\y\in\mathbb{B}(\0,\alpha\sqrt{T})}\Big|\sum_{i=1}^T(y_i-y_{i+1})^2-(y_i^{\prog}-y_{i+1}^{\prog})^2\Big|\\
    &\quad\qquad\leq 4Tr_0^2+2\alpha r_0\sqrt{T}.
\end{align}
Thus we can conclude that
\begin{align}
    \max_{\y\in\mathbb{B}(\0,\alpha\sigma\sqrt{T})}\big|\bar{F}_{T;\mu}(\delta_{r_0}(\y))-\bar{F}_{T;\mu}(\y)\big|
    &\leq50r_0^2T+2\alpha r_0\sqrt{T},
\end{align}
and
\begin{align}
    \max_\x\big|\tilde{\mathfrak{F}}_{T;U}(\x)-\mathfrak{f}_{r_0;T;U}\big|\leq\lambda\sigma^2(50r_0^2T+2\alpha r_0\sqrt{T}).
\end{align}
Similarly, we can observe that
\begin{align}
    \max_{\x}\big\|\nabla\tilde{\mathfrak{F}}_{T;U}(\x)-\nabla\mathfrak{f}_{r_0;T;U}(\x)\big\|\leq\lambda\sigma\max_{\y\in\mathbb{B}(\0,\alpha\sqrt{T})}\big\|\nabla\bar{F}_{T;\mu}(\delta_{r_0}(\y))-\nabla\bar{F}_{T;\mu}(\y)\big\|,
\end{align}
where
\begin{align}
    &\max_{\y\in\mathbb{B}(\0,\alpha\sqrt{T})}\big\|\nabla\bar{F}_{T;\mu}(\delta_{r_0}(\y))-\nabla\bar{F}_{T;\mu}(\y)\big\|\\
    &\quad\qquad\leq\frac{1}{2}\max_{\y\in\mathbb{B}(\0,\alpha\sqrt{T})}\big\|\nabla\cdot\sum_{i=0}^{T}\big[(y_i-y_{i+1})^2-(y_i^{\prog}-y_{i+1}^{\prog})^2\big]\big\|\label{eqn:first-difference-second-term}\\
    &\quad\qquad+\mu\max_{\y\in\mathbb{B}(\0,\alpha\sqrt{T})}\big\|\nabla\cdot\sum_{i=1}^{T+1}\big[\Gamma(y_i)-\Gamma(y_i^{\prog})\big]\big\|\label{eqn:first-difference-third-term},
\end{align}
where the term in \eqn{first-difference-third-term} satisfies
\begin{align}
    \max_{\y\in\mathbb{B}(\0,\alpha\sqrt{T})}\big\|\nabla\cdot\sum_{i=1}^{T+1}\big[\Gamma(y_i)-\Gamma(y_i^{\prog})\big]\big\|\leq2\sqrt{T}\cdot|\Gamma'(r_0)-\Gamma'(0)|\leq 2r_0\sqrt{T}
\end{align}
As for the term in \eqn{first-difference-second-term}, we first note that for any $0\leq i\leq T$ and any $\y\in\mathbb{B}(\0,\alpha\sqrt{T})$, we have
\begin{align}
    \frac{1}{2}\big\|\nabla\cdot[(y_i-y_{i+1})^2-(y_i^{\prog}-y_{i+1}^{\prog})^2]\big\|\leq 2\sqrt{2}r_0,
\end{align}
which leads to 
\begin{align}
    \max_{\y\in\mathbb{B}(\0,\alpha\sqrt{T})}\big\|\nabla\cdot\sum_{i=1}^{T}\big[(y_i-y_{i+1})^2-(y_i^{\prog}-y_{i+1}^{\prog})^2\big]\big\|\leq 4r_0\sqrt{T}.
\end{align}
Hence,
\begin{align}
    \max_{\y\in\mathbb{B}(\0,\alpha\sqrt{T})}\big\|\nabla\bar{F}_{T;\mu}(\delta_{r_0}(\y))-\nabla\bar{F}_{T;\mu}(\y)\big\|\leq 6r_0\sqrt{T},
\end{align}
by which we can conclude that
\begin{align}
    \max_{\x}\big\|\nabla\tilde{\mathfrak{F}}_{T;U}(\x)-\nabla\mathfrak{f}_{r_0;T;U}(\x)\big\|\leq 6\lambda\sigma r_0\sqrt{T}.
\end{align}
\end{proof}

%%%%%%%%%%%%%%%%%%%%%%%%%%%%%%%%%%%%%%%%%%%%%%%%%%%%%%%%%%%%%%%%%%%%%%%%%%%%%%%%%%%%%%%%%%%%%%%%%%%%%%%%%%%%%%%%%%%%%%%%%%%%%%%%%%%%%%%%%

\section{Perturbed Gradient Descent with Quantum Simulation and Gradient Estimation}\label{append:PGDQSGC}
In this section, we consider an alternative version of \algo{PGDQGC} that has a faster convergence rate in some cases. Inspired by Ref.~\cite{zhang2021quantum}, we replace the uniform perturbation in \algo{PGDQGC} with quantum simulation. We consider the scaled evolution under Schr\"odinger equation
\begin{align}\label{eq:Schrodinger}
i\frac{\partial}{\partial t}\Phi=\left[-\frac{r_0^2}{2}\Delta+\frac{1}{r_0^2}f\right]\Phi,
\end{align}
where $\Phi$ is a wave function in $\mathbb{R}^d$, $\Delta$ is the Laplacian operator, $r_0$ is the scaling parameter, and $f$ is the potential of the evolution. To construct a quantum algorithm using this evolution, quantum simulations are required. There is rich literature on the cost of quantum simulations~\cite{berry2007efficient,berry2015hamiltonian,childs2017note,lloyd1996universal,low2017optimal,low2019hamiltonian}. Here, we introduce the following theorem concerning the cost of simulating \eq{Schrodinger} using zeroth-order oracle $F$, which was originally proposed in Ref.~\cite{zhang2021quantum}.
\begin{lemma}[Lemma 2, Ref.~\cite{zhang2021quantum}]\label{lem:CostQS}
Let $F(\x)\colon\mathbb{R}^d\to\mathbb{R}$ be a real-valued function that has a saddle point at $\x=0$ such that $f(\0)=0$. Consider the (scaled) Schr\"odinger equation in \eq{Schrodinger} defined on the domain $\Omega=\{\x\in\mathbb{R}^d\colon\norm{\x}\leq M\}$ with periodic boundary condition, where $M>0$ is the diameter specified later. Given the noiseless zeroth-order oracle $U_f(\ket{\x}\otimes\ket{0})=\ket{\x}\otimes\ket{f(\x)}$ and an arbitrary initial state. The evolution for time $t>0$ can be simulated using $\tO(t\log d\log^2(t/\epsilon))$ queries to $U_f$, where $\epsilon$ is the simulation precision.
\end{lemma}

Notice that $F$ is assumed to be Hessian Lipschitz in both \assume{ZeroProb} and \assume{FirstProb}, we can approximate the function value near a saddle point. The approximation is more accurate on a ball with radius $r_0$ centered at this saddle point. We scale the initial distribution and the Schr\"odinger equation to be localized in term of $r_0$ and results in \algo{QS}, which is originally proposed in Ref.~\cite{zhang2021quantum}. 

\begin{algorithm}[htbp]
\caption{QuantumSimulation($\tilde{\x},r_{0},t_{e},f(\cdot)$).}
\label{algo:QS}
Evolve a Gaussian wave packet in the potential field $f$, with its initial state being:
\begin{align}\label{eq:ground_state_Phi0}
\Phi_{0}(\x)=\Big(\frac{1}{2\pi}\Big)^{n/4}\frac{1}{r_{0}^{n/2}}\exp(-(\x-\tilde{\x})^{2}/4r_{0}^{2});
\end{align}
Simulate such evolution in potential field $f$ under the Schr\"odinger equation for time $t_{e}$\;
Measure the position of the wave packet and output the outcome.
\end{algorithm}

\algo{QS} is the main building block of the quantum implementation of perturbation in PGD. It can effectively reduce the iteration number compared to the classical perturbations in \algo{PGDQGC}~\cite{zhang2021quantum} for some functions. To achieve a better performance than \thm{ZeroJordanF}, we have to add some constraints on the target and the noisy function $(F,f)$. Specifically, we consider the following setting. 
\begin{assumption}\label{assume:ZeroProbQS}
The underlying target function $F$ is $B$-bounded, $\ell$-smooth, and $\rho$-Hessian Lipschitz. We can query a noisy function $f$ that is twice differentiable. We assume
\begin{align}
\nu&=\sup_{\x}\norm{F(\x)-f(\x)}_\infty=\tO\left(\frac{\epsilon^6}{d^4}\right),\\
\tnu&=\sup_{\x}\norm{\nabla F(\x)-\nabla f(\x)}_\infty=\ell_f M=O\left(d^{-3}\right),\\
\hat{\nu}&=\sup_{\x}\norm{\nabla^2 F(\x)-\nabla^2 f(\x)}_\infty=\rho_f M=O\left(d^{-3}\right),
\end{align}
where $\ell_f$ and $\rho_f$ are arbitrary constants, and $M=O(d^{-3})$ is some value to be fixed later.
\end{assumption}
By using quantum simulations to implement perturbations, we propose the following \algo{PGDQGCQS} that can effectively find an $\epsilon$-SOSP of $F$ using queries to noisy $f$ for function pair $(F,f)$ in \assume{ZeroProbQS} with high probability.
\begin{algorithm}[htbp]
\caption{Perturbed Gradient Descent with Quantum Gradient Computation.}
\label{algo:PGDQGCQS}
\begin{algorithmic}[1]
\REQUIRE $\x_0$, learning rate $\eta$, noise ratio $r$
\FOR{$t=0,1,\ldots,T$}
\STATE Apply \lem{JordanPerf} to compute an estimate $\tnabla F(\x)$ of $\nabla F(\x)$
\IF{$\norm{\tnabla F(\x_{t})}\leq \epsilon$}
\STATE $\xi\sim$QuantumSimulation$\big(\x_{t},r_{0},\ut,f(\x)-\langle\tnabla F(\x_t),\x-\x_t\rangle\big)$
\STATE $\Delta_t\leftarrow\frac{2\xi}{3\|\xi\|}\sqrt{\frac{\rho^*}{\epsilon}}$
\STATE $\x_{t}\leftarrow\mathop{\arg\min}_{\zeta\in\left\{\x_t+\Delta_t,\x_t-\Delta_t\right\}}f(\zeta)$
\ENDIF
\STATE $\x_{t+1}\leftarrow \x_t-\eta\tnabla F(\x)$
\ENDFOR
\end{algorithmic}
\end{algorithm}

\algo{PGDQGCQS} has the following performance guarantee:
\begin{theorem}\label{thm:ZeroJordanQS}
Suppose we have a target function $F$ and its noisy evaluation $f$ satisfying \assume{ZeroProbQS} with $\nu\leq\tO(\delta^2\epsilon^6/d^4)$. Then with probability at least $1-\delta$, \algo{PGDQGCQS} can find an $\epsilon$-SOSP of $F$ satisfying \eq{SOSPDef}, using
\begin{align}
\tO\left(\frac{\ell B}{\epsilon^2}\cdot\log^2 d\right)
\end{align}
queries to $U_f$, under the following parameter choices:
\begin{align}
&\ell'=\ell+\frac12\ell_f,\quad\rho'=\rho+\rho_f,\quad\eta=\frac{1}{\ell'},\quad\delta_0=\min\left\{\frac{\delta}{162B}\frac{\epsilon^3}{\rho'},\frac{\eta\epsilon^2\delta}{16B}\right\},\quad\uf=\frac{2}{81}\sqrt{\frac{\epsilon^3}{\rho'}},\\
&\ut=\frac{8}{(\rho'\epsilon)^{1/4}}\log\left(\frac{\ell'}{\delta_0\sqrt{\rho'\epsilon}}(d+2\log(\frac{3}{\delta_0}))\right),\quad r_0=\frac{4c_r^3}{9\ut^4}\left(\frac{\delta_0}{3}\cdot\frac{1}{d^{3/2}+2c_0d\ell'(\log \ut)^\alpha}\right)^2,
\end{align}
where $c_0$, $\alpha$, and $c_r$ are absolute constants specified in the proof.
\end{theorem}

Before proving \thm{ZeroJordanQS}, we first consider the effectiveness of quantum simulations for adding perturbations. We focus on the scenarios with $\epsilon\leq\ell^2/\rho$, which is the standard assumption adopted in Ref.~\cite{jin2018accelerated}. The local landscape in this case ``flat" and the Hessian has only a small spectral radius. The classical gradient descent will move slowly while the variance of the probability distribution corresponding
to the Gaussian wavepacket still has a large increasing rate. If we evolve the Gaussian wavepacket for a long enough time period and measure its position, we will obtain a vector that indicates a negative curvature direction with high probability.

However, \algo{QS} using quantum simulation and noisy oracle in \assume{ZeroProbQS} suffers two deviation terms from the ideal Gaussian evolution: the deviation of $F$ from quadratic potential and the noise of $f$ from $F$. We have to bound the resulting deviation on the distribution from the perfect Gaussian wavepacket. We specify the constant $c_r$ as the ratio between the wavepacket variance and the radius of the simulation region. By choosing a small enough $c_r$, the simulation region is much larger than the range of the wavepacket. As the function $F$ is $\ell$-smooth, the spectral norm of the Hessian matrix is upper bounded by constant $\ell$. The radius $M$ of the simulation region is chosen as
\begin{align}\label{eq:MChoice}
M=\frac{r_0}{C_r}=\frac{4c_r^2}{9\ut^4}\left(\frac{\delta_0}{3}\cdot\frac{1}{d^{3/2}+2c_0 d\ell(\log \ut)^\alpha}\right)^2\leq 1.
\end{align}

By choosing the above $M$, we can reach the following lemma.
\begin{lemma}\label{lem:GaussDeviation}
Under the setting of \assume{ZeroProbQS} and \thm{ZeroJordan}, let $H$ be the Hessian matrix of $F$ at the saddle point $\x_s$, and define $F_q(\x)=f(\x_s)+(\x-\x_s)^\top H(\x-\x_s)$ to be the quadratic approximation of $F$ near $\x_s$. We denote the measurement outcome from \algo{QS} with noisy function $f$ and evolution $t_e$ as a random variable $\xi$, and the measurement outcome from the ideal potential $F_q$ and the same evolution time $t_e$ as another random variable $\xi'$. We define $\mathbb{P}_\xi$ and $\mathbb{P}_{\xi'}$ to be the distribution of $\xi$ and $\xi'$. If the quantum wavepacket is confined to a hypercube with regions length $M$, then
\begin{align}
TV(\mathbb{P}_\xi,\mathbb{P}_{\xi'})\leq\left(\frac{\sqrt{d}\rho'}{2}+\frac{2c_f\ell'}{\sqrt{r_0}}(\log t_e)^\alpha\right)\frac{dMt_e^2}{2},
\end{align}
where $TV(\cdot,\cdot)$ denotes the total variation distance, $\alpha$ is an absolute constant, and $c_f$ is an $F$-related constant. 
\end{lemma}
\begin{proof}
We first define the following notations:
\begin{align}
A=-\frac{r_0^2}{2}\Delta,\quad B=\frac{1}{r_0^2}f,\quad B'=\frac{1}{r_0^2}F_q,\\
H=A+B,\quad H'=A+B',\quad E=H-H'=\frac{1}{r_0^2}(f-F_q).
\end{align}
We denote $\ket{\Phi(t)}=e^{-iHt}\ket{\Phi_0}$ and $\ket{\Phi'(t)}=e^{-iH't}\ket{\Phi_0}$ be the wave functions at time $t$ for two different Hamiltonians $H$ and $H'$. By 
\lem{zhanglem9}, we have
\begin{align}
\norm{e^{iEt_e}\ket{\Phi'(t_e)}-\ket{\Phi(t_e)}}&\leq\frac{t_e^2}{2}\sup_{\tau_1,\tau_2\in[0,t_e]}\norm{[H',E]e^{-iE\tau_2}e^{-iH'\tau_1}\ket{\Phi_0}}\\
&=\frac{t_e^2}{2}\sup_{\tau_1\in[0,t_e]}\norm{[H',E]e^{-iH'\tau_1}\ket{\Phi_0}}.
\end{align}
Denoting $\ket{\Psi(\tau_1)}=e^{-iH'\tau_1}\ket{\Phi_0}$, we have
\begin{align}
\sup_{\tau_1\in[0,t_e]}\norm{[H',E]\Psi(\tau_1)}&=\frac12\sup_{\tau_1\in[0,t]}\norm{[-\Delta,f-F_q]\Psi(\tau_1)}\\
&=\frac12\sup_{\tau_1\in[0,t_e]}\norm{-\Delta(f-F_q)\Psi(\tau_1)-2\nabla(f-F_q)\cdot\nabla\Psi(\tau_1)}\\
&\leq\frac12\norm{\Delta(f-F_q)}_\infty+\norm{\nabla(f-F_q)}_\infty\norm{\nabla\Psi(\tau_1)}.
\end{align}
The first equality follows from $[H',E]=[A+B',E]$ and $B'$ commutes with $E$. The second equality follows from $[-\Delta,g]\varphi=-(\Delta g)\varphi-2\nabla g\cdot\nabla\varphi$ for any smooth function $\varphi$ and $g$. As we assume $F$ is $\rho$-Hessian Lipschitz, we can deduce that
\begin{align}
\abs{\Delta(f(\x)-F_q(\x))}&=\abs{\tr(\nabla^2 f(\x)-\nabla^2 F(\x))}+\abs{\tr(\nabla^2 F(\x)-\nabla^2 F_q(\x))}\\
&=\abs{\tr(\nabla^2 f(\x)-\nabla^2 F(\x))}+\abs{\tr(\nabla^2 F(\x)-\nabla^2 F(\x_s))}\\
&\leq d\norm{\nabla^2 f(\x)-\nabla^2 F(\x)}+d\norm{\nabla^2 F(\x)-\nabla^2 F(\x_s)}\\
&\leq d^{3/2}(\rho+\rho_f)M\\
&=d^{3/2}\rho'M.
\end{align}
Next, we bound the term on the gradient of $f-F_q$:
\begin{align}
\norm{\nabla(f-F_q)}_\infty&\leq\sup_\x\norm{\nabla f(\x)-\nabla F_q(\x)}\\
&=\sup_\x\norm{\nabla f(\x)-\nabla F(\x)}+\sup_\x\norm{\nabla F(\x)-H(\x-\x_s)}\\
&\leq \sup_\x\norm{\nabla f(\x)-\nabla F(\x)}+\sup_\x\norm{\nabla F(\x)}+\sup_\x\norm{H(\x-\x_s)}\\
&\leq(2\ell+\ell_f)Md^{-1/2}\\
&=2\ell'Md^{-1/2}
\end{align}
The upper bound for $\sup_{\tau_1\in[0,t_e]}\norm{\nabla\Psi(\tau_1)}$ is given by \lem{zhangcor1}. Combining the above bounds, we obtain
\begin{align}
\norm{e^{iEt_e}\ket{\Phi'(t_e)}-\ket{\Phi(t_e)}}\leq\left(\frac{\sqrt{d}\rho'}{2}+\frac{2c_f\ell'}{\sqrt{r_0}}(\log t_e)^\alpha\right)\frac{dMt_e^2}{2}.
\end{align}
In the following part, we denote $\Psi'$ for $\Psi'(t_e)$ and $\ket{\Psi''}=e^{-iEt_e}\ket{\Psi'}$. We observe that $\abs{\Psi'}^2=\abs{\Psi''}^2$ as $e^{-iEt_e}$ is a scalar function with modulus $1$. Thus
\begin{align}
TV(\mathbb{P}_\xi,\mathbb{P}_{\xi'})&=TV(\abs{\Psi}^2,\abs{\Psi''}^2)\\
&=\frac12\int_{\x}\abs{\Psi\Psi^\dagger-\Psi''\Psi^{''\dagger}}d\x\\
&\leq\frac12\int_{\x}\abs{(\Psi-\Psi'')\Psi^\dagger}d\x+\frac12\int_{\x}\abs{\Psi''(\Psi-\Psi'')^\dagger}d\x\\
&\leq(\frac12\int_{\x}\abs{\Psi-\Psi''}^2d\x)^{1/2}\\
&\leq\left(\frac{\sqrt{d}\rho'}{2}+\frac{2c_f\ell'}{\sqrt{r_0}}(\log t_e)^\alpha\right)\frac{dMt_e^2}{2}.
\end{align}
\end{proof}

\lem{GaussDeviation} indicates that the actual perturbation given by quantum simulation $\xi\sim\mathbb{P}_\xi$ deviates from the ideal Gaussian case $\xi'\sim\mathbb{P}_{\xi'}$ for at most $\tO(Md^{3/2}t_e^2)$. In \algo{PGDQGCQS} with $t_e=\ut=O(\log d)$, such deviation can be bounded for the choice of $M$ in \eq{MChoice}. Based on \lem{GaussDeviation}, we reach the following lemma.
\begin{lemma}[Adaptive version of Proposition 1, Ref.~\cite{zhang2021quantum}]\label{lem:zhangprop1}
Suppose $(F,f)$ satisfies \assume{ZeroProbQS}. For arbitrary $\delta_0$, we choose the following parameters:
\begin{align}
&\eta=\frac{1}{\ell'},\ut=\frac{8}{(\rho'\epsilon)^{1/4}}\log\left(\frac{\ell'}{\delta_0\sqrt{\rho'\epsilon}}(d+2\log(\frac{3}{\delta_0}))\right),\\
&\uf=\frac{2}{81}\sqrt{\frac{\epsilon^3}{\rho'}},\quad r_0=\frac{4c_r^3}{9\ut^4}\left(\frac{\delta_0}{3}\cdot\frac{1}{d^{3/2}+2c_0d\ell'(\log \ut)^\alpha}\right)^2,
\end{align}
where $\rho'$, $\ell'$, $c_r$, $c_0$, and $\alpha$ are the same with \thm{ZeroJordanQS}. Then, for an saddle point $\x_s$ with $\norm{F(\x_s)}\leq\epsilon$ and $\lambda_{\min}(\nabla^2 F(\x_s))\leq-\sqrt{\rho\epsilon}$, \algo{PGDQGCQS} provides a perturbation that decreases the function value for at least $\uf$ with probability at least $1-\delta_0$.
\end{lemma}

We now prove \thm{ZeroJordanQS}.
\begin{proof}
We set the zeroth-order noise bound $\nu\leq c_\nu\delta^2\epsilon^6/d^4$ for small enough $c_\nu$ and let the total total iteration number to be
\begin{align}
T=4\max\left\{\frac{2B}{\uf},\frac{4B}{\eta\epsilon^2}\right\}=\tO\left(\frac{B}{\epsilon^2}\cdot\log d\right).
\end{align}

We first consider the iteration number at saddle points $\x_t$ with $\norm{F(\x_t)}\leq\epsilon$ and $\lambda_{\min}(\nabla^2 F(\x_t))\leq-\sqrt{\rho\epsilon}$. According to \lem{zhangprop1}, each iteration of \algo{PGDQGCQS} in this case will decrease the function value for at least $\ut$. Under this assumption, \algo{QS} can be called for at most $T/2$ times, for otherwise the function value decreases greater than $2B\geq F(\x^*)-F(\x_0)$. The failure probability is bounded by
\begin{align}
81\sqrt{\frac{\rho}{\epsilon^3}}\cdot\delta_0=\frac{\delta}{2}.
\end{align}

Except for these iterations that quantum simulation is implemented to add perturbations, we still have $T/2$ iterations. We consider the iterations with large gradients $\norm{F(\x_t)}\geq\epsilon$. In each iteration, the function value will decrease at least $\eta\epsilon^2/4$. There can be at most $T/4$ iterations, for otherwise the function value decreases greater than $2B\geq F(\x^*)-F(\x_0)$. The failure probability is bounded by 
\begin{align}
\frac{4B}{\eta\epsilon^2}\cdot\delta_0=\delta/2.
\end{align}

In summary, with probability at least $1-\delta$, there are at most $T/2$ iterations when the quantum simulation is called and at most $T/4$ iterations when the gradient is large. There are thus at least $T/4$ iterations resulting in $\epsilon$-SOSP of $F$.

The number of queries can be decomposed into two parts, the number of queries required for gradient estimations, denoted by $T_1$, and the number of queries required for quantum simulations, denoted by $T_2$. For the first part, we have 
\begin{align}
T_1=O(T)=\tO\left(\frac{B}{\epsilon^2}\cdot\log d\right).
\end{align}
For $T_2$, the number of queries is given by \lem{CostQS} as
\begin{align}
T_2=\tO\left(\frac{B}{\epsilon^2}\cdot\log^2 d\right).
\end{align}
The total query complexity $T_1+T_2$ is bounded by
\begin{align}
\tO\left(\frac{B}{\epsilon^2}\cdot\log^2 d\right).
\end{align}
\end{proof}

\section{Existence of Example with Quantum Advantage using Quantum Tunnelling Walk}\label{append:QTW}
In \thm{ZeroNoLower}, we have proved that when the noise under \assume{ZeroProb} increases to $\nu\geq\Theta(\epsilon^{1.5})$, we can find a hard instance for any classical or quantum algorithm even using exponentially many queries. Although the noise bound for the quantum algorithms is the same, there might be quantum speedup for a specific instance. In this section, we provide a candidate for this argument under some proper additional assumptions.

We set the constant $\mu=300$. For the target function $F$, we still consider the following function defined in \eq{FDef}:
\begin{align}\label{eq:FDefApp}
F(\x):=h(\sin \x)+\norm{\sin\x}^2,
\end{align}
where $h(\x):=h_1(\v^\top\x)\cdot h_2\left(\sqrt{\norm{\x}^2-(\v^\top\x)^2}\right)$, and
\begin{align}
h_1(x)=g_1(\mu x),\quad g_1(x)&=(-16\abs{x}^5+48x^4-48\abs{x}^3+16x^2)\cdot\mathbb{I}\{\abs{x}<1\},\\
h_2(x)=g_2(\mu x),\quad g_2(x)&=(3x^4+8\abs{x}^3+6x^2-1)\cdot\mathbb{I}\{\abs{x}<1\}.
\end{align}

We adopt the construction in Ref.~\cite{liu2022quantum} for the construction of noisy function $f$. In the following, we denote $R$ to be the radius of the hyperball where the main construction is. We also choose $\v$ uniformly in the unit sphere. We define two regions $W_-=\mathbb{B}(0,a)$ and $W_+=\mathbb{2b\v,a}$ with $b\geq a$. We choose $a$ and $b$ such that $W_-$ and $W_+$ are in $\mathbb{B}(0,R)$. We denote the region $S_\v\coloneqq\{\x|\x\in\mathbb{B}(0,R),\abs{\x\cdot\v}\leq w\}$, where $w$ is chosen in $[0,\sqrt{3}w/2)$. We define
\begin{align}
B_\v\coloneqq\{\x|w<\x\cdot\v<2b-w,\sqrt{\norm{\x}^2-(\x\cdot\v)^2}<\sqrt{a^2-w^2},\x\notin W_-\cup W_+\}.
\end{align}
The construction of $f$ is given by
\begin{align}\label{eq:fDefinApp}
f=\begin{cases}
\frac{1}{2}\omega^2\norm{\x}^2,\quad\x\in W_-,\\
\frac{1}{2}\omega^2\norm{\x-2b\v}^2,\quad\x\in W_+,\\
H_1,\quad\x\in B_\v,\\
H_2,\quad\text{otherwise.}
\end{cases}
\end{align}
We choose $0<\omega^2a^2/2\sim H_1\ll H_2$. There are two local minima for $f$ in \eq{fDefinApp}, $\0$ and $2b\v$. We can verify that the function pair $(F,f)$ satisfies the following properties.
\begin{itemize}
    \item $\sup_{x}\norm{f-F}_\infty\leq\tO(1)$.
    \item $F$ is $O(d)$-bounded, $O(1)$-Hessian Lipschitz, and $O(1)$-gradient Lipschitz.
\end{itemize}

We apply the same scaling in the main text as
\begin{align}\label{eq:FDefScaleApp}
\tilde{F}(\x)&\coloneqq\epsilon rF\left(\frac{\x}{r}\right),\\\label{eq:fDefinScaleApp}
\tilde{f}(\x)&\coloneqq\epsilon rf\left(\frac{\x}{r}\right),
\end{align}
where $r=\sqrt{\epsilon/\rho}$. According to Ref.~\cite{liu2022quantum}, quantum tunneling walk can provide a speedup for finding SOSP of $f$ using ground states containing information of $\v$ and $W_+$.
\begin{lemma}[Proposition 4.2 and Theorem 4.1, Ref.~\cite{liu2022quantum}]\label{lem:LiuComp}
Assume we start from point $\0$ and we are provided with knowledge that $\0$ is a local minimum. We know local ground states associated with $W_-$ and $W_+$. By properly choosing the parameter $a$, $b$, $H_1$, and $H_2$, quantum tunneling walk can find the another local minima with high probability using $O(\poly(d))$ queries while any classical algorithm requires $\Omega(e^{dB})$ queries to zeroth-order oracle $f$.
\end{lemma}
Under the setting of this paper, we consider choosing $R$ such that $2b\v$ is also a local minimum of $F$. Notice that $\0$ is not a local minimum of $F$, our goal is to find the $\epsilon$-SOSP near the local minima $2b\v$ of $F$ taking queries to noisy oracle $f$ in \eq{QZeroOracle}. We can reach the following corollary using \lem{LiuComp}.
\begin{corollary}\label{cor:QTWSpeedup}
Consider the hard instance $(\tilde{F},\tilde{f})$ defined by \eq{FDefScaleApp}, \eq{fDefinScaleApp}, and a proper chosen $R$ such that $2b\v$ is also a local minimum of $\tilde{F}$. There exists a choice of parameters $a$, $b$, $H_1$, and $H_2$ such that a quantum algorithm starting at $\0$ can find an $\epsilon$-SOSP of $\tilde{F}$ with high probability using $O(\poly(d))$ queries to the noisy $\tilde{f}$ and proper initial ground state. However, any classical algorithm with proper initial ground states requires $\Omega(e^{dB})$ queries.
\end{corollary}
The above corollary demonstrates that if we assume that we have some ground states revealing information above $\v$, the quantum algorithm can provide an exponential speedup in solving a special hard instance $(F,f)$ that satisfies \assume{ZeroProb}. It is worthwhile to mention that the additional assumption on the local ground state is essential for this speedup and the quantum algorithm also requires query complexity that is exponential in $d$ without such assumption ~\cite{liu2022quantum}.
\end{document}